\documentclass{article}
\usepackage{fullpage}
\usepackage{natbib}
\usepackage{comment}
\usepackage{amsfonts}
\usepackage{amsmath}
\usepackage{amssymb}
\usepackage{amsthm}
\usepackage{booktabs}
\usepackage{calc}
\usepackage{colonequals}
\usepackage{graphicx}
\usepackage[hidelinks]{hyperref}
\usepackage{hyphenat} 
\usepackage[utf8]{inputenc}
\usepackage{listings}
\usepackage{mathpartir}
\usepackage{mathtools}
\usepackage{marginnote}
\usepackage{microtype}
\usepackage{newunicodechar}
\usepackage{stmaryrd} 
\usepackage{placeins}
\usepackage{tikz}
\usepackage{xcolor}
\usepackage{xspace}
\newtheorem{theorem}{Theorem}

\theoremstyle{definition}

\newcommand{\wpr}{where\hyp{}provenance\xspace}

\newcommand{\Lin}{\kw{Lin}}

\newcommand{\Links}{\textsf{Links}\xspace}
\newcommand{\WLinks}{\textsf{Links\textsuperscript{W}}\xspace}
\newcommand{\LLinks}{\textsf{Links\textsuperscript{L}}\xspace}
\newcommand{\kw}[1]{\text{\sf\bfseries#1}}
\newcommand{\IntTy}{\kw{Int}}
\newcommand{\StringTy}{\kw{String}}
\newcommand{\BoolTy}{\kw{Bool}}

\newcommand{\code}[1]{\text{\sf#1}}
\newcommand{\QueryType}{\mathsf{QType}}
\newcommand{\BaseRow}{\mathsf{BaseRow}}
\newcommand{\ProvSpec}{\mathsf{ProvSpec}}
\newcommand{\drarr}{\ensuremath{\mathrel{\text{\texttt{<--}}}}}
\newcommand{\rarr}{\ensuremath{\mathrel{\text{\texttt{<-}}}}}
\newcommand{\arr}{\ensuremath{\mathrel{\text{\texttt{->}}}}}
\newcommand{\plusplus}{\ensuremath{\mathrel{\text{\texttt{++}}}}}
\newcommand{\Prov}{\kw{Prov}}
\newcommand{\rcd}[1]{\code{(}#1\code{)}}
\newcommand{\coll}[1]{\code{[}#1\code{]}}
\newcommand{\braces}[1]{\code{\{}#1\code{\}}}
\newcommand{\erase}[1]{{\downharpoonleft} #1 {\downharpoonright}}

\definecolor{keyword}{RGB}{96,0,53}
\definecolor{darkgreen}{RGB}{0,128,0}
\definecolor{listingsFrame}{gray}{0.6}
\definecolor{string}{RGB}{0,128,0}
\definecolor{string}{RGB}{0,0,128}

\lstdefinestyle{normal}{%
  xleftmargin=\parindent,
}

\lstdefinestyle{nospace}{%
  aboveskip=0em,
  belowskip=0em,
  xleftmargin=0em,
  xrightmargin=0em,
}

\lstdefinestyle{nohspace}{%
  xleftmargin=0em,
  xrightmargin=0em,
}

\lstdefinestyle{outfigure}{%
  xrightmargin=0em,
  basicstyle=\scriptsize\sf,
}

\lstdefinestyle{infigure}{%
  aboveskip=0em,
  belowskip=0em,
  xleftmargin=0em,
  xrightmargin=0em,
  basicstyle=\small\sf,
}

\lstdefinestyle{scriptsize}{%
  basicstyle=\scriptsize\sf,
}

\lstdefinestyle{footnote}{%
  basicstyle=\footnotesize\sf,
}

\lstdefinestyle{figure}{%
  xleftmargin=0pt,
  frame=single,
  rulecolor=\color{listingsFrame},
}

\newcommand{\mytilde}{%
  \texttt{\resizebox{.48em}{1ex}{\hbox{$\sim$}}}%
}

\lstdefinelanguage{Links}{%
  keywords={delete,fun,if,else,for,from,sig,server,where,query,table,with,var,typename},
  mathescape=true,
  comment=[l]{\#},
  morestring=[b]/,
  morestring=[b]",
  sensitive=true,
  escapechar=\%,
}

\lstdefinelanguage{LaTeX}[]{Tex}{%
  language=Tex,
  morekeywords={emph, label, ref, begin, end, part, chapter, section,
                subsection, subsubsection, paragraph, subparagraph, cite},
}

\lstdefinelanguage{SugarJ}[]{Java}{%
  language=Java,
  morekeywords={sugar, context, free, syntax, desugarings, sorts, signature,
    constructors, rules, strategies, assert, as, editor, services, colorer, folding,
    outliner, checks, recursive, errors, warnings, where,
    css, outlining, rec, color, completions, completion, template,
    layout, analyses, then, extension},
  mathescape=true,
  deletestring=[b]',
  morecomment=[l]{//},
  escapechar=\%,
}

\lstdefinelanguage{SugarJXML}[]{SugarJ}{%
  morekeywords={xmlschema},
  moredelim=*[s][\color{blue}]{<}{>},
  emph={$}, 
  emphstyle={\bfseries\color{keyword}},
  deletecomment=[l]{//},
  morecomment=[s]{<!--}{-->}
}

\lstdefinelanguage{SugarMDD}[]{SugarJ}{%
  morekeywords={model, transformation},
}

\lstdefinelanguage{SugarJATM}[]{SugarMDD}{%
  morekeywords={statemachine, initial, state, events},
}

\lstdefinelanguage{SugarJEntity}[]{SugarMDD}{%
  morekeywords={entity},
}
\lstdefinelanguage{SugarJATMEntity}[]{SugarJATM}{%
  morekeywords={data, entity},
}

\lstdefinelanguage{SugarJTemplate}[]{SugarMDD}{%
  morekeywords={in, \$for, template},
  mathescape=false
}

\lstdefinelanguage{SugarJFeature}[]{SugarMDD}{%
  morekeywords={featuremodel, features, constraint, config, enable, disable, variable},
  otherkeywords={\#ifdef},
}

\lstdefinelanguage{scala}{
  morekeywords={abstract,case,catch,class,def,%
    do,else,extends,false,final,finally,%
    for,if,implicit,import,match,mixin,%
    new,null,object,override,package,%
    private,protected,requires,return,sealed,%
    super,this,throw,trait,true,try,%
    type,val,var,while,with,yield},
  sensitive=true,
  morecomment=[l]{//},
  morecomment=[n]{/*}{*/},
  morestring=[b]",
  morestring=[b]',
  morestring=[b]"""
}

\lstdefinelanguage{SDF}{%
  morekeywords={context, free, syntax, sorts, signature,
    constructors, rules, strategies},
  escapechar=_,
  mathescape=true,
  comment=[l]{\%\%},
  morestring=[b]",
}

\lstdefinelanguage{MyPython}[]{Python}{%
  escapechar=\%,
  mathescape=true,
}

\lstdefinelanguage{MyHaskell}[]{Haskell}{%
  escapechar=\%,
  mathescape=true,
  deletekeywords={Nothing,Just,False,True,putStrLn,fail,fromJust,lookup,Num,exp,free,snd,String,
  return,error,otherwise,not,show,read,Eval,Read,readsPrec,print},
}

\lstdefinelanguage{SugarHaskell}[]{MyHaskell}{%
  morekeywords={context, free, syntax, desugarings, sorts, signature,
    constructors, rules, strategies, lexical, reject},
  mathescape=false,
}

\lstdefinelanguage{SugarHaskellArrows}[]{SugarHaskell}{%
  morekeywords={proc},
}

\lstdefinelanguage{EBNF}{
  morestring=[b]"
}

\lstdefinelanguage{Constraint}{
}

\lstdefinelanguage{Plain}{}

\lstdefinelanguage{Questionnaire}[]{Java}{%
  morekeywords = {questionnaire, question, value, Boolean, String, Integer, group,
    if, else, define, ask}
}

\lstdefinelanguage{SugarFomega}{
  keywords = {module, val, type, mu, if, then, else, case, of,
    fold, unfold, true, false, as, public, import, syntax, desugaring, typing, context, free, let, in, end, forall, do},
  mathescape = true,
  morestring=[b]",
}

\lstdefinelanguage{WLinks}[]{Links}{%
  language=Links,
  morekeywords={data,prov,default}
}
\lstdefinelanguage{LLinks}[]{Links}{%
  language=Links,
  morekeywords={lineage}
}
\newlength{\olinewidth}
\newlength{\ocharwidth}

\newcommand{\concat}{\ensuremath{+\!\!\!\!+\,}}

\newcommand{\DD}[1]{\ensuremath{\mathfrak{D}\sem{#1}}}

\newcommand{\LL}[1]{\ensuremath{\mathfrak{L}\sem{#1}}}
\newcommand{\LLstar}[1]{\ensuremath{\mathfrak{L}^*\sem{#1}}}
\newcommand{\sem}[1]{\llbracket #1 \rrbracket}

\newunicodechar{⟦}{\llbracket}
\newunicodechar{⟧}{\rrbracket}
\newunicodechar{↦}{\mapsto}
\newunicodechar{𝔇}{\ensuremath{\mathfrak{D}}}
\newunicodechar{𝔏}{\ensuremath{\mathfrak{L}}}
\newunicodechar{𝔚}{\ensuremath{\mathfrak{W}}}
\newunicodechar{ℰ}{\ensuremath{\mathcal{E}}}

\newunicodechar{Γ}{\ensuremath{\Gamma}}
\newunicodechar{⊢}{\ensuremath{\vdash}}
\newunicodechar{τ}{\ensuremath{\tau}}
\newunicodechar{⟶}{\ensuremath{\longrightarrow}}
\newunicodechar{∧}{\ensuremath{\wedge}}

\newtheorem{lemma}[theorem]{Lemma}
\newtheorem{corollary}[theorem]{Corollary}

\begin{document}


\title{Language-integrated Provenance}

\author{Stefan Fehrenbach\\
University of Edinburgh\\10 Crichton Street\\ Edinburgh, EH8 9AB \\United Kingdom
\and James Cheney\\
University of Edinburgh\\10 Crichton Street\\ Edinburgh, EH8 9AB\\ United Kingdom}
\date{}
\maketitle

\begin{abstract}
  Provenance, or information about the origin or derivation of data,
  is important for assessing the trustworthiness of data and
  identifying and correcting mistakes.  Most prior implementations of
  data provenance have involved heavyweight modifications to database
  systems and little attention has been paid to how the provenance
  data can be used outside such a system.  We present extensions to
  the Links programming language that build on its support for
  language-integrated query to support provenance queries by rewriting
  and normalizing monadic comprehensions and extending the type system
  to distinguish provenance metadata from normal data.  The main
  contribution of this article is to show that the two most common forms
  of provenance can be implemented efficiently and used safely as a
  programming language feature with no changes to the database system.
\end{abstract}

\section{Introduction}

A Web application typically spans at least three different
computational models: the server-side program, browser-side HTML or
JavaScript, and SQL to execute on the database.  Coordinating these
layers is a considerable challenge.  Recently, programming languages
such as \Links~\cite{FMCO2006CooperLWY}, Hop~\cite{hop} and
Ur/Web~\cite{POPL2015Chlipala} have pioneered a \emph{cross-tier}
approach to Web programming.  The programmer writes a single program,
which can be type-checked and analyzed in its own right, but parts of
it are executed to run efficiently on the multi-tier Web architecture
by translation to HTML, JavaScript and SQL. Cross-tier Web programming
builds on \emph{language-integrated
  query}~\cite{meijer06sigmod,pialorsi2007introducing}, a technique
for safely embedding database queries into programming languages,
which has been popularized by Microsoft’s LINQ library, which provides
language-integrated query for .NET languages such as C\# and F\#.
(The language \Links was developed concurrently with Meijer et al.'s
work on LINQ; their names are coincidentally similar but they are
different systems.)


When something goes wrong in a database-backed Web application,
understanding what has gone wrong and how to fix it is also a
challenge.  Often, the database is the primary ``state'' of the
program, and problems arise when this state becomes inconsistent or
contains erroneous data.  For example, Figure~\ref{fig:boat-tours} shows \Links code for querying
data from a (fictional) Scottish tourism database, with the result
shown in Figure~\ref{fig:boat-tours-result}.  Suppose one of the phone
numbers is incorrect: we might want to know \emph{where} in the source
database to find the source of this incorrect data, so that we can
correct it.  Alternatively, suppose we are curious \emph{why} some
data is produced: for example, the result shows \lstinline!EdinTours!
twice.  If we were not expecting these results, e.g.\ because we believe that
\lstinline!EdinTours! is a bus tour agency and does not offer boat
tours, then we need to see additional input data to understand why they
were produced.

\begin{figure}[tb]
\small
\begin{lstlisting}[language=Links]
var agencies = table "Agencies"
 with (name:String, based_in:String, phone:String)
 from db;
var externalTours = table "ExternalTours"
 with (name:String, destination:String, type:String, price:Int)
 from db;
var q1 = query {
  for (a <-- agencies)
    for (e <-- externalTours)
    where (a.name == e.name && e.type == "boat")
      [(name = e.name,
        phone = a.phone)]
}
\end{lstlisting}
\caption{\Links table declarations and example query}
\label{fig:boat-tours}

\begin{center}
  \begin{tabular}{ll}
    name & phone\\
    \midrule
    EdinTours & 412 1200\\
    EdinTours& 412 1200\\
    Burns's & 607 3000
  \end{tabular}
\end{center}
\caption{Example query results}
\label{fig:boat-tours-result}
\end{figure}

Automatic techniques for producing such explanations, often called
\emph{provenance}, have been explored extensively in the database
literature~\cite{TODS2000CuiWW,ICDT2001BunemanKT,green07pods}.
Neither conventional nor cross-tier Web programming currently provides
direct support for provenance.  A number of implementation strategies
for efficiently computing provenance for query results have been
explored~\cite{VLDB2005BhagwatCTV,ICDE2009GlavicA,Festschrift2013GlavicMG},
but no prior work considers the interaction of provenance with
clients of the database.


We propose \emph{language-integrated provenance}, a new approach to
implementing provenance that leverages the benefits of
language-integrated query.  In this article, we present two instances of
this approach, one which computes \emph{where-provenance} showing
where in the underlying database a result was copied from, and another
which computes \emph{lineage} showing all of the parts of the
database that were needed to compute part of the result.
Both techniques are implemented by a straightforward source-to-source
translation which adjusts the types of query expressions to
incorporate provenance information and changes the query behavior to
generate and propagate this information.  Our approach is implemented
in \Links, and benefits from its strong support for rewriting queries
to efficient SQL equivalents, but the underlying ideas may be
applicable to other languages that support language-integrated query,
such as F\#~\cite{ML2006Syme}, SML\#~\cite{ICFP2011OhoriU}, or
Ur/Web~\cite{POPL2015Chlipala}.

Most prior implementations of provenance involve changes to relational
database systems and extensions to the SQL query language, departing from
the SQL standard that relational databases implement.  To date, none
of these proposals have been incorporated into the SQL standard or
supported by mainstream database systems.  If such extensions are
adopted in the future, however, we can simply generate queries that
use these extensions in \Links.  In some of these systems, enabling
provenance in a query changes the result type of the query (adding an
unpredictable number of columns).  Our approach is the first (to the
best of our knowledge) to provide type-system support that makes sure
that the extra information provided by language-integrated provenance
queries is used safely by the client.

Our approach builds on \Links's support for queries that construct
nested collections~\cite{SIGMOD2014CheneyLW}.  This capability is
crucial for lineage, because the lineage of an output record is a \emph{set} of relevant input
records.  Moreover, our provenance translations can be used with queries
that construct nested results.  Our approach is also distinctive in
allowing fine-grained control over where-provenance.  In particular, the
programmer can decide whether to enable or disable where-provenance
tracking for individual input table fields, and whether to
keep or discard provenance for each result field.

We present two simple extensions to \Links to support where-provenance
and lineage, and give (provably type-preserving) translations from
both extensions to plain \Links.  We have implemented both approaches
and experimentally validated them using a synthetic benchmark.
Provenance typically slows down query evaluation because more data is
manipulated.  For where-provenance, our experiments indicate a
constant factor overhead of 1.5--2.8.  For lineage, the slowdown is
between 1.25 and 7.55, in part because evaluating lineage queries
usually requires manipulating more data.  We also compare \Links to
Perm~\cite{ICDE2009GlavicA}, a database-integrated provenance system,
whose authors report slowdowns of 3--30 for a comparable form of
lineage.  In our experiments Perm generally outperforms \Links but
\Links is within an order of magnitude.

\paragraph{Contributions and outline}
Section~\ref{sec:overview} gives a high-level overview of our
approach, illustrated via examples.  Section~\ref{sec:background}
reviews background material on \Links upon which we rely.
This article makes the following three contributions:
\begin{itemize}
\item Definition of the \WLinks and \LLinks extensions to \Links, along
  with their semantics and provenance correctness properties
  (Section~\ref{sec:design})
\item Implementations of \WLinks and \LLinks by type-preserving translation to plain
  \Links (Section~\ref{sec:translation})
\item Experimental evaluation of the implementations on a number of
  queries (Section~\ref{sec:evaluation})
\end{itemize}
Related work is discussed in greater detail in Section~\ref{sec:related}.

This article significantly extends an earlier conference
paper~\cite{PPDP2016FehrenbachC}.  The conference version presented the where-provenance
and lineage translations and their implementation and evaluation; this
article in addition describes the semantics of \Links
(Section~\ref{sec:background}), 
and proves correctness and type-preservation properties that were not
included in the conference paper (Sections~\ref{sec:design} and \ref{sec:translation}).

\section{Overview}
\label{sec:overview}

In this section we give an overview of our approach, first reviewing
necessary background on \Links and language-integrated query based on
comprehensions, and then showing how provenance can be supported by
query rewriting in this framework. We will use a running example of a
simple tours database, with some example data shown  in Figure~\ref{fig:example-data}.

\subsection{Language-integrated query}

Writing programs that interact with databases can be tricky, because
of mismatches between the models of computation and data structures
used in databases and those used in conventional programming
languages.  The default solution (employed by JDBC and other typical
database interface libraries) is for the programmer to write queries
or other database commands as uninterpreted strings in the host
language, and these are sent to the database to be executed.  This
means that the types and names of fields in the query cannot be
checked at compile time and any errors will only be discovered as a
result of a run-time crash or exception.  More insidiously, failure to
adequately sanitize user-provided parameters in queries opens the door to
SQL injection attacks~\cite{IEEEComputer2013SharTan}.

Language-integrated query is a technique for embedding queries into
the host programming language so that their types can be checked
statically and parameters are automatically sanitized.  Broadly, there are
two common approaches to language-integrated query.  The first
approach, which we call \emph{SQL embedding}, adds
specialized constructs resembling SQL queries to the host
language, so that they can be typechecked and handled correctly by the
program.  This is the approach taken in C\#~\cite{meijer06sigmod,pialorsi2007introducing},
SML\#~\cite{ICFP2011OhoriU}, and Ur/Web~\cite{POPL2015Chlipala}.   The second approach, which we call \emph{comprehension}, uses monadic
comprehensions or related constructs of the host language, and
generates queries from such expressions.  The comprehension approach
builds on foundations for querying databases using comprehensions
developed by \citet{TCS1995BunemanNTW}, and has been adopted in
languages such as F\#~\cite{ML2006Syme} and
\Links~\cite{FMCO2006CooperLWY} as well as libraries such as
Database-Supported Haskell~\cite{IFL2010GiorgidzeGSW}.

\begin{figure}[tb]
\small
\textbf{Agencies}\\[.3em]
\begin{tabular}{rlll}
(oid) & name & based\_in & phone\\
\cmidrule{2-4}
1 & EdinTours & Edinburgh & 412 1200\\
2 & Burns's & Glasgow & 607 3000\\
\end{tabular}
\\[.8em]

\textbf{ExternalTours}\\[.3em]
\begin{tabular}{rlllr}
(oid) & name & destination & type & price in \pounds\\
\cmidrule{2-5}
3 & EdinTours & Edinburgh & bus & 20\\
4 &  EdinTours & Loch Ness & bus & 50\\
5 & EdinTours &Loch Ness & boat & 200\\
6 & EdinTours & Firth of Forth & boat & 50\\
7 & Burns's & Islay & boat & 100\\
8 & Burns's & Mallaig & train & 40
\end{tabular}
\caption{Example input data}\label{fig:example-data}
\end{figure}

The advantage of the comprehension approach is that it provides a
higher level of abstraction for programmers to write queries, without
sacrificing performance.  This advantage is critical to our work, so
we will explain it in some detail.
For example, the query shown in
Figure~\ref{fig:boat-tours} illustrates \Links comprehension syntax.
It asks for the names and phone numbers of all agencies having an
external tour of type \lstinline!"boat"!. The keyword \lstinline!for!
performs a comprehension over a table (or other collection), and the
\lstinline!where! keyword imposes a Boolean condition filtering the
results.  The result of each iteration of the comprehension is a
singleton collection containing the record %
\lstinline!(name = e.name,phone = a.phone)!.
%

Monadic comprehensions do not always correspond exactly to SQL
queries, but for queries that map flat database tables to flat
results, it is possible to normalize these comprehension expressions
to a form that is easily translatable to SQL~\cite{JCSS1996Wong}.  For
example, the following query
\begin{lstlisting}[language=Links]
var q1' = query {
  for (e <-- externalTours)
  where (e.type == "boat")
    for (a <-- agencies)
    where (a.name == e.name)
    [(name = e.name, phone = a.phone)]
}
\end{lstlisting}
does not directly correspond to a SQL query due to the alternation of
\lstinline!for! and \lstinline!where! operations; nevertheless, query
normalization generates a single equivalent SQL query in
which the \lstinline!where! conditions are both pushed into the SQL
query's \lstinline!WHERE! clause:
\begin{verbatim}
SELECT e.name AS name, a.phone AS phone
  FROM ExternalTours e, Agencies a
 WHERE e.type = 'boat' AND a.name = e.name
\end{verbatim}
Normalization frees the programmer to write
queries in more natural ways, rather than having to fit the
query into a pre-defined template expected by SQL.

However, this freedom can also lead to problems, for example if the
programmer writes a query-like expression that contains an operation,
such as \lstinline{print} or regular expression matching, that cannot
be performed on the database.  In early versions of \Links, this could
lead to unpredictable performance, because queries would unexpectedly
be executed on the server instead of inside the database.  The current version uses a
type-and-effect system (as described by \citet{DBPL2009Cooper} and
\citet{TLDI2012LindleyC}) to track which parts of the program must be
executed in the host language and which parts may be executed on the
database.  Using the \lstinline!query! keyword above forces the
typechecker to check that the code inside the braces will successfully
execute on the database.
\subsection{Higher-order functions and nested query results}

Although comprehension-based language-integrated query may seem (at
first glance) to be little more than a notational convenience, it has
since been extended to provide even greater flexibility to programmers
without sacrificing performance.

The original results on normalization (due to \citet{JCSS1996Wong})
handle queries over flat input tables and producing flat result
tables, and did not allow calling user-defined functions inside
queries.  Subsequent work has shown how to support higher-order
functions~\cite{DBPL2009Cooper,DBPL2013GrustU} and queries that
construct nested collections~\cite{SIGMOD2014CheneyLW}.  For example,
we can use functions to factor the previous query into reusable
components, provided the functions are nonrecursive and only perform
operations that are allowed in the database.

\begin{lstlisting}[language=Links]
fun matchingAgencies(name) {
  for (a <-- agencies)
  where (a.name == name)
    [(name = e.name, phone = a.phone)]
}
var q1'' = query {
  for (e <-- externalTours)
  where (e.type == "boat")
    matchingAgencies(e.name)
}
\end{lstlisting}

Cooper's results show that these queries still normalize to
SQL-equivalent queries, and this algorithm is implemented in \Links.
Similarly, we can write queries whose result type is an arbitrary
combination of record and collection types, not just a flat
collection of records of base types as supported by SQL:

\begin{lstlisting}[language=Links]
var q2 = query {
  for (a <-- agencies)
    [(name = a.name,
       tours = for (e <-- externalTours)
                 where (e.name == a.name)
                 [(dest = e.destination, type = e.type)]
}
\end{lstlisting}
This query produces records whose second \lstinline!tours! component
is itself a collection --- that is, the query result is of the type
\lstinline![(name:String,[(dest:String, type:Type)])]! which contains
a nested occurrence of the collection type constructor \lstinline![]!.
SQL does not directly support queries producing such nested results
--- it requires flat inputs and query results.



Our previous work on \emph{query shredding} \cite{SIGMOD2014CheneyLW}
gives an algorithm that evaluates queries with nested results
efficiently by translation to SQL.  Given a query whose return type
contains $n$ occurrences of the collection type constructor, query
shredding generates $n$ SQL queries that can be evaluated on the
database, and constructs the nested result from the resulting tables.
This is typically much more efficient than loading the database data
into memory and evaluating the query there.  \Links supports query
shredding and we will use it in this article to implement lineage.

Both capabilities, higher-order functions and nested query results, are
essential building blocks for our approach to provenance.  In what
follows, we will use these techniques without further explanation of
their implementation.  The details are covered in previous
papers~\cite{DBPL2009Cooper,TLDI2012LindleyC,SIGMOD2014CheneyLW}, but
are not needed to understand our approach.

\subsection{Where-provenance and lineage}


As explained in the introduction, provenance tracking for queries has been
explored extensively in the database community.  We are
now in a position to explain how these provenance techniques can be
implemented on top of language-integrated query in \Links.  We
review two of the most common forms of provenance, and illustrate our
approach using examples; the rest of the article will use similar
examples to illustrate our implementation approach.

\textbf{Where-provenance} is information about where information in
the query result ``came from'' (or was copied from) in the input.
\citet{ICDT2001BunemanKT} introduced this idea; our approach is based
on a later presentation for the nested relational calculus by
\citet{TODS2008BunemanCV}.  A common reason for asking for
where-provenance is to identify the source of incorrect (or
surprising) data in a query result.  For example, if a phone number in
the result of the example query is incorrect, we might ask for its
where-provenance.  In our system, this involves modifying the input
table declaration and query as follows:
\begin{lstlisting}[language=WLinks]
var agencies = table "Agencies"
 with (name:String, based_in:String, phone:String)
 where phone prov default
\end{lstlisting}
The annotation \lstinline[language=WLinks]!where phone prov default! says to assign phone
numbers the ``default'' provenance annotation of the form
\lstinline!(Agencies, phone, i)! where $i$ is the object id (oid) of
the corresponding row.  The field value will be of type
$\Prov(\StringTy)$; the data value can be accessed using the keyword
$\kw{data}$ and the provenance can be accessed using the keyword
$\kw{prov}$, as follows:
\begin{figure}[h!]
\begin{lstlisting}[language=WLinks]
var q1''' = query {
  for (a <-- agencies)
    for (e <-- externalTours)
    where (a.name == e.name && e.type == "boat")
      [(name = e.name,
        phone = data a.phone, p_phone = prov a.phone)]
}
\end{lstlisting}
  \caption{\WLinks query $q1'''$.}
  \label{fig:q1p3}
\end{figure}

The result of this query is as follows:
\begin{center}
\begin{tabular}{lll}
name & phone & p$\_$phone\\
\midrule
EdinTours & 412 1200 &\lstinline!(Agencies,phone,1)!\\
EdinTours & 412 1200 & \lstinline!(Agencies,phone,1)!\\
Burns's & 607 3000 &\lstinline!(Agencies,phone,2)!\\
\end{tabular}
\end{center}

We would like to emphasize one important point about our approach to
where-provenance: as illustrated by the above query, we need to change
the table definitions to indicate which fields carry provenance, and
we also need to annotate the query to indicate where the data or
provenance are used.  This effort is reasonable because queries are
typically small, but alternative strategies, such as automatically
annotating all fields, could also be considered.

\textbf{Why-provenance} is information that explains ``why'' a result
was produced.  In a database query setting, this is usually taken to
mean a \emph{justification} or \emph{witness} to the query result,
that is, a subset of the input records that includes all of the data
needed to generate the result record.  Actually, several related forms
of why-provenance have been studied~\cite{TODS2000CuiWW,ICDT2001BunemanKT,FTDB2009CheneyCT,Festschrift2013GlavicMG},
however, many of these only make sense for set-valued collections,
whereas \Links currently supports multiset semantics.  In this
article, we focus on a simple form of why-provenance called
\emph{lineage} which is applicable to either semantics.

Intuitively, the lineage of a record $r$ in the result of a query is a
subset $L$ of the records in the underlying database $db$ that
``justifies'' or ``witnesses'' the fact that $r$ is in the result of
$Q$ on $db$.  That is, running $Q$ on the lineage $L$ should produce a
result containing $r$, i.e.\ $r \in Q(L)$.  Obviously, this property
can be satisfied by many subsets of the input database, including the
whole database $db$, and this is part of the reason why there exist
several different definitions of why-provenance (for example, to
require minimality).
We follow the common approach of defining the lineage to be
the set of all input database records accessed in the process of
producing $r$; this is a safe overapproximation to the minimal
lineage, and usually is much smaller than the whole database.

We identify records in input database tables using pairs such as
\lstinline!(Agencies,2)! where the first component is the table
name and the second is the row id, and the lineage of an element of a
collection is just a collection of such pairs.  (Again, this has the
benefit that we can use a single type for references to data in
multiple input tables.)  Using this representation, the lineage
for \lstinline!q1! (Figure~\ref{fig:boat-tours}) is as follows:

\begin{center}
\begin{tabular}{ll||l}
name & phone &lineage\\
\midrule
EdinTours & 412 1200 & \lstinline![(Agencies,1),(ExternalTours,5)]!\\
EdinTours & 412 1200 & \lstinline![(Agencies,1),(ExternalTours,6)]!\\
Burns's & 607 3000 & \lstinline![(Agencies,2),(ExternalTours,7)]!\\
\end{tabular}
\end{center}

In our system, to obtain these results we simply use the keyword
\lstinline!lineage! instead of \lstinline!query!; for example, for
\lstinline!q1! we would simply write:
\begin{lstlisting}
lineage {
  for (a <-- agencies)
    for (e <-- externalTours)
    where (a.name == e.name && e.type == "boat")
      [(name = e.name,
        phone = a.phone)]
}
\end{lstlisting}

\Links's capabilities for normalizing and efficiently evaluating
queries provide the key ingredients needed for computing provenance.
For both where-provenance and lineage, we can translate programs using
the extensions described above, in a way that both preserves types and
ensures that the resulting query expressions can be converted to SQL
queries.  In the rest of this article, we give the details of these
translations and present an experimental evaluation showing that its
performance is reasonable.

\subsection{Pragmatics and limitations}

Most research on provenance in databases has focused on the process of
propagating annotations (e.g.\ source locations) through queries to the
output. 
This article is the first to consider support for provenance at the
programming language level.  Our attempt to do so has raised some
interesting issues that have not been considered in this previous
work, such as:
\begin{enumerate}
\item Where do the initial provenance annotations come from?
\item What are appropriate correctness criteria in a setting where the underlying
  program may be updated (by the program or other database users)?
\item Should we also track provenance information for updates, and if so how?
\end{enumerate}
In our approach, we require table declarations to be annotated to
indicate how the table's data is annotated with provenance.  Thus, we
do not assume that the underlying relational database schema contains
provenance data, but if such data is available we can use it.
However, as we shall see, this complicates matters since we need to be
able to handle updates to such tables.  We deal with this by
translating table references to pairs, with the first component
containing the raw table reference for use in updates and the second
containing a delayed query expression that produces the initial
annotated version of the table for use in queries.

Concerning the second question, we revisit correctness
criteria for where-provenance and lineage that have been considered in
previous work, and show that similar properties hold for our
approach.  However, as in previous work, our correctness properties
assume that the underlying database is unchanging.  This is of course
not a realistic assumption: \Links includes update operations that can
change the database tables, and other database users might
concurrently update the data or even change the structure of the
data.  It is an interesting question (beyond the scope of this paper)
how to generalize existing criteria for provenance correctness to this setting.

We mention two additional limitations.  First, since Links itself does
not yet support grouping and aggregation in queries, our approach does
not attempt to handle these features either.  This is an important
obstacle to be overcome in future work.  Likewise, we do not consider
the process of tracking provenance for updates to the database, even
when the updates are performed by \Links.  This has been considered by
Buneman et al.~\cite{TODS2008BunemanCV}, but in this paper we focus on
provenance tracking for queries and leave (language-integrated)
provenance tracking for updates for future work.

\section{Links background}
\label{sec:background}
We first review a subset of the \Links programming language that
includes all of the features relevant to our work; we omit some
features (such as effect typing, polymorphism, and concurrency) that
are not required for the rest of the article.  We also present a
simplified operational semantics for \Links, omitting detail regarding
query normalization and shredding that is presented in
more detail in previous work~\cite{TLDI2012LindleyC,SIGMOD2014CheneyLW}.
\ref{app:notation} lists notations introduced in this
  paper, with a brief explanation and reference to their first occurrence.

Figure~\ref{fig:links-syntax} presents a simplified subset of \Links
syntax, sufficient for explaining the provenance translations in this
article. Types include base types $O$ (such as integers, booleans and
strings), table types \lstinline!table($l_i$: $A_i$)$_{i=1}^n$!,
function types \lstinline!$A$ -> $B$!, record types
\lstinline!($l_i$: $A_i$)$_{i=1}^n$!, and collection types
\lstinline![$A$]!. In \Links, collection types are treated as
multisets inside database queries (reflecting SQL's default multiset
semantics), but represented as lists during ordinary execution.

\begin{figure}[tb]
\[\small
\begin{array}{lrcl}
\text{Base types} & O & \Coloneqq &\IntTy \mid \BoolTy \mid \StringTy\smallskip\\
\text{Rows} & R & \Coloneqq&\cdot \mid R,l:A\smallskip\\
\text{Table types} & T &\Coloneqq& \kw{table}(R)\smallskip\\
\text{Types} & A, B &\Coloneqq &O \mid T \mid A ~\lstinline!->!~ B
                                 \mid \rcd{R} \mid \coll{A}\smallskip\\
\text{Contexts} & \Gamma & \Coloneqq & \cdot \mid \Gamma,x:A\medskip\\
\text{Expressions} &L, M, N&\Coloneqq & c \mid x \mid  (l_i = M_i)_{i=1}^n \mid N.l\smallskip\\
&           &  \mid  &  \kw{fun}~f(x_i|_{i=0}^n) ~{ N } \mid N(M_i|_{i=0}^n)\smallskip\\
&            & \mid  &  \kw{var} ~x = M; N\mid \kw{if}~ (L)~ \{M\}~ \kw{else}~ \{N\}\smallskip\\
&           &  \mid  &  \kw{query}~ \{ N \} \mid \kw{table}~ \mathit{name}~ \kw{with}~ (l_i: O_i)_{i=1}^n\smallskip\\
&           &  \mid  &  \coll{} \mid \coll{N} \mid N
                       ~\plusplus~ M \mid \kw{empty}(M)\smallskip\\
& &  \mid  &  \kw{for} ~(x ~\lstinline!<-! ~L) ~M\mid ~\kw{where} (M)~
             N \smallskip\\
&           &  \mid  &  \kw{for} ~(x ~\lstinline!<--! ~L) ~M \mid   \kw{insert} ~L~\kw{values}~M\smallskip\\
& &  \mid  &  \kw{update} ~(x \drarr L)~\kw{where}~M~\kw{set}~N\smallskip\\
& &  \mid  &  \kw{delete} ~(x \drarr L)~\kw{where}~M
\end{array}
\]
\caption{Syntax of a subset of \Links.}
\label{fig:links-syntax}
\end{figure}

\begin{figure}
  \small
  \begin{align*}
    \Sigma, (\kw{fun}\,f(x_i|_{i=0}^n)\,M)(V_i|_{i=0}^n) & ⟶ \Sigma, M[f \coloneqq \kw{fun}\,f(x_i)\,M, x_i \coloneqq V_i]\\
    \Sigma, \kw{var}\, x = V; M & ⟶ \Sigma, M[x \coloneqq V] \\
    \Sigma, (l_i = V_i)_{i=1}^n.l_k & ⟶ \Sigma, V_k \\
    \Sigma, \kw{if}\,(\kw{true})\,M\,\kw{else}\,N & ⟶ \Sigma, M \\
    \Sigma, \kw{if}\,(\kw{false})\,M\,\kw{else}\,N & ⟶ \Sigma, N \\
    \Sigma, \kw{query}\,M & ⟶ \Sigma, M \\
    \Sigma, \kw{empty}(\texttt{[]}) & ⟶ \Sigma, \kw{true} \\
    \Sigma, \kw{empty}(V) & ⟶ \Sigma, \kw{false} \quad \text{iff } V \neq \texttt{[]} \\
    \Sigma, \kw{for}\,(x\,\texttt{<-}\,\texttt{[]})\,M & ⟶ \Sigma, \texttt{[]} \\
    \Sigma, \kw{for}\,(x\,\texttt{<-}\,[V])\,M & ⟶ \Sigma, M[x \coloneqq V]\\
    \Sigma, \kw{for}\,(x\,\texttt{<-}\,V \concat W)\,M& ⟶ \Sigma, (\kw{for}\,(x\,\texttt{<-}\, V)\,M) \concat (\kw{for}\,(x\,\texttt{<-}\, W)\,M)\\
    \Sigma, \kw{for}\,(x\,\texttt{<--}\, \kw{table}\,n)\, M & ⟶ \Sigma, \kw{for}\,(x\,\texttt{<-}\, \Sigma(n))\, M \\
    \Sigma, \kw{insert}\,(\kw{table}\,t)\,\kw{values}\,V & ⟶ \Sigma[t ↦ \Sigma(t) \concat V], ()
  \end{align*}%
  \begin{mathpar}
    \infer[]
    {\Sigma' = \Sigma[t \mapsto [X \in \Sigma(t) | \Sigma, M[x := X] ⟶^* \Sigma, \kw{false}]]}
    {\Sigma, \kw{delete}\,(x\,\texttt{<--}\,\kw{table}\,t)\,\kw{where}\,M ⟶ \Sigma', ()}

    \infer[]
    {\Sigma' = \Sigma[t \mapsto [u(X)
           | X \in \Sigma(t)]] \quad u(X) = {\begin{cases}
               (X\,\kw{with}\,l_i = V_i) & \text{if } M[x \coloneqq X] ⟶^* \kw{true}\\
               & \text{ and } N_i[x \coloneqq X]⟶^* V_i \\
            X & \text{otherwise}
          \end{cases}}}
    {\Sigma, \kw{update}\,(x\,\texttt{<--}\,\kw{table}\,t)\,\kw{where}\,M\,\kw{set}\,(l_i = N_i)_{i=1}^n ⟶ \Sigma', ()}
  \end{mathpar}
  
  \[
    \infer[]
    {\Sigma,M ⟶ \Sigma',M'}
    {\Sigma,\mathcal{E}[M] ⟶ \Sigma',\mathcal{E}[M']}
  \]

  \[
    \begin{array}{rcl}
      \mathcal{E} & \Coloneqq & [] \mid \mathcal{E}(M_1, \dots, M_n) \mid V(V_1, \dots, V_{i-1}, \mathcal{E}, M_{i+1}, \dots, M_n) \\
                  & \mid & (l_1 = V_1, \dots, l_{i-1} = V_{i-1}, l_i = \mathcal{E}, l_{i+1} = M_{i+1}, \dots, l_n = M_n) \mid \mathcal{E}.l \\
                  & \mid & \kw{if}\,(\mathcal{E})\, M\,\kw{else}\,N \\
                  & \mid & \kw{empty}(\mathcal{E}) \\
                  & \mid & [\mathcal{E}] \mid \mathcal{E} \concat M \mid V \concat \mathcal{E} \\
                  & \mid & \kw{for}\,(x\,\texttt{<-}\,\mathcal{E})\,M \mid \kw{for}\,(x\,\texttt{<--}\,\mathcal{E})\,M \\
                  & \mid & \kw{insert}\,(\mathcal{E})\, M \mid \kw{insert}\,(\kw{table}\,n)\,\mathcal{E} \\
                  & \mid & \kw{update}\,(x\,\texttt{<--}\,\mathcal{E})\,\kw{where}\,M\,\kw{set}\,(l_i=N_i)_{i=1}^n \\
                  & \mid & \kw{delete}\,(x\,\texttt{<--}\,\mathcal{E})\,\kw{where}\,M
    \end{array}
  \]

  \caption{Semantics of \Links.}
\label{fig:links-semantics}
\end{figure}

Expressions include standard constructs such as constants, variables,
record construction and field projection, conditionals, $n$-ary recursive functions and
application.  We freely use pair types $\rcd{A,B}$ and pair syntax
$\rcd{M,N}$ and projections $M.1$, $M.2$ etc., which are easily
definable using records.  Constants $c$ can be functions such as
integer addition, equality tests, etc.; their types are collected in a
signature $\Sigma$. The signature $\Sigma$ is also a simple model of a database: it maps tables to their contents.
In \Links we write $\kw{var}~ x = M; N$ for
binding a variable $x$ to the value of $M$ in expression $N$.  The semantics of the \Links
constructs discussed so far is call-by-value.  The expression
$\kw{query} ~\braces{M}$ introduces a query block, whose content is
not evaluated in the usual call-by-value fashion but instead first
\emph{normalized} to a form equivalent to an SQL query, and then
submitted to the database server.  The resulting table (or tables, in
the case of a nested query result) are then translated into a \Links
value.  Queries can be constructed using the expressions for the empty
collection $\coll{}$, singleton collection $\coll{M}$, and
concatenation of collections $M \plusplus N$.  In addition, the
comprehension expressions \lstinline!for(x <-- $M$) $N$! and
\lstinline!for(x <- $M$) $L$!  allow us to form queries involving
iteration over a collection.  The difference between the two
expressions is that \lstinline!for($x$ <-- $M$)! expects $M$ to be a
table reference, whereas \lstinline!for($x$ <- $M$)!  expects $M$ to
be a collection.  The expression $\kw{where}~(M)~N$ is equivalent to
$\kw{if}~(M)~\{N\}~\kw{else}~\{\coll{}\}$, and is
intended for use in filtering query results.  The expression
$\kw{empty}~(M)$ tests whether the collection produced by $M$ is
empty.  These comprehension syntax constructs can also be used outside
a query block, but they are not guaranteed to be translated to queries
in that case. The $\kw{insert}$, $\kw{delete}$ and $\kw{update}$
expressions perform updates on database tables; they are implemented
by direct translation to the analogous SQL update operations.

Figure~\ref{fig:links-semantics} presents the evaluation judgment
$\Sigma,M \to \Sigma',M'$ for \Links expressions.  We employ
evaluation contexts (following Felleisen and
Hieb~\cite{TCS1992FelleisenH}) $\mathcal{E}$ and define the semantics
using several axioms that handle redexes and a single inference rule
that shows how to evaluate an expression in which a redex occurs
inside an evaluation context. The rule for $\kw{update}$ uses
syntactic sugar for record update called $\kw{with}$ for brevity.
Most of the rules in
Figure~\ref{fig:links-semantics} are pure in the sense that they have
no side-effect on the state of the database. Only the rules for
$\kw{insert}$, $\kw{delete}$ and $\kw{update}$ may change the database
state. The rules here present the semantics of \Links at a high level,
and do not model the exact behavior of query evaluation; instead the
$\kw{query}~\{M\}$ operation just evaluates to $M$.
We assume functions used in database queries and updates are total and have a database equivalent.
This is assured by a type and effect system in the full language.
Lindley and Cheney~\cite{TLDI2012LindleyC} present a more detailed model that
also shows how flat \Links queries are normalized and evaluated externally
using SQL and Cheney et al.~\cite{SIGMOD2014CheneyLW}
shows how nested queries are implemented.

\begin{figure*}[tbp]
\small
\begin{mathpar}
\infer[Const]
{ \Sigma(c) = A}
{Γ ⊢ c : A}

\infer[Var]
{ x:A \in \Gamma }
{Γ ⊢ x : A}

\infer[Record]
{ Γ ⊢ M_i : A_i \quad (i \in \{1,\dots,n\})}
{ Γ ⊢ \rcd{ l_i \texttt{\small\ = } M_i }_{i=1}^n : \rcd{ l_i : A_i }_{i=1}^n}

\infer[Projection]
{ Γ ⊢ M : \rcd{ l_i : A_i}_{i=1}^n}
{ Γ ⊢ M\code{.}l_k : A_k }

\infer[Fun]
{Γ , [x_i : A_i]_{i=1}^n ⊢ M : B}
{Γ ⊢ \kw{fun }\rcd{ {x_i}|_{i=1}^n } \braces{M}: \rcd{A_i|_{i=1}^n } \arr B}

\infer[App]
{Γ ⊢ M : \rcd{A_i| _{i = 1}^n} ~\lstinline!->!~ B\\
Γ ⊢ N_i : A_i \quad (i \in \{1,\ldots,n\})}
{Γ ⊢ M(N_i|_{i=1}^n) : B}

\infer[Var]
{Γ ⊢ M : A\\
Γ, x:A ⊢ N : B}
{Γ ⊢ \kw{var}~x~=M; N : B}

\infer[Query]
{Γ ⊢ M : \coll{ A}  \\
A :: \QueryType}
{Γ ⊢ \kw{query}~\{M\} :  \coll{A}}

\infer[Empty]
{Γ ⊢ M : \coll{A}  }
{Γ ⊢ \kw{empty}(M) :  \BoolTy}

\infer[Table]
{ R :: \BaseRow}
{Γ ⊢ \kw{table}~n~\kw{with}~(R) : \kw{table} (R)}

\infer[Empty-List]
{\strut}
{ Γ ⊢ \coll{} : \coll{ A }}

\infer[List]
{ Γ ⊢ M : A}
{Γ ⊢ \coll{ M} : \coll{ A }}

\infer[Concat]
{Γ ⊢ M : \coll{A}\\
Γ ⊢ N : \coll{A}}
{Γ ⊢ M \plusplus N : \coll{A}}

\infer[For-List]
  {Γ ⊢ L : \coll{ A }\\
   Γ, x : A ⊢ M : \coll{B}}
  {\Gamma \vdash \kw{for}\ \rcd{ x\  \texttt{\small<-}\ L }\
    M : \coll{B}}

\infer[Where]
  {Γ ⊢ M : \code{Bool}\\
   Γ ⊢ N : \coll{B}}
  {\Gamma \vdash \kw{where}\ \rcd{ M }\ N : \coll{B}}

\infer[For-Table]
  {Γ ⊢ L : \kw{table}(R)\\
   Γ, x : \rcd{R}  ⊢ M : \coll{B}}
  {\Gamma \vdash \kw{for}\ \rcd{ x\  \drarr\ L }\ M : \coll{B}}

\infer[Insert]
{Γ ⊢ L : \kw{table}(R)\\
Γ ⊢ M : \coll{\rcd{R}}}
{Γ ⊢ \kw{insert}~L~\kw{values}~{M} : \rcd{}}

\infer[Update]
{Γ ⊢ L : \kw{table}(R)\\
Γ,x:\rcd{R}  ⊢ M : \BoolTy\\
 Γ, x:\rcd{R} ⊢ N : \rcd{R}}
{Γ ⊢ \kw{update}~(x \drarr L)~\kw{where}~{M}~\kw{set}~{N} : \rcd{}}

\infer[Delete]
{Γ ⊢ L : \kw{table}(R)\\
Γ,x:\rcd{R} ⊢ M : \BoolTy}
{Γ ⊢ \kw{delete}~(x \drarr L)~\kw{where}~{M} : \rcd{}}
\end{mathpar}
  \caption{Typing rules for \Links.}
  \label{fig:links-typerules}
\end{figure*}

The type system (again a simplification of the full system) is
illustrated in Figure~\ref{fig:links-typerules}.  Many rules are
standard; we assume a typing signature $\Sigma$ mapping constants and
primitive operations to their types.  The rule for
\lstinline!query {$M$}! refers to an auxiliary judgment
$A :: \QueryType$ that essentially checks that $A$ is a valid
query result type, meaning that it is constructed using base types and
collection or record type constructors only:
\[\infer{\strut}{O ::\QueryType}\quad
\infer{[A_i :: \QueryType]_{i=1}^n}{\rcd{l_i:A_i}_{i=1}^n::\QueryType}\quad
\infer{A::\QueryType}{\coll{A}::\QueryType}\]
Similarly, the $R :: \BaseRow$ judgment ensures that the types used
in a row are all base types:
\[\infer{
\strut
}{
\cdot :: \BaseRow
}
\quad
\infer
{
R :: \BaseRow
}{
R,l:O :: \BaseRow
}
\]
The full \Links type
system also checks that the body $M$ uses only features available on
the database (and only calls functions that satisfy the same
restriction).  The rules for other query operations are
straightforward, and similar to those for monadic comprehensions in
other systems.  Finally, the rules for updates (insert, update, and
delete) are also mildly simplified; in the full system, the conditions
and update expressions are required to be database-executable
operations.  \citet{TLDI2012LindleyC} present a more complete
formalization of \Links's type system that soundly characterizes the
intended run-time behavior.

The core language of \Links we are using is a simplification of the
full language in several respects.  \Links includes a number of
features (e.g.\ recursive datatypes, XML literals, client/server
annotations, and concurrency features) that are important parts of its
Web programming capabilities but not needed to explain our
contribution.  \Links also uses a type-and-effect system to determine
whether the code inside a \lstinline!query! block is translatable to
SQL, and which functions can be called safely from query blocks.  We
use a simplified version of \Links's type system that leaves out these
effects and does not deal with polymorphism.  Our implementation
does handle these features, with some limitations discussed later.

\section{Extending Links with provenance}
\label{sec:design}

In this paper we follow a well-explored approach to modeling
provenance by propagating \emph{annotations} of various kinds. Roughly
speaking, the idea is to interpret a query using a nonstandard
semantics over data with additional annotations on fields or records.
The nonstandard semantics propagates annotations from the input to the
output in a way that is intended to convey useful information about
how the results were derived from the inputs; sometimes the semantics
is proved correct with respect to some specification of the intended
meaning.  This idea dates to Wang and Madnick's \emph{polygen}
model~\cite{VLDB1990WangM}, and is adopted in much subsequent work on
provenance in databases (see~\cite{FTDB2009CheneyCT} for a survey).

In this section we describe two extensions of \Links: \WLinks and
\LLinks which provide language support for \wpr and lineage,
respectively. For both languages, we discuss language design, syntax,
semantics, type system, and most importantly, how provenance
annotations are propagated. We discuss how to provide initial
annotations for \WLinks here, and in Section~\ref{sec:translation} for
\LLinks. For both languages, the correctness theorems are only
concerned with the faithful propagation of annotations, not what the
annotations actually are.

\subsection{\WLinks}
\lstset{language=WLinks}

\WLinks extends \Links with language support for computing the \wpr of
database queries.  The syntax shown in Figure~\ref{fig:links-syntax}
is extended as follows:
\begin{eqnarray*}
  V & \Coloneqq & \cdots \mid V^c \\
  O &\Coloneqq& \cdots \mid \Prov(O)\\
  L,M,N &\Coloneqq & \cdots \mid \kw{data}~M \mid \kw{prov}~M 
                     \mid \kw{table}~n~\kw{with}~(R)~\kw{where}~S\\
  S &\Coloneqq& \cdot \mid S, l ~\kw{prov}~s\\
  s &\Coloneqq& \kw{default} \mid M
\end{eqnarray*}
Values $V$ can be annotated with an element $c$ of some sufficiently
large set of distinguishable atomic annotations, often called
\emph{colors}.  We will use \wpr triples for colors.  That is, an
annotation consists of a triple $(R,f,i)$ where $R$ is the source
table name, $f$ is the field name, and $i$ is the row identifier.  We
introduce the type constructor $\Prov(O)$, where $O$ is a type
argument of base type.  We treat $\Prov(O)$ itself as a base type, so
that it can be used as part of a table type.  (This is needed for
initializing provenance as explained below.)  Values of type
$\Prov(O)$ are annotated values $V^c$, where the annotation consists
of a triple $(R,f,i)$ where $R$ is the source table name, $f$ is the
field name, and $i$ is the row identifier.  For example,
$42^{("QA", "a", 23)}$ represents the answer 42 which was copied from
row 23, column {\sf a}, of table {\sf QA}.  The syntax above allows
arbitrary values to be annotated; however, the type system will only
permit values of base type to be annotated.  Annotated values are not
available in source programs; only the \WLinks runtime can construct
annotated values.

\begin{figure}[htb]
  \small
  \[
    \begin{array}{rcl}
      \Sigma, \kw{prov}\,V^c & ⟶ & \Sigma, c\\
      \Sigma, \kw{data}\,V^c & ⟶ & \Sigma, V\\
    \end{array}
  \]

  \[
    \begin{array}{rcl}
      \mathcal{E} & \Coloneqq & \dots \mid \kw{prov}\,\mathcal{E} \mid \kw{data}\,\mathcal{E} \\
    \end{array}
  \]
  \caption{Additional evaluation and context rules for \WLinks.}
  \label{fig:wlinks-semantics}
\end{figure}

We add two additional keywords \lstinline!prov! and \lstinline!data! to extract from an annotated value the provenance annotation and the value itself, respectively.
We extend the semantics from Figure~\ref{fig:links-semantics} with rules for these keywords as seen in Figure~\ref{fig:wlinks-semantics}.

\begin{figure}[tbp]
\small
\begin{mathpar}
\infer [Prov]
  { Γ ⊢ M : \Prov(A)}
  { Γ ⊢ \kw{prov}~M : \code{(\StringTy, \StringTy, \IntTy)} }

\infer [Data]
  { Γ ⊢ M : \Prov(A)}
  { Γ ⊢ \kw{data}~M : A }

\infer[Table]
 { R :: \BaseRow\\
Γ ⊢ S : \ProvSpec(R)
}
 { Γ ⊢ \kw{table}\ n\ \kw{with}\ \rcd{ R }\ \kw{where}\ S : \kw{table}\rcd{R \triangleright S}}

\infer[Insert]
{Γ ⊢ L : \kw{table}(R)\\
Γ ⊢ M : \coll{\rcd{\erase{R}}}}
{Γ ⊢ \kw{insert}~L~\kw{values}~{M} : \rcd{}}

\infer[Update]
{Γ ⊢ L : \kw{table}(R)\\
Γ,x:\rcd{\erase{R}}  ⊢ M : \BoolTy\\
 Γ, x:\rcd{\erase{R}} ⊢ N : \rcd{R}}
{Γ ⊢ \kw{update}~(x \drarr L)~\kw{where}~{M}~\kw{set}~{N} : \rcd{}}

\infer[Delete]
{Γ ⊢ L : \kw{table}(R)\\
Γ,x:\rcd{\erase{R}} ⊢ M : \BoolTy}
{Γ ⊢ \kw{delete}~(x \drarr L)~\kw{where}~{M} : \rcd{}}

\\

\infer{
\strut
}{
Γ ⊢ \cdot : \ProvSpec(R)
}

\infer{
Γ ⊢ S : \ProvSpec(R)
}{
Γ ⊢ S, l~\kw{prov}~\kw{default} : \ProvSpec(R)
}

\infer{
Γ ⊢ S : \ProvSpec(R)\\
Γ ⊢ M : \rcd{ R} ~\texttt{\small ->} ~\rcd{\StringTy, \StringTy ,\IntTy}
}{
Γ ⊢ S, l~\kw{prov}~{M} : \ProvSpec(R)
}
\end{mathpar}
  \caption{Additional typing rules for \WLinks. }
  \label{fig:wlinks-typerules}

\begin{eqnarray*}
\erase{O} &=& O\\
\erase{\Prov(A)} &=& \erase{A}\\
\erase{(l_i:A_i)_{i=1}^n} &=& (l_i:\erase{A_i})_{i=1}^n\smallskip\\
R \triangleright \cdot &=&R \\
(R,l:O )\triangleright (S,l~\kw{prov}~s) &=& (R \triangleright S),l:\Prov(O)
\end{eqnarray*}
  \caption{\WLinks type erasure and augmentation. }
  \label{fig:wlinks-auxiliary}
\end{figure}

Only the \LLinks runtime can create annotated values, and it only annotates database values.
We allow programmers to indicate which columns in a database table should carry annotations and give some control over what the annotations themselves are.
To this end, we extend the syntax of table expressions to
allow a list of \emph{provenance initialization specifications} $l~\kw{prov}~s$.
A specification $s$ is either the keyword $\kw{default}$ or an
expression $M$ which is expected to be of type
$\rcd{\overline{l_i:O_i}} \arr (\StringTy,\StringTy,\IntTy)$.
This way we have three different kinds of columns:
plain columns without annotations;
columns with \emph{default} \wpr where the annotation will be the table name, column name, and the row's \lstinline!oid!; and
columns with annotations that are computed by some user-defined function that takes the table row as input.

Default \wpr can be understood as user-defined \wpr with a
compiler-generated function of the form 
\lstinline!fun (r) { (T, C,  r.oid) }! 
where \lstinline!T! and \lstinline!C! are replaced by the
table and column name, respectively.
For example, if we added default \wpr to the phone field of the Agencies table, we would execute the following function on every row, to obtain the phone numbers provenance: \lstinline!fun (a) { ("Agencies", "phone", a.oid) }!.

The typing rules for the new constructs of \WLinks are shown in
Figure~\ref{fig:wlinks-typerules}.  These rules employ an auxiliary
judgment $\Gamma \vdash S : \ProvSpec(R)$, meaning that in context
$\Gamma$, the provenance specification $S$ is valid with respect to
record type $R$.  As suggested by the typing rule, the $\kw{prov}$
keyword extracts the provenance from a value of type $\Prov(A)$, and
$\kw{data}$ extracts its data, the $A$-value.  The most complex rule
is that for the $\kw{table}$ construct.  

The rules make use of an \emph{erasure} operation $\erase{R}$ that
takes a record or base type and replaces all occurrences of $\Prov(A)$
with $A$.  The rule for typing table references also uses an auxiliary
operation $R \triangleright S$ that defines the type of the provenance
view of a table whose fields are described by $R$ and whose provenance
specification is $S$.  As for ordinary tables, we check that the
fields are of base type.  These operations are defined in
Figure~\ref{fig:wlinks-auxiliary}.

The following proofs and definitions are based on previous work by Buneman et al.~\cite{TODS2008BunemanCV} in the context of nested relational algebra.
The main correctness property of \wpr is that annotations on values are correctly propagated.
It should not be the case that we construct annotated values out of thin air.
For the propagation behavior to be correct, it does not matter what the annotations are or where they come from.
Buneman et al.\ discuss some other interesting properties which do not hold in our language.
In their work, annotations are completely abstract, and queries have no way to inspect them.
Therefore, they can show that queries are invariant under recoloring of the input.
\WLinks has the \lstinline!prov! keyword to inspect provenance, therefore we cannot expect the same to hold here.
However, we speculate that a similar property holds for sufficiently polymorphic functions.

We assume a context $\Sigma$ where values inside tables are annotated with colors.
We do not make any assumptions about these colors.
However, they are particularly useful when they are distinct.
In the case of distinct annotations on the input, we can look at the output and trace back annotated values to their source (assuming evaluation does not conjure up new annotated values out of thin air).
In Figure~\ref{fig:cso} we define the function $\mathit{cso}_\Sigma$ for finding all \emph{colored subobjects} of a \WLinks{} term.
This function allows us to find the annotations in the program and state that we do not invent any during evaluation.
Thus, if we start with a distinctly annotated database and no annotated constants, we can then guarantee that all annotated values in the result of evaluation come, without modification, directly from the database.
Theorem~\ref{thm:where-correctness} formally states this intuition of evaluation not inventing annotated values.

\newcommand\cso[1]{\mathit{cso}_\Sigma#1}
\newcommand\csop[1]{\mathit{cso}'_\Sigma#1}

\begin{figure}
  \small
  \[
    \begin{array}{lcl}
      \cso(V^a) &=& \{V^a\} \cup \cso(V)\\
      \cso(c) &=& \emptyset\\
      \cso(\texttt{[]}) &=& \emptyset\\
      \cso(\texttt{[}M\texttt{]}) &=& \cso(M)\\
      \cso(M \concat N) &=& \cso(M) \cup \cso(N) \\
      \cso((l_i = M_i)_{i=1}^n) &=& \bigcup_{i=1}^n \cso(M_i)\\
      \cso(M.l) &=& \cso(M)\\ 
      \cso(\kw{fun}\,f(x_i| _{i = 1}^n)\,M) &=& \cso(M) \\
      \cso(M(N_i | _{i=1}^n)) &=& \cso(M) \cup \bigcup_{i=1}^n \cso(N_i)\\
      \cso(\kw{var}\, x = M; N) &=& \cso(M) \cup \cso(N)\\
      \cso(\kw{if}\,(L)\,M\,\kw{else}\,N ) &=& \cso(L) \cup \cso(M) \cup \cso(N)\\
      \cso(\kw{query}\,M) &=& \cso(M)\\
      \cso(\kw{table}\, n) &=& \cso(\Sigma(n)) \\ 
      \cso(\kw{empty}(M)) &=& \cso(M)\\
      \cso(\kw{for}\,(x\,\texttt{<-}\,M)\,N ) &=& \cso(M) \cup \cso(N)\\
      \cso(\kw{for}\,(x\,\texttt{<--}\,M)\,N ) &=& \cso(M) \cup \cso(N) \\
    \end{array}
  \]
  \caption{Colored subobjects in \WLinks{} expressions.}
  \label{fig:cso}
\end{figure}




We first show a helpful lemma: the colored subobjects of a term substituted into an
  evaluation context $\mathcal{E}[M]$ can be obtained by considering
  the evaluation context $\mathcal{E}$ and term $M$ separately,
  instead.  We extend $\cso(-)$ to operate on evaluation contexts
  in the obvious way.
\begin{lemma}\label{lemma:cso-evaluation-context}
Given evaluation context $\mathcal{E}$ and term $M$, we have:
\[ \cso(\mathcal{E}[M]) = \cso(\mathcal{E}) \cup \cso(M) \]
\end{lemma}

\begin{proof}
  Proof by induction on the structure of the evaluation context.
  In the case for $\mathcal{E}=[]$ we take the colored subobjects of a hole to be the empty set.
  The other cases are straightforward.




\end{proof}

\begin{theorem}[Correctness of where-provenance]\label{thm:where-correctness}
  Let $M$ and $N$ be \WLinks terms, and let $\Sigma$ be a context that provides annotated table rows. We have:
  \[
    \Sigma, M \longrightarrow \Sigma, N \Rightarrow \cso(N) \subseteq \cso(M) 
  \]
\end{theorem}

\begin{proof}
  Proof by induction on the derivation of the evaluation relation $\longrightarrow$.
  We show some representative cases here, the full proof is in \ref{sec:where-correctness-proof}.

  \begin{itemize}
  \item Case $\kw{for}\,(x\,\texttt{<-}\,\texttt{[]})\,M ⟶ \texttt{[]}$:
$    \cso(\texttt{[]})
     = \emptyset  \subseteq \cso(\kw{for}\,(x\,\texttt{<-}\,\texttt{[]})\,M)$

  \item Case $\kw{for}\,(x\,\texttt{<-}\,\texttt{[}V\texttt{]})\,M ⟶ M[x\coloneqq V]$:
  \begin{align*}
    \cso(M[x\coloneqq V])
    & \subseteq \cso(M) \cup \cso(V) \\
    & = \cso(\kw{for}\,(x\,\texttt{<-}\,\texttt{[}V\texttt{]})\,M)
  \end{align*}

  \item Case $\kw{for}\,(x\,\texttt{<-}\,V \concat W)\,M ⟶ (\kw{for}\,(x\,\texttt{<-}\, V)\,M) \concat (\kw{for}\,(x\,\texttt{<-}\, W)\,M)$:
  \begin{align*}
    \cso(\kw{for}\,(x\,\texttt{<-}\,V \concat W)\,M)
    & = \cso(V \concat W) \cup \cso(M) \\
    & = \cso(V) \cup \cso(W) \cup \cso(M) \\
    & = \cso((\kw{for}\,(x\,\texttt{<-}\, V)\,M) \concat (\kw{for}\,(x\,\texttt{<-}\, W)\,M))
  \end{align*}
    
  \item Case $M ⟶ M' \Rightarrow \mathcal{E}[M] ⟶ \mathcal{E}[M']$ (evaluation step inside a context):
    \begin{align*}
      \cso(\mathcal{E}[M'])
      & = \cso(\mathcal{E}) \cup \cso(M') & \text{Lemma~\ref{lemma:cso-evaluation-context}} \\
      & \subseteq \cso(\mathcal{E}) \cup \cso(M) & \text{IH} \\
      & = \cso(\mathcal{E}[M]) & \text{Lemma~\ref{lemma:cso-evaluation-context}}
    \end{align*}
  \end{itemize}
\end{proof}


\subsection{Lineage}
\lstset{language=LLinks}
\newcommand\linarr{\longrightarrow_\mathsf{L}}

\LLinks adds the keyword \lstinline!lineage! to \Links.
Like the keyword \lstinline!query!, it is followed by a block of code that will be translated into SQL and executed on the database.
The \lstinline!query! keyword only affects where and how the evaluation takes place.
The result is the same as if database tables were lists in memory.
The \lstinline!lineage! keyword also triggers translation of the following code block into SQL\@.
However, the query is rewritten to not only compute the result, but
every row of the result is annotated  with its lineage.
The syntax is extended as follows:
\begin{eqnarray*}
  L,M,N &\Coloneqq& \cdots \mid\kw{lineage}\{M\}
\end{eqnarray*}
The expression \lstinline!lineage {$M$}! is similar to
\lstinline!query {$M$}!, in that $M$ must be an expression that can be
executed on the database (that is, terminating and side-effect free;
this is checked by \Links's effect type system just as for
\lstinline!query {M}!).  If $M$ has type $\coll{A}$ (which must be an
appropriate query result type) then the type of the result of
\lstinline!lineage {$M$}!  will be $𝔏⟦\coll{A}⟧$, where $𝔏⟦-⟧$ is a
type translation that adjusts the types of collections $\coll{A}$ to
allow for lineage, as shown in
Figures~\ref{fig:llinks-type-translation} and~\ref{fig:llinks-typing}.

\begin{figure}[tb]
  \begin{eqnarray*}
\Lin(A) &=& (\mathsf{data}:A,\mathsf{prov}:\coll{(\StringTy,\IntTy)})\\
  𝔏⟦O⟧ &=& O\\
  𝔏⟦A\arr B⟧ &=& 𝔏⟦A⟧ \arr 𝔏⟦B ⟧\\
  𝔏⟦\rcd{l_i:A_i}_{i=1}^n⟧ &=& \rcd{l_i:𝔏⟦A_i⟧}_{i=1}^n\\
  𝔏⟦\coll{A}⟧ &=& \coll{\Lin(𝔏⟦A⟧)}\\
  𝔏⟦\kw{table}(R)⟧ &=& 𝔏⟦\coll{\rcd{R}} ⟧
  \end{eqnarray*}
\caption{Lineage type translation}
\label{fig:llinks-type-translation}
  \begin{mathpar}
\infer[Lineage]
{\Gamma \vdash M : \coll{A} \\
A :: \QueryType}
{\Gamma \vdash \kw{lineage}~\{M\} :  𝔏⟦\coll{A}⟧}
\end{mathpar}
  \caption{Additional typing rule for \LLinks}
\label{fig:llinks-typing}
\end{figure}

\begin{figure}[tb]
  \small
  \begin{mathpar}
    \infer[]
    {\hat{\Sigma}, \mathit{annotate}(M) \linarr^* \hat{\Sigma}, \hat{L} }
    {\Sigma, \kw{lineage}\,M \longrightarrow \Sigma, \mathit{a2d}(\hat{L})}
  \end{mathpar}

  \begin{align*}
    \mathit{annotate}(\texttt{[]}) & = \texttt{[]} \\
    \mathit{annotate}(\texttt{[$V$]}) & = \texttt{[}\mathit{annotate}(V)\texttt{]}^\emptyset \\
    \mathit{annotate}(V \concat W) & = \mathit{annotate}(V) \concat \mathit{annotate}(W) \\
\\
    \mathit{a2d}(\texttt{[]}) & = \texttt{[]} \\    
    \mathit{a2d}(\texttt{[}V\texttt{]}^{\{a_1, \dots, a_n\}}) & = \texttt{[}(\code{data}=\mathit{a2d}(V), \code{prov}=\texttt{[}a_1, \dots, a_n\texttt{]})\texttt{]} \\
    \mathit{a2d}(V \concat W) & = \mathit{a2d}(V) \concat \mathit{a2d}(W)
  \end{align*}
  
  \caption{\LLinks semantics.}
  \label{fig:llinks-semantics}
\end{figure}

A \lstinline!lineage! block evaluates in one step to its result, as can be seen in Figure~\ref{fig:llinks-semantics}.
The result is determined by a second evaluation relation that is only used ``inside'' lineage blocks: $\linarr$.
The language which $\linarr$ operates on is \LLinks, except that list values are replaced by a variant of lists, $\hat{L}$, where every list element is annotated with a set of colors:

\begin{eqnarray*}
  V &\Coloneqq& \cdots \mid \hat{L} \\
  \hat{L} &\Coloneqq& \texttt{[]} \mid \texttt{[}V\texttt{]}^a \mid
                      \hat{L} \concat \hat{L}\\
M& \Coloneqq& \dots \mid M^{\cup b}
\end{eqnarray*}

Note how the set of annotations $a$ is on the singleton list constructor, not the actual element value as you might expect.
We use annotations to track lineage, which describes \emph{why} the value, or row, is in the result.
Lineage is not concerned with what the value actually is.

We represent lineage as a list of rows in the database and identify rows by their table name and row number.
Every occurrence of the list type constructor in the type of a lineage query result is replaced by a list of records of data and its provenance. For example, if a \lstinline!query! block has type \lstinline![Bool]!, the result of the same code in a \lstinline!lineage! block has type \lstinline![(data: Bool, prov: [(String, Int)])]!.

There are two functions for going from \LLinks values to annotated values used inside \lstinline!lineage! blocks, and back.
The first function is $\mathit{annotate}$, which recursively annotates \LLinks lists with empty lineage annotations.
We assume an extension of this function to non-list values and arbitrary \LLinks terms in the obvious way.
Only rows in database tables will have nonempty lineage annotations, provided by an extended context $\hat{\Sigma}$.
The second function is $\mathit{a2d}$, which recursively transforms annotated lists into plain data \LLinks lists.
Nonlist values are traversed in the obvious way.
Every annotated list element will be transformed into a record with \lstinline!data! and \lstinline!prov! fields.
The \lstinline!prov! field will hold the lineage annotations, a set of colors, as a list.
Here we assume that colors are \LLinks values.  In practice they will be pairs of table name and row number; in theory we could use anything and define one more function to go from color to \LLinks value.


\begin{figure}[tb]
  \small

  \begin{align*}
    \hat\Sigma, \texttt{[]}^{\cup b} & \linarr \hat\Sigma, \texttt{[]} \\
    \hat\Sigma, (\texttt{[} V \texttt{]}^a)^{\cup b} & \linarr \hat\Sigma, \texttt{[} V \texttt{]}^{a \cup b} \\
    \hat\Sigma, (V \concat W)^{\cup b} & \linarr \hat\Sigma, V^{\cup b} \concat W^{\cup b} \\
    \hat\Sigma, (\kw{fun}\,f(x_i|_{i=0}^n)\,M)(V_i|_{i=0}^n) & \linarr \hat\Sigma, M[x_i \coloneqq V_i]_{i=0}^n \\
    \hat\Sigma, \kw{var}\,x = V; M & \linarr \hat\Sigma, M[x \coloneqq V] \\
    \hat\Sigma, \kw{for}\,(x\,\texttt{<- []})\,M & \linarr \hat\Sigma, \texttt{[]} \\
    \hat\Sigma, \kw{for}\,(x\,\texttt{<-}\,\texttt{[}V\texttt{]}^a)\, M & \linarr \hat\Sigma, (M[x \coloneqq V])^{\cup a} \\
    \hat\Sigma, \kw{for}\,(x\,\texttt{<-}\,V \concat W)\,M & \linarr \hat\Sigma, (\kw{for}\,(x\,\texttt{<-}\,V)\,M) \concat \kw{for}\,(x\,\texttt{<-}\,W)\,M \\
    \hat\Sigma, \kw{for}\,(x\,\texttt{<--}\,\kw{table}\,t)\, M & \linarr \hat\Sigma, \kw{for}\,(x\,\texttt{<-}\,\hat{\Sigma}(t))\,M \\
    \hat\Sigma, \kw{query}(V) & \linarr \hat\Sigma, V \\
    \hat\Sigma, \kw{if}(\kw{true})\,M\,\kw{else}\,N & \linarr \hat\Sigma, M \\
    \hat\Sigma, \kw{if}(\kw{false})\,M\,\kw{else}\,N & \linarr \hat\Sigma, N \\
    \hat\Sigma, (l_i = V_i)_{i=1}^n.l_k & \linarr \hat\Sigma, V_k \\
  \end{align*}
  \[
    \begin{array}{rcl}
    \mathcal{E} & \Coloneqq & \dots \mid \mathcal{E}^{\cup b} \\
    \end{array}
  \]
  \caption{Propagation of lineage annotations.}
  \label{fig:llinks-lineage-semantics}
\end{figure}

Evaluation inside \lstinline!lineage! blocks is almost the same as evaluation outside.
A \lstinline!lineage! block is similar to a \lstinline!query! block in that it can contain only pure, nonrecursive functions, and no database updates.
  We do not support \lstinline!empty! inside lineage blocks, because it can lead to nonmonotonic queries.
Figure~\ref{fig:llinks-lineage-semantics} shows the evaluation rules.
The major differences from regular evaluation are in the treatment of \lstinline!for! comprehensions and the new syntax $M^{\cup b}$.
A table comprehension takes the table values from an annotated signature $\hat{\Sigma}$, which maps tables to lists with lineage annotations.
A \lstinline!for! comprehension over a singleton list adds the singleton's annotation to all of the elements in the output list.
For this use alone we introduce the new type of expression $M^{\cup b}$.
It takes a term and a set of annotations, evaluates the term to a list value, and adds the annotations.
This is not syntax intended to be used by the programmer.

\begin{figure}[tb]
  \small
  \begin{align*}
    \| \texttt{[} M \texttt{]}^a \| & = a \cup \| M \| \\
    \| \texttt{[]} \| & = \emptyset \\
    \| M \concat N \| & = \| M \| \cup \| N \| \\
    \| M^{\cup b}\| & = b \cup \| M \| \\
    \| \kw{table}\, t \| & = \| \hat{\Sigma}(t) \| \\
    \| \kw{for}\,(x\,\texttt{<-}\,M)\,N\| & = \| M \| \cup \| N \| \\
  \end{align*}
  
  \begin{align*}
    \texttt{[} M \texttt{]}^a|_b & = \begin{cases}
      \texttt{[} M|_b \texttt{]}^a & \text{if } a \subseteq b \\
      \texttt{[]} & \text{otherwise}
    \end{cases} \\
    \texttt{[]}|_b & = \texttt{[]} \\
    (M \concat N)|_b & = M|_b \concat N|_b \\
    M^{\cup a}|_b & = \begin{cases}
      (M|_b)^{\cup a} & \text{if } a \subseteq b\\
      \texttt{[]} & \text{otherwise}
    \end{cases}\\
    \kw{table}\, t|_b & = \kw{table}\, t \\
    (\kw{for}\,(x\,\texttt{<-}\,M)\,N)|_b & = \kw{for}\,(x\,\texttt{<-}\,M|_b)\,N|_b
  \end{align*}

  \begin{mathpar}
    \infer[]
    { }
    {V \sqsubseteq V}

    \infer[]
    { }
    {\texttt{[]} \sqsubseteq L}

    \infer[]
    {V \sqsubseteq V'}
    {\texttt{[} V \texttt{]}^b \sqsubseteq \texttt{[} V' \texttt{]}^b}

    \infer[]
    {V \sqsubseteq V' \\ W \sqsubseteq W' }
    {V \concat W \sqsubseteq V' \concat W'}

    \infer[]
    {\forall 1 \leq i \leq n: \\
      l_i = l_i'
      \\ V_i \sqsubseteq V_i'}
    {(l_i = V_i)_{i=1}^n \sqsubseteq (l_i' = V_i')_{i=1}^n}
  \end{mathpar}
  
  \caption{Auxiliary definitions to collect lineage, restrict values, and find sublists.}
  \label{fig:llinks-supporting-definitions}
\end{figure}

Lineage of a query result tells us which elements of the input were responsible for each element of the output to exist.
If we run the same query again, but on only that part of the input that was mentioned in the lineage annotations, we should get the same output.
Nonmonotonic queries, that is queries that use aggregations,
emptiness tests, or set difference, cause issues here:
For example consider the query that selects everything from table $a$ if table $b$ is empty.
Every row in the result would be annotated with a corresponding row in $a$.
One would also need to record somehow the fact that $b$ was empty.
We could annotate whole tables in addition to individual rows but this would complicate the annotation model.
For this work, we chose to only consider monotonic queries.

In order to state the lineage correctness property formally, we need three auxiliary definitions from Figure~\ref{fig:llinks-supporting-definitions}.
We only show the most relevant cases here, but extend both functions to the entire language in the obvious way. The full definitions can be found in \ref{def:lineage-aux-full}.
The function $\| \cdot \|$ collects all lineage annotations mentioned in a value and is extended to \LLinks terms.
The function $\cdot|_b$ restricts values, in particular list elements, to those annotated with a subset of annotations $b$.
We extend this to \LLinks terms in the obvious way and to annotated contexts such that tables mentioned in a restricted context $\hat{\Sigma}|_b$ do not contain rows which are not in $b$.
Note that this function always preserves list literals and values originating in the surrounding program because those are annotated with empty lineage.
Finally we have the recursive sublist relation $\sqsubseteq$.
For example \lstinline![(a = [2])]! $\sqsubseteq$ \lstinline![(a = [1]), (a = [2, 3])]!.

Suppose a monotonic \LLinks query $q$ evaluates, inside a lineage block, to an annotated value $\hat{v}$ in a context $\hat{\Sigma}$.
For every part $\hat{p}$ of the value $\hat{v}$ we can obtain a smaller context $\hat{\Sigma}|_{\|\hat{p}\|}$ by erasing all values from the original context $\hat{\Sigma}$ which are not mentioned in $\hat{p}$.
The lineage annotations are correct if every part $\hat{p} \sqsubseteq \hat{v}$ of the output $\hat{v}$ is also a part of the output $\hat{v}'$ obtained by evaluating the same query $q$ in the restricted context $\hat{\Sigma}|_{\|\hat{p}\|}$. 



\begin{theorem}\label{thm:lin-correctness-ind}
  Given monotonic terms $M$ and $N$, a context $\hat\Sigma$, and a set of annotations $c$, we have
  \[
    \hat\Sigma, M \linarr \hat\Sigma, N \quad \Rightarrow \quad M|_c = N|_c \quad \vee \quad \hat{\Sigma}|_c, M|_c \linarr \hat{\Sigma}|_c, N|_c
  \]
\end{theorem}
\begin{proof} By induction on the evaluation relation $\linarr$.
  We need the alternative $M|_c = N|_c$ because sometimes restriction can yield the empty list, on both sides, in which case there is no evaluation step to be made.
  The two interesting cases are the singleton \lstinline!for! comprehension, which introduces $M^{\cup a}$, and adding annotations to a singleton list, which eliminates $M^{\cup a}$.

  Case $\hat\Sigma, \kw{for}\,(x\,\texttt{<-}\,\texttt{[}V\texttt{]}^a)\,M \linarr \hat\Sigma, M[x \coloneqq V]^{\cup a}$:\\
  We have two cases, depending on $c$. If $a \subseteq c$ then
  $(\kw{for}\,(x\,\texttt{<-}\,\texttt{[}V\texttt{]}^a)\,M)|_c = \kw{for}\,(x\,\texttt{<-}\,\texttt{[}V|_c\texttt{]}^a)\,(M|_c)$ and therefore
  \[ \hat\Sigma|_c, \kw{for}\,(x\,\texttt{<-}\,\texttt{[}V|_c\texttt{]}^a)\,(M|_c)
    \linarr
    \hat\Sigma|_c, (M|_c[x \coloneqq V|_c])^{\cup a}
  \]
  Furthermore, we have $(M|_c[x \coloneqq V|_c])^{\cup a} = ((M[x \coloneqq V])|_c)^{\cup a}$, which can be shown by induction, but only states that $\cdot |_c$ is well-behaved with respect to substitution, and $((M[x \coloneqq V])|_c)^{\cup a} = (M[x \coloneqq V])^{\cup a}|_c$ by definition of $M^{\cup a} |_c$ in the case that $a \subseteq c$, and therefore
  \[
    \hat\Sigma|_c, (\kw{for}\,(x\,\texttt{<-}\,\texttt{[}V\texttt{]}^a)\,M)|_c \linarr \hat\Sigma|_c, (M[x \coloneqq V]^{\cup a})|_c
  \]
  Otherwise $a \not\subseteq c$ and on the left hand side we have
  \[
    (\kw{for}\,(x\,\texttt{<-}\,\texttt{[}V\texttt{]}^a)\,M)|_c = \kw{for}\,(x\,\texttt{<-}\,(\texttt{[}V\texttt{]}^a)|_c)\,(M|_c) = \kw{for}\,(x\,\texttt{<-}\,\texttt{[]})\,(M|_c)
  \] which evaluates to the empty list:
  \[
    \hat\Sigma|_c, \kw{for}\,(x\,\texttt{<-}\,\texttt{[]})\,(M|_c) \linarr \hat\Sigma|_c, \texttt{[]}
  \]
  Since $(M[x \coloneqq V]^{\cup a})|_c = \texttt{[]}$ we can conclude that \[\hat\Sigma|_c, (\kw{for}\,(x\,\texttt{<-}\,\texttt{[}V\texttt{]}^a)\,M)|_c \linarr \hat\Sigma|_c, (M[x \coloneqq V]^{\cup a})|_c\]

  Case $\hat\Sigma, (\texttt{[}V\texttt{]}^b)^{\cup a} \linarr \hat\Sigma, \texttt{[}V\texttt{]}^{a \cup b}$:\\
  Depending on $c$ we, again, have two cases.
  If $a \subseteq c$ then $(\texttt{[}V\texttt{]}^b)^{\cup a}|_c = (\texttt{[}V\texttt{]}^b|_c)^{\cup a}$.
  Now, if $b \subseteq c$ then $\texttt{[}V\texttt{]}^b|_c = \texttt{[}V|_c\texttt{]}^b$ and we have an evaluation step $\hat\Sigma|_c, (\texttt{[}V|_c\texttt{]}^b)^{\cup a} \linarr \hat\Sigma|_c, \texttt{[}V|_c\texttt{]}^{a \cup b}$ where the term on the right hand side is equal to $\texttt{[}V\texttt{]}^{a \cup b}|_c$.
  Otherwise, $b \not\subseteq c$ and $\texttt{[}V\texttt{]}^b|_c = \texttt{[]}$ but on the right hand side we also have $\texttt{[}V\texttt{]}^{a \cup b}|_c = \texttt{[]}$. In other words, by restricting with $c$ we get the same value on both sides.
  We reach the same conclusion in the case that $a \not\subseteq c$.
\end{proof}

\begin{corollary}\label{cor:lin-correctness-ind}
  By repeated application of Theorem~\ref{thm:lin-correctness-ind} we have
  \[
    \hat\Sigma, M \linarr^j \hat\Sigma, N
    \quad\Rightarrow\quad
    \hat{\Sigma}|_c, M|_c \linarr^k \hat{\Sigma}|_c, N|_c
  \]
  where $j,k \in \mathbb{N}$ and $k \leq j$.
  
\end{corollary}

\begin{lemma}\label{lem:pinvp}
  Given a value $\hat{v}$ and a subvalue $\hat{p} \sqsubseteq \hat{v}$ of that value, we have
  \[\hat{p} \sqsubseteq \hat{v}|_{\|\hat{p}\|}\]
\end{lemma}
\begin{proof} By induction on the subvalue relation $\sqsubseteq$.
  \begin{itemize}
  \item Cases $V \sqsubseteq V$ and $\texttt{[]}\sqsubseteq V$ are trivially true.
  \item Case $\texttt{[} V \texttt{]}^b \sqsubseteq \texttt{[} V' \texttt{]}^b$:
    We have $\texttt{[} V' \texttt{]}^b|_{\|\texttt{[} V \texttt{]}^b\|} = \texttt{[} V' \texttt{]}^b|_{b \cup \|V\|}$ by definition, and $V'|_{\|V\|} \sqsupseteq V$ by the induction hypothesis, and can therefore conclude $\texttt{[} V \texttt{]}^b \sqsubseteq \texttt{[} V' \texttt{]}^b|_{\|\texttt{[} V \texttt{]}^b\|}$.
  \item The cases for list concatenation and records are similar.

  
  \end{itemize}
\end{proof}

\begin{theorem}[Correctness of lineage]\label{thm:lin-correctness}
  Let $q$ be a monotonic query with $\|q\| = \emptyset$ and let $\hat\Sigma$ be a context, such that
  $q$ evaluates to $\hat{v}$ in $\hat\Sigma$: $ \hat{\Sigma}, q \linarr^* \hat{\Sigma}, \hat{v}$.
  Then for every sublist $\hat{p} \sqsubseteq \hat{v}$ we can evaluate $q$ in a restricted context $\hat\Sigma|_{\|\hat{p}\|}$ to obtain a value $\hat{v}'$ and $\hat{p}$ will be a sublist of $\hat{v}'$.

  \[ \forall \hat{p} \sqsubseteq \hat{v}: \hat\Sigma|_{\|\hat{p}\|}, q \linarr^* \hat{\Sigma}|_{\|\hat{p}\|}, \hat{v}' \wedge \hat{p} \sqsubseteq \hat{v}' \]
\end{theorem}
\begin{proof}
  Using Corollary~\ref{cor:lin-correctness-ind} of Theorem~\ref{thm:lin-correctness-ind} we have
  \[ \hat\Sigma|_{\|\hat{p}\|}, q|_{\|\hat{p}\|} \linarr^* \hat\Sigma|_{\|\hat{p}\|}, \hat{v}|_{\|\hat{p}\|} \]
  for any $\hat{p}$ and, because of Lemma~\ref{lem:pinvp}, $\hat{v}|_{\|\hat{p}\|} \sqsupseteq \hat{p}$ so set
  \[ \hat{v}' = \hat{v}|_{\|\hat{p}\|} \]
  Since $q$ has no annotations on its own, it is not affected by restriction: $q|_{\|\hat{p}\|} = q$ and we can conclude that
  \[ \hat\Sigma|_{\|\hat{p}\|}, q \linarr^* \hat{\Sigma}|_{\|\hat{p}\|}, \hat{v}' \wedge \hat{p} \sqsubseteq \hat{v}' \]
\end{proof}

\section{Provenance translations}
\label{sec:translation}

In the previous section, we have presented two extensions of \Links:
\WLinks, which supports where-provenance in queries, and \LLinks,
which supports lineage in queries. Here, we show that both extensions
can be implemented by a type-preserving source-to-source translation
to plain \Links.

\subsection{Where-Provenance}

\lstset{language=WLinks}

We define a type-directed translation from \WLinks to \Links based on
the semantics presented in the previous section.  The syntactic
translation of types $𝔚⟦-⟧$ is shown in
Figure~\ref{fig:wlinks-type-translation}.  We write $𝔚⟦\Gamma⟧$ for
the obvious extension of the type translation to contexts.  The
implementation extends the \Links parser and type checker, and
desugars the \WLinks AST to a \Links AST after type checking, reusing
the backend mostly unchanged.  The expression translation function is
also written $𝔚⟦-⟧$ and is shown in
Figure~\ref{fig:where-translation}.

\begin{figure}
\begin{eqnarray*}
    𝔚⟦O⟧ &=& O\\
    𝔚⟦A \arr B⟧ &=& 𝔚⟦A ⟧ \arr 𝔚⟦B⟧\\
    𝔚⟦(l_i:A_i)_{i=1}^n⟧ &=& (l_i:𝔚⟦A_i⟧)_{i=1}^n\\
    𝔚⟦\coll{A}⟧ &=& \coll{𝔚⟦A⟧}\\
    𝔚⟦\Prov(A)⟧ &=& (\mathsf{data}:  𝔚⟦A⟧,
                   \mathsf{prov}:(\StringTy,\StringTy,\IntTy) )\\
 𝔚⟦\kw{table}(R)⟧ &=& (\kw{table}(\erase{R}),() \arr \coll{𝔚⟦\rcd{R} ⟧} )
\end{eqnarray*}
  \caption{Type translation for \WLinks}
  \label{fig:wlinks-type-translation}
\end{figure}

\begin{figure*}[tb]
\small
\begin{minipage}{12cm}
\begin{eqnarray*}
𝔚⟦c ⟧ &=& c\\
𝔚⟦x⟧ &=& x\\
𝔚⟦(l_i = M_i)_{i=1}^n⟧ &=& (l_i=𝔚⟦M_i⟧)_{i=1}^n\\
𝔚⟦N.l⟧ &=& 𝔚⟦N⟧.l\\
𝔚⟦\kw{fun}(x_i|_{i=0}^n)~\{M\}⟧ &=& \kw{fun}(x_i|_{i=0}^n)~\{𝔚⟦M⟧
                                    \}\\
𝔚⟦M(N_i|_{i=0}^n) ⟧ &=& 𝔚⟦M ⟧(𝔚⟦N_i⟧|_{i=0}^n) \\
𝔚⟦\kw{var}~x=M;N ⟧ &=& \kw{var}~x = 𝔚⟦M ⟧; 𝔚⟦N⟧\\
𝔚⟦\kw{query}~\{M\}⟧ &=& \kw{query}~\{𝔚⟦M⟧\}\\
𝔚⟦[] ⟧  &=& []\\
𝔚⟦[M] ⟧  &=& [𝔚⟦M⟧ ]\\
𝔚⟦M \plusplus N ⟧  &=& 𝔚⟦M⟧ \plusplus 𝔚⟦N⟧
\\
𝔚⟦\kw{if}~(L)~\{M\}~\kw{else}~\{N\}⟧ &=& \kw{if}~(𝔚⟦ L⟧)~\{𝔚⟦M⟧\}~\kw{else}~\{𝔚⟦N⟧\}\\
𝔚⟦\kw{empty}~(M)⟧ &=& \kw{empty}~(𝔚⟦M⟧)\\
𝔚⟦\kw{for}~(x \rarr L) ~M⟧ &=& \kw{for}~(x \rarr 𝔚⟦L⟧) ~𝔚⟦M⟧ \\
𝔚⟦\kw{where}(M)~N⟧ &=& \kw{where}(𝔚⟦M⟧)~𝔚⟦N⟧ \\
 𝔚⟦\kw{for}~(x \drarr  L)~M⟧ &=&
\kw{for}~ (x \rarr 𝔚⟦L⟧.2())~𝔚⟦M⟧ \\
𝔚⟦\kw{data}~M⟧ &=& 𝔚⟦M⟧.\mathsf{data}\\
𝔚⟦\kw{prov}~M⟧ &=& 𝔚⟦M⟧.\mathsf{prov}\\
𝔚⟦\kw{insert} ~L~\kw{values}~M⟧ &=& \kw{insert} ~𝔚⟦ L⟧.1~\kw{values}~𝔚⟦ M⟧\\
𝔚⟦\kw{update} ~(x \rarr L)~\kw{where}~M~\kw{set}~N⟧ &=& \kw{update} ~(x \rarr 𝔚⟦L⟧.1)~\kw{where}~𝔚⟦M⟧~\kw{set}~𝔚⟦N⟧\\
𝔚⟦\kw{delete} ~(x \rarr L)~\kw{where}~M⟧  &=&\kw{delete} ~(x
                                                    \rarr
                                                    𝔚⟦L⟧.1)~\kw{where}~𝔚⟦M⟧
\end{eqnarray*}
\end{minipage}
\[
𝔚⟦\kw{table}~n~\kw{with} (R) \kw{where}~S⟧
=
(\kw{table}~n~\kw{with} ~(R), \kw{ fun} () \{ \kw{for} (x \drarr
    \kw{table}~n~\kw{with} ~(R)) \coll{\rcd{R \triangleright^n_x S}}\})
\]
\[\begin{array}{rcl}
\cdot \triangleright^n_x \cdot &=& \cdot\\
 (R,l:O) \triangleright^n_x \cdot &=& (R \triangleright^n_x \cdot),l=x.l \\
 (R,l:O) \triangleright^n_x (S,l~\kw{prov}~\kw{default}) &=& (R
                                                           \triangleright^n_x
                                                           S), l=\rcd{\kw{data} = x.l,
                                                         \kw{prov} = \rcd{n,l_d,x.oid}}\\
(R,l:O )\triangleright^n_x (S,l~\kw{prov}~M) &=& (R
                                                           \triangleright^n_x
                                                           S), l=\rcd{\kw{data} = x.l,
                                                         \kw{prov} =
                                                          𝔚⟦ M⟧(x) }
\end{array}\]
\caption{Translation of \WLinks to \Links, and auxiliary operation
  $R \triangleright^n_x S$}
\label{fig:where-translation}
\end{figure*}

Values of type $\Prov(O)$ are represented at runtime as
ordinary \Links records with type
\lstinline[language=Links]!(data: $O$, prov: (String, String, Int))!.
Thus, the keywords \lstinline!data!
and \lstinline!prov! translate to projections to the respective
fields.

We translate table declarations to pairs.  The first component is a
simple table declaration where all columns have their primitive
underlying non-provenance type.  We will use the underlying table
declaration for insert, update, and delete operations.  The second
component is essentially a delayed query that calculates \wpr for the
entire table.  (The fact that it is delayed is important here, because it
means that it can be inlined and simplified later, rather than loaded
into memory.)  We compute provenance for each record by iterating over
the table.  For every record of the input table, we construct a new
record with the same fields as the table.  For every column with
provenance, the field's value is a record with
\lstinline[language=Links]!data! and
\lstinline[language=Links]!prov!  fields.  The \lstinline[language=Links]!data! field is just the value.
The translation of table
references also uses an auxiliary operation $R \triangleright^n_x S$
which, given a row type $R$, a table name $n$, a variable $x$ and a
provenance specification $S$, constructs a record in which each
field contains data from $x$ along with the specified provenance (if any).
We wrap the iteration in an anonymous function to delay execution:
otherwise, the provenance-annotated table would be constructed in
memory when the table reference is first evaluated.  We will
eventually apply this function in a query, and the \Links query
normalizer will inline the provenance annotations and normalize them
along with the rest of the query.

We translate table comprehensions to comprehensions over 
the second component of a
translated table declaration.  Since that component is a function, we have
to apply it to a (unit) argument.

For example, recall the example query \lstinline!q1'''! from
Section~\ref{sec:overview}, Figure~\ref{fig:q1p3}. The table declaration translates as
follows:
\begin{lstlisting}[language=Links]
var agencies = (table "Agencies"
                    with (name: String, based_in: String, phone: String),
 fun () { for (t <-- table "Agencies"
                          with (name: String, based_in: String, phone: String))
   [(name = t.name, based_in = t.based_in,
     phone = (data = t.phone, prov = ("Agencies", "phone", t.oid)))] })
\end{lstlisting}
The translation of the \lstinline!externalTours! table reference is
similar, but simpler, since it has no \lstinline!prov!
annotations. The query translates to
\begin{lstlisting}[language=Links]
query {
  for (a <-- agencies.2())
    for (e <-- externalTours.2())
    where (a.name == e.name && e.type == "boat")
      [(name = e.name,
        phone = a.phone.data, p_phone = a.phone.prov)]
}
\end{lstlisting}
Moreover, after
inlining the adjusted definitions of \lstinline!agencies! and
\lstinline!externalTours! and normalizing, the provenance computations
in the delayed query \lstinline!agencies.2! are also inlined, resulting in the following SQL query.
In this query, the table and column part of the \wpr are in fact static, and the generated SQL query reflects this by using constants in the \lstinline[language=SQL]!select! clause.
We see no trace of function application, or nested record projections in the guise of \lstinline!data! and \lstinline!prov!.

\begin{lstlisting}[language=SQL]
select
  e.name as name,
  a.phone as phone,
  'agencies' as p_phone_1,
  'phone' as p_phone_2,
  a.oid as p_phone_3
from
  Agencies as a,
  ExternalTours as e
where
  a.name = e.name and e.type = 'boat'
\end{lstlisting}

The type-preservation correctness property of the where-provenance
translation is that it preserves well-formedness.  We first need 
\begin{lemma}\label{lem:where-helper}
Let $R$ be a row and $S$ be a provenance specification.  Then 
\begin{itemize}
\item $𝔚⟦\rcd{\erase{R}}⟧ = \rcd{R}$.
\item $\erase{\rcd{R \triangleright S}} = \rcd{R}$.
\end{itemize}
\end{lemma}
The type-preservation property for the translation is stated as follows and proved in \ref{thm:where-type-preservation-proof}:
\begin{theorem}\label{thm:where-type-preservation}~
  \begin{enumerate}
  \item For every \WLinks context $\Gamma$, term $M$, and type $A$, if $
    \Gamma \vdash_\WLinks M : A$ then $ 𝔚⟦\Gamma⟧ \vdash_\Links
    𝔚⟦M⟧ : 𝔚⟦A⟧$.
\item For every \WLinks context $\Gamma$, provenance specification
  $S$, row $R$ and subrow $R'$ such that $R'\triangleright^n_x S$ is defined, if $Γ ⊢ S : \ProvSpec(R)$ then 
$𝔚⟦\Gamma⟧,x{:}\rcd{R}   ⊢ \rcd{R' \triangleright^n_x  S} :  𝔚⟦\rcd{R' \triangleright S} ⟧ $.
  \end{enumerate}
\end{theorem}





We have shown that annotation-propagation in \WLinks is
color-propagating (Theorem~\ref{thm:where-correctness}) and that the
translation to \Links is type-preserving
(Theorem~\ref{thm:where-type-preservation}).  We have not, however,
shown that the translation correctly implements the semantics.  This
is intuitively clear, but a formal proof is nontrivial because a
single step in \WLinks can translate to multiple steps in \Links,
involving terms that have no \WLinks counterpart.

\subsection{Lineage}\label{sec:translation-lineage}

\lstset{language=LLinks}
We define a typed translation from \LLinks to \Links.  The translation
has two parts: an outer translation called \emph{doubling} (𝔇) and an inner
part called \emph{lineage translation} (𝔏).  The former is used for
translating ordinary \LLinks code while the latter is used to
translate query code inside a $\kw{lineage}$ keyword.  
The syntactic translation of \LLinks types for the doubling translation is shown in
Figure~\ref{fig:llinks-type-translation}, and the translation used for
the lineage translation is the $𝔏 $ translation shown earlier.
We write $𝔇⟦\Gamma⟧$ and
$𝔏⟦\Gamma⟧$ for the obvious extensions of these translations to
contexts.

The translation of \LLinks expressions to \Links is shown in
Figures~\ref{fig:lineage-outer-translation}--\ref{fig:lineage-term-translation}.  
Following the type translation, term translation operates in two modes: 𝔇 and
𝔏.  We translate ordinary \Links programs using the translation $𝔇⟦-⟧$.
When we reach a \lstinline!lineage! block, we switch to using the $𝔏⟦-⟧$
translation.  $𝔏⟦\coll{M}⟧$ provides initial lineage for list literals.  Their
lineage is simply empty.  Table comprehension is the most interesting
case.  We translate a table iteration
\lstinline!for ($x$ <-- $L$) $M$! to a nested list comprehension.  The
outer comprehension binds $y$ to the results of the lineage-computing
view of $L$.  The inner comprehension binds a fresh variable $z$,
iterating over $𝔏⟦M⟧$---the original comprehension body $M$
transformed using 𝔏.  The original comprehension body $M$ is defined
in terms of $x$, which is not bound in the transformed comprehension.
We therefore replace every occurrence of $x$ in $𝔏⟦e⟧$ by
\lstinline!$y$.data!.  In the body of the nested comprehension we thus
have $y$, referring to the table row annotated with lineage, and $z$,
referring to the result of the original comprehension's body, also
annotated with lineage.  As the result of our transformed
comprehension, we return the plain data part of $z$ as our data, and
the combined lineage annotations of $y$ and $z$ as our provenance.
(Handling \lstinline!where!-clauses is straightforward, as shown in
Figure~\ref{fig:lineage-inner-translation}.)

\begin{figure}[tb]
\small
\begin{eqnarray*}
  𝔇⟦O⟧ &=& O\\
  𝔇⟦A  \arr B⟧ &=& (𝔇⟦A⟧\arr 𝔇⟦B⟧, 𝔏⟦A⟧ \arr 𝔏⟦B ⟧)\\
  𝔇⟦\rcd{l_i:A_i}_{i=1}^n⟧ &=& \rcd{l_i:𝔇⟦A_i⟧}_{i=1}^n\\
  𝔇⟦\coll{A}⟧ &=& \coll{𝔇⟦A⟧}\\
  𝔇⟦\kw{table}(R)⟧ &=& (\kw{table}(R),() \arr 𝔏⟦[(R)] ⟧ )
\end{eqnarray*}
  \caption{Doubling translation}
\label{fig:llinks-doubling-translation}
\end{figure}

\begin{figure*}[tb]
\small
\begin{minipage}{12cm}
\begin{eqnarray*}
𝔇⟦c⟧ &=& c\\
𝔇⟦x⟧ &=& x\\
𝔇⟦(l_i = M_i)_{i=1}^n⟧ &=& (l_i=𝔇⟦M_i⟧)_{i=1}^n\\
𝔇⟦N.l⟧ &=& 𝔇⟦N⟧.l\\
𝔇⟦\kw{fun}(x_i|_{i=1}^n)~\{M\}⟧ &=&
                                    (\kw{fun}(x_i|_{i=1}^n)~\{𝔇⟦M⟧\}, 𝔏^*⟦\kw{fun}(x_i|_{i=1}^n)~\{M\}⟧)\\
𝔇⟦M(N_i|_{i=1}^n) ⟧ &=& 𝔇⟦M ⟧.1(𝔇⟦N_i⟧_{i=1}^n) \\
𝔇⟦\kw{var}~x=M;N ⟧ &=& \kw{var}~x = 𝔇⟦M ⟧; 𝔇⟦N⟧\\
𝔇⟦\coll{} ⟧  &=& \coll{}\\
𝔇⟦\coll{M} ⟧  &=& \coll{𝔇⟦M⟧ }\\
𝔇⟦M \plusplus N ⟧  &=& 𝔇⟦M⟧ \plusplus 𝔇⟦N⟧\\
𝔇⟦\kw{if}~(L)~\{M\}~\kw{else}~\{N\}⟧ &=& \kw{if}~(𝔇⟦ L⟧)~\{𝔇⟦M⟧\}~\kw{else}~\{𝔇⟦N⟧\}\\
𝔇⟦\kw{query}~\{M\}⟧ &=& \kw{query}~\{𝔇⟦M⟧\}\\
𝔇⟦\kw{empty}~(M)⟧ &=& \kw{empty}~(𝔇⟦M⟧)\\
𝔇⟦\kw{for}~(x \rarr L)~M⟧ &=& \kw{for}~(x \rarr 𝔇⟦L⟧) ~𝔇⟦M⟧ \\
𝔇⟦\kw{where}(M)~N⟧ &=& \kw{where}(𝔇⟦M⟧) ~𝔇⟦N⟧ \\
𝔇⟦\kw{for}~(x \drarr  L)~M⟧ &=&
\kw{for}~ (x \rarr 𝔇⟦L⟧.1) ~𝔇⟦M⟧\\
𝔇⟦\kw{insert} ~L~\kw{values}~M⟧ &=& \kw{insert} ~𝔇⟦ L⟧.1~\kw{values}~𝔇⟦ M⟧\\
𝔇⟦\kw{update} ~(x \rarr L)~\kw{where}~M~\kw{set}~𝔇⟦ N⟧ &=& \kw{update} ~(x \rarr 𝔇⟦L⟧.1)~\kw{where}~𝔇⟦M⟧~\kw{set}~N⟧\\
𝔇⟦\kw{delete} ~(x \rarr L)~\kw{where}~M⟧  &=&\kw{delete} ~(x
                                                    \rarr
                                                    𝔇⟦L⟧.1)~\kw{where}~𝔇⟦M⟧\\
𝔇⟦\kw{lineage}~\{M\}⟧ &=& \kw{query}~\{𝔏^*⟦M⟧\}
\end{eqnarray*}
\end{minipage}
\begin{eqnarray*}
𝔇⟦\kw{table}~n~\kw{with} ~(R)⟧
&=&
(\kw{table}~n~\kw{with} ~(R),\kw{ fun} () \{ 𝔏⟦\kw{table}~n~\kw{with} ~(R)⟧
  \})
\end{eqnarray*}
\caption{Translation of \LLinks to \Links: outer translation}
\label{fig:lineage-outer-translation}
\end{figure*}
\begin{figure*}[tb]
\small
\begin{minipage}{12cm}
\begin{eqnarray*}
𝔏⟦c ⟧ &=& c\\
𝔏⟦x⟧ &=& x\\
𝔏⟦(l_i = M_i)_{i=1}^n⟧ &=& (l_i=𝔏⟦M_i⟧)_{i=1}^n\\
𝔏⟦N.l⟧ &=& 𝔏⟦N⟧.l\\
𝔏⟦\kw{fun}(x_i|_{i=1}^n)~\{M\}⟧ &=&  (\kw{fun}(x_i|_{i=1}^n)~\{𝔏⟦M⟧\})\\
𝔏⟦M(N_i|_{i=1}^n) ⟧ &=& 𝔏⟦M ⟧(𝔏⟦N_i⟧|_{i=1}^n) \\
𝔏⟦\kw{var}~x=M;N ⟧ &=& \kw{var}~x = 𝔏⟦M ⟧; 𝔏⟦N⟧\\
𝔏⟦\coll{} ⟧  &=& \coll{}\\
𝔏⟦\coll{M} ⟧  &=& \coll{\rcd{\mathsf{data} = 𝔏⟦M⟧,\mathsf{prov} = \coll{}}}\\
𝔏⟦M \plusplus N ⟧  &=& 𝔏⟦M⟧ \plusplus 𝔏⟦N⟧
\\
𝔏⟦\kw{if}~(L)~\{M\}~\kw{else}~\{N\}⟧ &=& \kw{if}~(𝔏⟦ L⟧)~\{𝔏⟦M⟧\}~\kw{else}~\{𝔏⟦N⟧\}\\
𝔏⟦\kw{query}~\{M\}⟧ &=& \kw{query}~\{𝔏⟦M⟧\}\\
𝔏⟦\kw{empty}~(M)⟧ &=& \kw{empty}~(𝔏⟦M⟧)\\
𝔏⟦\kw{for}~(x \rarr L)~M⟧ &=&\begin{array}[t]{l}
\kw{for}~(y \rarr 𝔏⟦L⟧)\\
\quad \kw{for}~(z \rarr 𝔏⟦M⟧[x \mapsto y.\mathsf{data}])\\
\qquad [(\mathsf{data} = z.\mathsf{data}, \mathsf{prov} = y.\mathsf{prov} \plusplus z.\mathsf{prov})]
\end{array} \\
𝔏⟦\kw{where}(M)~N⟧ &=& \kw{where}(𝔏⟦M⟧)~(𝔏⟦N⟧)\\
 𝔏⟦\kw{for}~(x \drarr  L) ~M⟧ &=&
\begin{array}[t]{l}
\kw{for}~(y \rarr 𝔏⟦L⟧)\\
\quad \kw{for}~(z \rarr 𝔏⟦M⟧[x \mapsto y.\mathsf{data}])\\
\qquad [(\mathsf{data} = z.\mathsf{data}, \mathsf{prov} = y.\mathsf{prov} \plusplus z.\mathsf{prov})]
\end{array}
\\
𝔏⟦\kw{lineage}~\{M\}⟧ &=& \kw{query}~\{ 𝔏⟦M⟧\}
\end{eqnarray*}
\end{minipage}
\begin{eqnarray*}
𝔏⟦\kw{table}~n~\kw{with}~ (R)⟧ &=& \kw{for} (x \drarr \kw{table}~n~\kw{with} ~(R)) [(\mathsf{data} = x, \mathsf{prov} = [(n, x.\mathsf{oid})])]
\end{eqnarray*}
\caption{Translation of \LLinks to \Links: inner translation}
\label{fig:lineage-inner-translation}
\end{figure*}
\begin{figure*}[tb]
\small
\begin{eqnarray*}
𝔏^*⟦M⟧ &=& 𝔏⟦M⟧[x_i \mapsto  d2l⟦A_i⟧(x_i) |_{i=1}^n] \\
&& \text{ where
           $x_1:A_1,\ldots,x_n:A_n$ are the free variables of $M$}\\
d2l⟦A⟧ &:& 𝔇⟦A⟧ \arr 𝔏⟦A⟧\\
d2l⟦O⟧(x) &=& x\\
d2l⟦A \arr B⟧(f) &=& f.2\\
d2l⟦(l_1:A_1,\ldots,l_n:A_n)⟧(x) &=&
                                     (l_1:d2l⟦A_1⟧(x.l_1),\ldots,l_n:d2l⟦A_n⟧(x.l_n))\\
d2l⟦[A]⟧(y) &=& \kw{for}(x \rarr y)
                [(\mathsf{data}=d2l⟦A⟧(x),\mathsf{prov}=[])]\\
d2l⟦\kw{table}(R)⟧(t) &=& t.2()
\end{eqnarray*}
\caption{Translation of \LLinks to \Links: term translation}
\label{fig:lineage-term-translation}
\end{figure*}

One subtlety here is that lineage blocks need not be closed,
and so may refer to variables that were defined (and will be bound to
values at runtime) outside of the lineage block.  This could cause
problems: for example, if we bind $x$ to a collection $[1,2,3]$
outside a lineage block and refer to it in a comprehension inside such
a block, then uses of $x$ will expect the collection elements to be
records such as $\rcd{\mathsf{data}=1,\mathsf{prov}=L}$ rather than
plain numbers.  Therefore, such variables need to be adjusted so that
they will have appropriate structure to be used within a lineage
block.  The auxiliary type-indexed function $d2l⟦A⟧$ in Figure~\ref{fig:lineage-term-translation} accomplishes this by mapping a
value of type $𝔇⟦A⟧$ to one of type $𝔏⟦A⟧$.  We define $𝔏^*⟦-⟧$ as a
function that applies $𝔏⟦-⟧$ to its argument and substitutes all free
variables $x:A$ with $d2l⟦A⟧(x)$.

The $𝔇⟦-⟧$ translation also has to account for functions that are
defined outside lineage blocks but may be called either outside or
inside a lineage block.  To support this, the case for functions in
the $𝔇⟦-⟧$ translation creates a pair, whose first component is the
recursive $𝔇⟦-⟧$ translation of the function, and whose second
component uses the $𝔏^*⟦-⟧$ translation to create a version of the
function callable from within a lineage block.  (We use $𝔏^*⟦-⟧$
because functions also need not be closed.)  Function calls outside
lineage blocks are translated to project out the first component;
function calls inside such blocks are translated to project out the
second component (this is actually accomplished via the $A \arr B$
case of $d2l$.)

Finally, notice that the $𝔇⟦-⟧$ translation maps table types and table
references to pairs.  This is similar to the $𝔚⟦-⟧$ translation, so we
do not explain it in further detail; the main difference is that we
just use the $\mathsf{oid}$ field to assign default provenance to all
rows.






For example, if we wrap the query from Figure~\ref{fig:boat-tours} in a \lstinline!lineage! block it will be rewritten to this:
\begin{lstlisting}[language=Links]
for (y_a <- agencies.2())
  for (z_a <- for (y_e <- externalTours.2())
                  for (z_e <- [(data = (name = y_a.data.name, phone = y_a.data.phone),
                                   prov = [])])
                  where (y_a.data.name == y_e.data.name && y_e.data.type == "boat")
                    [(data = z_e.data, prov = y_e.prov ++ z_e.prov)])
    [(data = z_a.data, prov = y_a.prov ++ z_a.prov)]
\end{lstlisting}
Once \lstinline!agencies! and \lstinline!externalTours!
are inlined, \Links's built-in normalization algorithm simplifies this query
to:
\begin{lstlisting}[language=Links]
for (y_a <- table "Agencies" with ...)
  for (y_e <- table "ExternalTours" with ...)
  where (y_a.data.name == y_e.data.name && y_e.data.type == "boat")
    [(data = (name = y_a.data.name,phone = y_a.data.phone),
      prov = [("Agencies", y_a.oid), ("ExternalTours",y_e.oid)])]
\end{lstlisting}

Before considering the main type-preservation result, we state some
auxiliary lemmas with corresponding proofs in \ref{lem:lin-helpers-proof}:
\begin{lemma}\label{lem:lin-helpers}
\begin{enumerate}
\item If $A :: \QueryType$ then  $𝔇⟦A⟧ = 𝔇⟦𝔏⟦A⟧⟧$.
\item If $\Gamma \vdash M : 𝔇⟦A⟧$ then $\Gamma\vdash d2l(M) : 𝔏⟦A⟧$.
\end{enumerate}
\end{lemma}
 The type-preservation property for the translation from \LLinks to \Links is
stated as follows:
\begin{theorem}\label{thm:lin-type-preservation}
  Let $M$ be given such that $\Gamma \vdash_\LLinks M : A$.  Then:
  \begin{enumerate}
 \item $𝔏⟦\Gamma⟧ \vdash_\Links 𝔏⟦M⟧ : 𝔏⟦A⟧ $
 \item $𝔇⟦\Gamma⟧ \vdash_\Links 𝔏^*⟦M⟧ : 𝔏⟦A⟧$
  \item $𝔇⟦\Gamma⟧ \vdash_\Links 𝔇⟦M⟧ : 𝔇⟦A⟧$
  \end{enumerate}
\end{theorem}
\begin{proof}
The proof of the first part is by induction on the structure of typing
derivations.  The interesting cases are for the \textsc{List},
\textsc{ForList} and \textsc{ForTable} cases, where lineage
annotations are created or propagated.  The detailed derivations are
given in \ref{thm:lin-type-preservation-proof}.

For the second part, suppose $\Gamma \vdash M : A$.  Then by part 1 we
know $𝔏⟦\Gamma⟧ \vdash 𝔏⟦M⟧ : 𝔏⟦A⟧$. Clearly, for each $x_i :A_i$ in
$\Gamma$ we have $ \DD{\Gamma} \vdash x_i:\DD{A_i}$, so it follows
that $ \DD{\Gamma} \vdash d2l(x_i) : \LL{A_i}$ for each $i$ by Lemma~\ref{lem:lin-helpers}(2).  Using the (standard)
substitution lemma for \Links typing, we can conclude $\DD{\Gamma}
\vdash \LLstar{M} : \LL{A}$.

Finally, for the third part, again the proof is by induction on the
structure of the derivation of $\Gamma \vdash M : A$.  Most cases are
straightforward; we show a few representative cases for
(single-argument) functions and the $\kw{lineage}$ keyword,
illustrating the need for duplicating code in the type translation for
functions and the use of $\LLstar{-}$.  The cases for updates and
table references are similar to those for \WLinks, but simpler because
the types of the fields do not change in the translation from \LLinks
to \Links.  We illustrate the case for translation of functions, since
it is one of the subtler cases; the cases for function application and
the $\kw{lineage}$ keyword are given in the appendix.
If the derivation is of the form:
\[
\infer[Fun]
{Γ , x : A⊢ M : B}
{Γ ⊢ \kw{fun }\rcd{ x} \braces{M}: A \arr B}
\]
then by induction we have $\DD{Γ} ,x : \DD{A} ⊢ \DD{M}
: \DD{B}$ and by part 2 we know that $\DD{\Gamma} \vdash \LLstar{\kw{fun }\rcd{ x } \braces{M}}
:  \LL{\rcd{A } \arr B}$.  We can proceed as follows:

\[\begin{array}{ll}
\DD{Γ} , x:\DD{A} ⊢ \DD{M}: \DD{B} 
&\text{by IH}\\
\DD{Γ} ⊢ \kw{fun }\rcd{ x}
  \braces{\DD{M}}: \DD{ A} \arr \DD{B}
&\text{by rule}\\
 \DD{\Gamma} \vdash \LLstar{\kw{fun }\rcd{ {x} } \braces{M}}
  :  \LL{A} \arr \LL{B}
& \text{by part 2}\\
\DD{Γ} ⊢ (\kw{fun }\rcd{x }
  \braces{\DD{M}},\LLstar{\kw{fun }\rcd{ x }
    \braces{M}}): \DD{A \arr B}
&\text{by rule}
\end{array}
\]
where the final step relies on the fact that $\DD{A \arr B} = (\DD{A}
\arr \DD{B}, \LL{A} \arr \LL{B})$.
\end{proof}



As with the where-provenance translation, we have proven the
correctness of lineage annotation propagation
(Theorem~\ref{thm:lin-correctness}) and type-preservation of the
translation (Theorem~\ref{thm:lin-type-preservation}).
The latter is a partial sanity check, but no proof, that this translation faithfully implements the semantics.

\section{Experimental Evaluation}
\label{sec:evaluation}

We implemented the two variants of \Links with language\hyp{}integrated
provenance, \WLinks and \LLinks, featuring our extensions for \wpr and
lineage, respectively.
In this section we compare them against plain \Links on a number of queries to determine their overhead.
We also compare both variants against \emph{Perm}, a database-integrated provenance system.

Both provenance variants of \Links build on its query shredding
capabilities as described by \citet{SIGMOD2014CheneyLW}.  They used
queries against a simple test database schema (see
Figure~\ref{fig:schema}) that models an organization with departments,
employees and external contacts.  We adapt some of their benchmarks
to return \wpr and lineage and compare against the same queries
without provenance.

Unlike \citet{SIGMOD2014CheneyLW} our database does not include an additional \lstinline!id! field, instead we use PostgreSQL's OIDs, which are used for identification of rows in \wpr and lineage.
We populate the databases at varying sizes using randomly generated data in the same way \citet{SIGMOD2014CheneyLW} describe it:
``We vary the number of departments in the organization from 4 to 4096
(by powers of 2). Each department has on average 100 employees and
each employee has 0–2 tasks.''
The largest database, with 4096 departments, is 142~MB on disk when exported by \texttt{pg\_dump} to a SQL file (excluding OIDs).
We create additional indices on \lstinline!tasks(employee)!, \lstinline!tasks(task)!, \lstinline!employees(dept)!, and \lstinline!contacts(dept)!.

All tests were performed on an otherwise idle desktop system with a 3.2~GHz quad-core CPU, 8~GB RAM, and a 500~GB HDD.
The system ran Linux (kernel 4.5.0) and we used PostgreSQL 9.4.2 as the database engine.
\Links and its variants \WLinks and \LLinks are interpreters written in OCaml, which were compiled to native code using OCaml 4.02.3.
The exact versions of \WLinks and \LLinks used for this set of benchmarks can be downloaded from \url{https://www.inf.ed.ac.uk/research/isdd/admin/package?download=188} and \url{https://www.inf.ed.ac.uk/research/isdd/admin/package?download=189} respectively.

\subsection{Where-provenance}

To be usable in practice, \wpr should not have unreasonable runtime overhead.
We compare queries \emph{without} any \wpr against queries that calculate \wpr on \emph{some} of the result and queries that calculate \emph{full} \wpr wherever possible.
This should give us an idea of the overhead of \wpr on typical queries, which are somewhere in between full and no provenance.

The nature of \wpr suggests two hypotheses:
First, we expect the asymptotic cost of \wpr-annotated queries to be the same as that of regular queries.
Second, since every single piece of data is annotated with a triple,
we expect the runtime of a fully \wpr-annotated query to be at most four times the runtime of an unannotated query just for handling more data.


We only benchmark \emph{default} \wpr, that is table name, column name, and the database-generated OID for row identification.
External provenance is computed by user-defined database-executable functions and can thus be arbitrarily expensive.


\begin{figure}
\begin{lstlisting}[language=Links,style=infigure]
table departments with (oid: Int, name: String)
table employees with (oid: Int, dept: String, name: String, salary: Int)
table tasks with (oid: Int, employee: String, task: String)
table contacts with (oid: Int, dept: String, name: String, client: Bool)
\end{lstlisting}
\caption{Benchmark database schema, cf. \citet{SIGMOD2014CheneyLW}.}
\label{fig:schema}
\vspace{1cm}
\begin{lstlisting}[language=WLinks,style=infigure]
# Q1 : [(contacts: [("client": Prov(Bool), name: Prov(String))], ...
for (d <-- departments)
  [(contacts = contactsOfDept(d),
    employees = employeesOfDept(d),
    name = d.name)]

# Q2 : [(d: Prov(String))]
for (d <- q1())
where (all(d.employees, fun (e) {
      contains(map(fun (x) { data x }, e.tasks), "abstract") }))
  [(d = d.name)]

# Q3 : [(b: [Prov(String)], e: Prov(String))]
for (e <-- employees)
  [(b = tasksOfEmp(e), e = e.name)]

# Q4 : [(dpt:Prov(String), emps:[Prov(String)])]
for (d <-- departments)
  [(dpt = (d.name),
    emps = for (e <-- employees)
              where ((data d.name) == (data e.dept))
                [(e.name)])]

# Q5 : [(a: Prov(String), b: [(name: Prov(String), ...
for (t <-- tasks)
  [(a = t.task, b = employeesByTask(t))]

# Q6 : [(d: Prov(String), p: [(name: Prov(String), tasks: [String])])]
for (x <- q1())
  [(d = x.name,
    p = get(outliers(x.employees),
            fun (y) { map(fun (z) { data z }, y.tasks) }) ++
        get(clients(x.contacts), fun (y) { ["buy"] }))]
\end{lstlisting}
\caption{``allprov'' benchmark queries used in experiments}
\label{fig:full-where}
\end{figure}

We use the queries with nested results from \citet{SIGMOD2014CheneyLW} and use them unchanged for comparison with the two variants with varying amounts of \wpr.

For \emph{full} \wpr we change the table declarations to add provenance to every field, except the OID.
The full declarations can be found in Figure~\ref{fig:table-declarations-where}.
This changes the types, so we have to adapt the queries and some of
the helper functions used inside the queries, see Figure~\ref{fig:prov-helpers}.
Figure~\ref{fig:full-where} shows the benchmark queries with full
provenance.
See \ref{appendix:benchmark-code} for the full code, including table declarations and helper functions.
Note that for example query Q2 maps the \lstinline[language=WLinks]!data! keyword over the employees tasks before comparing the tasks against \lstinline!"abstract"!.
Query Q6 returns the outliers in terms of salary and their tasks, concatenated with the clients, who are assigned the \emph{fake} task \lstinline!"buy"!.
Since the fake task is not a database value it cannot have \wpr.
\WLinks type system prevents us from pretending it does.
Thus, the list of tasks has type \lstinline![String]!, not \lstinline![Prov(String)]!.

The queries with \emph{some} \wpr are derived from the queries with full provenance.
Query Q1 drops provenance from the contacts' fields.
Q2 returns data and provenance separately.
It does not actually return less information, it is just less type-safe.
Q3 drops provenance from the employee.
Q4 returns the employees' provenance only, and drops the actual data.
Q5 does not return provenance on the employees fields.
Q6 drops provenance on the department.
(These queries make use of some auxiliary functions which
are included in the appendix.)

\noindent\textbf{Setup.}
We have three \WLinks programs, one for each level of \wpr annotations.
For each database size, we drop all tables and load a dump from disk, starting with 4096.
We then run \WLinks three times, once for each program in order \emph{all}, \emph{some}, \emph{none}.
Each of the three programs performs and times its queries 5 times in a row and reports the median runtime in milliseconds.
The programs measure runtime using the \WLinks built-in function \lstinline!serverTimeMilliseconds! which in turn uses OCaml's \lstinline!Unix.gettimeofday!.

\begin{figure}
\centering{\sf \include{graph} }

\caption{Where-provenance query runtimes.}
\label{fig:wpr-graph}

\begin{center}
  \begin{tabular}{lrrr|r}
    \toprule
    Query & \multicolumn{3}{c}{median runtime$^*$ in ms} &  overall slowdown\\
          & allprov & someprov & noprov &  (geom mean)\\
    \midrule
    Q1 & 6068 & 3653 & 1763 & 2.26 \\
    Q2 & 60 & 60 & 60 & 1.52 \\
    Q3 & 8100 & 8064 & 4497 & 1.88 \\
    Q4 & 1502 & 1214 & 573 & 2.8 \\
    Q5 & 6778 & 3457 & 2832 & 1.85 \\
    Q6 & 17874 & 18092 & 16716 & 1.22 \\
    \bottomrule
  \end{tabular}\vspace{.5em}
\end{center}
\caption{Median runtimes for largest dataset (Q1 at 512 departments,
  Q5 at 1024 departments, Q6 at 2048 departments, others at 4096
  departments) and geometric means of overall slowdowns.}
\label{fig:wpr-slowdown}
\end{figure}

\noindent\textbf{Data.}
Figure~\ref{fig:wpr-graph} shows our experimental results.
We have one plot for every query, showing the database size on the x-axis and the median runtime over five runs on the y-axis.
Note that both axes are logarithmic.
Measurements of full \wpr are in black circles, no provenance are yellow triangles, some provenance is blue squares.
Based on test runs we had to exclude some results for queries at larger
database sizes because the queries returned results that were too
large for \Links to construct as in-memory values.

The graph for query Q2 looks a bit odd.
This seems to be due to Q2 not actually returning any data for some database sizes, because for some of the (randomly generated) instances there just are no departments where all employees have the task \lstinline!"abstract"!.

The table in Figure~\ref{fig:wpr-slowdown} lists all queries with their median runtimes with full, some, and no provenance.
The time reported is in milliseconds, for the largest database instance that both variants of a query ran on.
For most queries this is 4096; for Q1 it is 512, 1024 for Q5, and 2048 for Q6.
Figure~\ref{fig:wpr-slowdown} also reports the slowdown of full \wpr versus no provenance as the geometric mean across all database sizes, for each query.
The slowdown ranges from 1.22 for query Q6 up to 2.8 for query Q4.
Note that query Q2 has the same runtime for all variants at 4096 departments, but full provenance is slower for some database sizes, so the overall slowdown is $>1$.


\noindent\textbf{Interpretation.}
The graphs suggest that the asymptotic cost of all three variants is the same, confirming our hypothesis.
This was expected, anything else would have suggested a bug in our implementation.

The multiplicative overhead seems to be larger for queries that return more data.
Notably, for query Q2, which returns no data at all on some of our test database instances, the overhead is hardly visible.
The raw amount of data returned for the full \wpr queries is three to four times that of a plain query.
Most strings are short names and provenance adds two short strings and a number for table, column, and row.
The largest overhead is 2.8 for query Q4, which exceeds our expectations due to just raw additional data needing to be processed.

\subsection{Lineage}
We expect lineage to have different performance characteristics than
\wpr.  Unlike \wpr, lineage is conceptually set valued.  A query with
few actual results could have huge lineage, because lineage is
combined for equal data.  In practice, due to \Links using multiset
semantics for queries, the amount of lineage is bounded by the shape of
the query.  Thus, we expect lineage queries to have the same
asymptotic cost as queries without lineage.  However, the
lineage translation still replaces single comprehensions by nested
comprehensions that combine lineage.  We expect this to have a larger
impact on performance than \wpr, where we only needed to trace more
data through a query.

\begin{figure}
\begin{lstlisting}[language=LLinks,style=infigure]
typename Lin(a) = (data: a, prov: [(String, Int)]);

# AQ6 : [Lin((department: String, outliers: [Lin((name: String, ...
for (d <- for (d <-- departments)
            [(employees = for (e <-- employees) where (d.name == e.dept)
                                [(name = e.name, salary = e.salary)],
              name = d.name)])
  [(department = d.name,
    outliers = for (o <- d.employees) where (o.salary > 1000000 || o.salary < 1000) [o])]

# Q3 : [Lin((b: [Lin(String)]), e: String)]
for (e <-- employees) [(b = tasksOfEmp(e), e = e.name)]

# Q4 : [Lin((dpt: String, emps: [Lin(String)]))]
for (d <-- departments)
  [(dpt = d.name,
    emps = for (e <-- employees) where (d.name == e.dept) [e.name])]

# Q5 : [Lin((a: String, b: [Lin((name: String, salary: Int, ...
for (t <-- tasks) [(a = t.task, b = employeesByTask(t))]

# Q6N : [Lin((department: String, people:[Lin((name: String, ...
for (x <-- departments)
  [(department = x.name,
    people = (for (y <-- employees)
                 where (x.name == y.dept && (y.salary < 1000 || y.salary > 1000000))
                    [(name = y.name,
                      tasks = for (z <-- tasks) where (z.employee == y.name) [z.task])])
           ++  (for (y <-- contacts) where (x.name == y.dept && y."client")
                    [(name = y.dept, tasks = ["buy"])]))]

# Q7 : [Lin((department: String, employee: (name: String, ...
for (d <-- departments) for (e <-- employees)
where (d.name == e.dept && e.salary > 1000000 || e.salary < 1000)
  [(employee = (name = e.name, salary = e.salary), department = d.name)]

# QC4 : [Lin((a: String, b: String, c: [Lin((doer: String, ...
for (x <-- employees) for (y <-- employees)
where (x.dept == y.dept && x.name <> y.name)
  [(a = x.name, b = y.name,
    c = (for (t <-- tasks) where (x.name == t.employee) [(doer = "a", task = t.task)])
     ++ (for (t <-- tasks) where (y.name == t.employee) [(doer = "b", task = t.task)]))]

# QF3 : [Lin((String, String))]
for (e1 <-- employees) for (e2 <-- employees)
where (e1.dept == e2.dept && e1.salary == e2.salary && e1.name <> e2.name)
  [(e1.name, e2.name)]

# QF4 : [Lin(String)]
(for (t <-- tasks) where (t.task == "abstract") [t.employee]) ++
(for (e <-- employees) where (e.salary > 50000) [e.name])
\end{lstlisting}
\caption{Lineage queries used in experiments}
\label{fig:lineage-queries}
\end{figure}

\begin{figure}
\centering
{\sf \include{lingraph} }
\caption{Lineage query runtimes.}\label{fig:lin-graph}
\begin{center}
  \begin{tabular}{lrrr}
    \toprule
    Query & \multicolumn{2}{c}{median runtime in ms} & overall slowdown\\
          & lineage & nolineage &  (geom mean)\\
    \midrule
    AQ6 & 493 & 108 & 3.8\\
    Q3 & 4234 & 969 &3.76\\
    Q4 & 1208 & 125 &7.55\\
    Q5 & 13662 & 11851 &1.25\\
    Q6N & 15200 & 7872 & 2.38\\
    Q7 & 16766 & 1283 &4.17\\
    QC4 & 13291 & 4021 &1.53\\
    QF3 & 22298 & 2412 &6.71\\
    QF4 & 682 & 73 &6.49\\
    \bottomrule
  \end{tabular}\vspace{.5em}
\end{center}
\caption{Median runtimes at largest dataset (Q7 at 128 departments,
  QC4 at 16 departments, QF3 at 512 departments, others at 1024
  departments) and geometric means of overall slowdowns}
\label{fig:lin-slowdown}
\end{figure}

\lstset{language=LLinks}

Figure~\ref{fig:lineage-queries} lists the queries used in the lineage
experiments.  For lineage, queries are wrapped in a
\lstinline[language=LLinks]!lineage! block.  Our implementation does
not currently handle function calls in lineage blocks automatically,
so in our experiments we have manually written lineage-enabled
versions of the functions \lstinline!employeesByTask! and \lstinline!tasksOfEmp!, whose bodies are wrapped in a
\lstinline!lineage! block.  
We reuse some of the queries from the \wpr experiments, namely Q3, Q4, and Q5.
Queries AQ6, Q6N, and Q7 are inspired by query Q6, but not quite the same.
Queries QF3 and QF4 are two of the flat queries from \citet{SIGMOD2014CheneyLW}.
Query QC4 computes pairs of employees in the same department and their
tasks in a ``tagged union''.  Again, these queries employ some helper
functions which are included in an appendix.

We use a similar experimental setup to the one for \wpr.
We only use databases up to 1024 departments, because most of the queries are a lot more expensive.
Query QC4 has excessive runtime even for very small databases.
Query Q7 ran out of memory for larger databases.
We excluded them from runs on larger databases.

\noindent\textbf{Data.}
Figure~\ref{fig:lin-graph} shows our lineage experiment results.
Again, we have one plot for every query, showing the database size on the x-axis and the median runtime over five runs on the y-axis.
Both axes are logarithmic.
Measurements with lineage are in black circles, no lineage is shown as yellow triangles.

The table in Figure~\ref{fig:lin-slowdown} lists queries and their median runtimes with and without lineage.
The time reported is in milliseconds, for the largest database instance that both variants of a query ran on.
For most queries this is 1024; for Q7 it is 128, 16 for QC4, and 512 for QF3.
The table also reports the slowdown of lineage versus no lineage as the geometric mean over all database sizes.
(We exclude database size 4 for the mean slowdown in QF4 which reported taking 0 ms for no lineage queries which would make the geometric mean infinity.)
The performance penalty for using lineage ranges from query Q5 needing a quarter more time to query Q4 being more than 7 times slower than its counterpart.

\noindent\textbf{Interpretation.}
Due to \Links multiset semantics, we do not expect lineage to cause an
asymptotic complexity increase.
The experiments confirm this.
Lineage is still somewhat expensive to compute, with slowdowns ranging
from 1.25 to more than 7 times slower.
Further investigation of
the SQL queries generated by shredding is needed.

\subsection{Threats to validity}
Our test databases are only moderately sized.
However, our result sets are relatively large.
Query Q1 for example returns the whole database in a different shape.
\Links' runtime representation of values in general and database results in particular has a large memory overhead.
In practice, for large databases we should avoid holding the whole result in memory.
This should reduce the overhead (in terms of memory) of provenance significantly.
(It is not entirely clear how to do this in the presence of nested results and thus query shredding.)
In general, it looks like the overhead of provenance is dependent on the amount of data returned.
It would be good to investigate this more thoroughly.
Also, it could be advantageous to represent provenance in a special way.
In theory, we could store the relation and column name in a more
compact way, for example.

One of the envisioned main use cases of provenance is debugging.
Typically, a user would filter a query anyway to pin down a problem
and thus only look at a small number of results and thus also query
less provenance.  Our experiments do not measure this scenario but
instead compute provenance for all query results eagerly.  Thus, the
slowdown factors we showed represent worst case upper bounds that may
not be experienced in common usage patterns.

Our measurements do not include program rewriting time.
However, this time is only dependent on the lexical size of the
program and is thus fairly small and, most importantly, independent of the database size.
Since \Links is interpreted, it does not really make sense to distinguish translation time from execution time, but both the \wpr translation and the lineage translation could happen at compile time, leaving only slightly larger expressions to be normalized at runtime.
Across the queries above, the largest observed time spent rewriting \WLinks or \LLinks to plain \Links was 5 milliseconds with the arithmetic mean coming to 0.5 milliseconds.









\subsection{Comparison with Perm}
In this section we compare \WLinks and \LLinks to \emph{Perm}~\cite{ICDE2009GlavicA}, as an instance of a database-integrated provenance system.
This is very much a comparison between apples and oranges.

The subset of queries supported by both \Links variants and \emph{Perm} is limited.
Most of the queries above use nested results which are not supported by \emph{Perm}.
Many common flat relational queries use aggregations which are not supported by \Links.
Others do not have large or interesting provenance annotations, be it where-provenance or lineage.

For this comparison we use a synthetic dataset.
We create tables of integers $1, \hdots, n$ for $n=(10000, 100000, 1000000)$; a simple string representation of the number; an English language cardinal like ``one'', ``two'', \dots; and an English language ordinal (``first'', ``second'', \dots).

\vspace{0.6em}
\begin{tabular}{rlll}
  \toprule
  i & s & cardinal & ordinal \\
  \midrule
  1 & "1" & "one" & "first" \\
  2 & "2" & "two" & "second" \\
  $\vdots$ \\
  $n$ & "$n$" & "$en$" & "$n$th" \\
  \bottomrule
\end{tabular}\vspace{0.6em}

We create 64 copies of these tables at each size $n$ and call them i\_s\_c\_o\_$n$\_1, i\_s\_c\_o\_$n$\_2, \dots.
Their content is the same, but their OIDs are distinct.
The data loading scripts are 55\,MB, 640\,MB, and 7.8\,GB on disk.

We use the same machine as before to run both databases and database clients.
We use \emph{Perm} version 0.1.1, which is a fork of Postgres 8.3 which adds support for provenance.
We compiled from source, which required passing \lstinline!-fno-aggressive-loop-optimizations! to GCC 6.3.1 as it would otherwise miscompile.
This seems to be a known problem with Postgres 8.3, which \emph{Perm} 0.1.1 is based on.
\Links uses the current version of Postgres as its database backend, which is Postgres 9.6.3.

In this set of benchmarks, we measure wall clock time of single runs.
\Links queries execute the query and print the result to stdout which is ignored.
Printing uses \Links's native format with pretty printing (line breaks and indentation) disabled.
\emph{Perm} queries are executed using \lstinline!psql! with a ``harness'' like this:
\begin{lstlisting}[language=SQL]
\COPY (SQL query goes here) TO STDOUT WITH CSV
\end{lstlisting}

\subsubsection{Where-provenance}
We use a family of queries that join $m=(16, 32, 64)$ of the tables described above on their integer column and select the provenance-annotated cardinal column for each of them.
Thus, the \wpr \WLinks queries look like this (table declarations are in \ref{appendix:benchmark-code-perm}):
\begin{lstlisting}[language=WLinks,morekeywords={database,tablekeys}]
query { 
  for (t_1 <-- i_s_c_o_$n$_1) $\hdots$ for (t_$m$ <-- i_s_c_o_$n$_$m$)
  where (mod(t_1.i, 100) < 5 && t_1.i == t_2.i && $\hdots$ && t_1.i = t_$m$.i)
    [(c1 = t_1.cardinal, c2 = t_2.cardinal, $\hdots$, c$m$ = t_$m$.cardinal)] }
\end{lstlisting}

Testing revealed that \WLinks runs out of memory for the largest (n=1000000,m=64) query.
Rather than using smaller input databases, we filtered the result using \lstinline!mod(t_1.i, 100) < 5! as an additional condition in the where clause.

Unfortunately, \emph{Perm's} \wpr support is too restrictive and refuses to execute an equivalent query with the following error message: ``WHERE-CS only supports conjunctive equality comparisons in WHERE clause.''
Fortunately, \emph{Perm} has no problems computing the full result, so we used queries of the following form, \emph{without} filtering based on \lstinline!t_1.i % 100 < 5!.

\begin{lstlisting}[language=SQL,escapechar={\#},morekeywords={PROVENANCE,ON,CONTRIBUTION}]
SELECT PROVENANCE ON CONTRIBUTION (WHERE)
       t_1.cardinal AS c1, $\hdots$, t_$m$.cardinal AS c$m$
FROM i_s_c_o_$n$_1 AS t_1, $\hdots$, i_s_c_o_$n$_$m$ AS t_$m$
WHERE t_1.i = t_2.i AND $\hdots$ AND t_1.i = t_$m$.i
\end{lstlisting}

We execute variants without \wpr of both the \WLinks and \emph{Perm} queries.
For \WLinks we keep the table declarations as they are, but use the \lstinline[morekeywords=data]!data! keyword to project to just the data and rely on query normalization to not compute provenance.
We run a fifth set of queries against Postgres 9.6.3 which are just like the plain \emph{Perm} queries, but \emph{with} filtering, like the \WLinks queries.

\begin{figure}
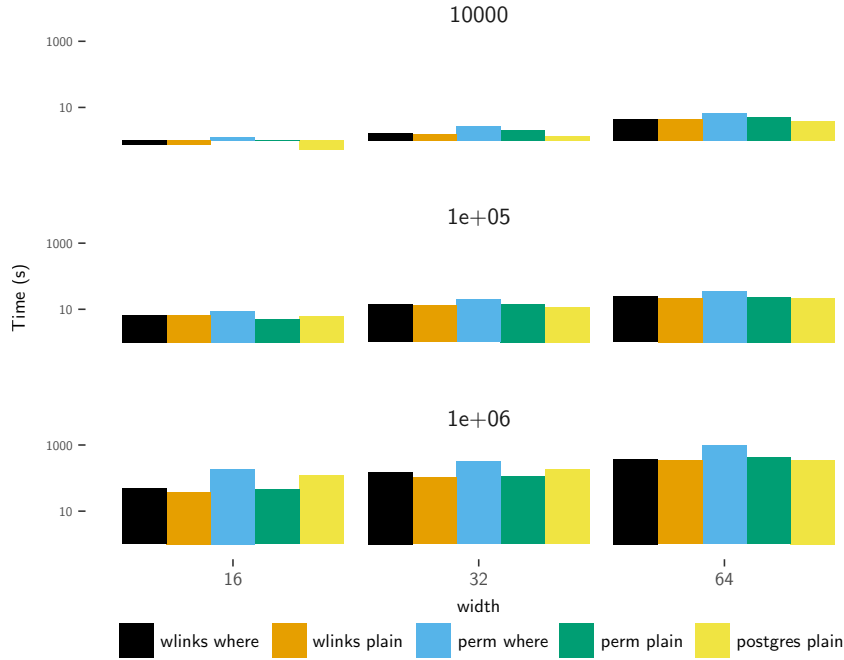

\centering{\sf \include{where_group_n} }
\caption{Where-provenance times grouped by table size ($n$) and number of tables ($m$).
Note that wlinks and postgres queries are filtered, perm queries are not.}
\label{fig:where-group-n}
\end{figure}

Figure~\ref{fig:where-group-n} shows query runtimes in seconds grouped by size of tables ($n$) and number of tables joined ($m$).
Keep in mind that the \emph{Perm} variants return a lot more data.
In the table below we show result size in megabytes at $n=1000000$ for \WLinks with \wpr annotations, \emph{Perm} with annotations, and Postgres without annotations.
We measure the size simply as byte count of the printed result.
Examples of the output can be found in \ref{appendix:benchmark-code-perm}.

\vspace{0.6em}
\begin{tabular}{l|rrr}
  \toprule
              & m=16       & m=32       & m=64       \\
  \midrule
  \WLinks     & 89.2\,MB   & 179.1\,MB  & 359.1\,MB  \\
  \emph{Perm} & 1589.3\,MB & 3187.5\,MB & 6384.0\,MB \\
  Postgres    & 37.2\,MB   & 74.3\,MB   & 148.6\,MB  \\
  \bottomrule
\end{tabular}\vspace{0.6em}

Looking at the runtime difference between the \emph{Perm} queries without \wpr and the plain Postgres queries we see that the result size does not have a great impact on runtime.
In general, the numbers between systems are hard to compare, not just because of result size.
We only consider one family of highly synthetic queries and the experimental setup is not necessarily a realistic reflection of any real-world use.
However, we do observe some trends:
The runtime difference between processing 10x data (going down one row in the graph) is larger than the difference between systems, by far.
Doubling the number of tables considered also dominates difference between systems.
We conclude that the overhead of \wpr in both \emph{Perm} and \WLinks is moderate and the systems are roughly comparable.

\subsubsection{Lineage}
We use the same data as before and similar queries to compare \LLinks to \emph{Perm Influence Contribution Semantics} (PI-CS).
Lineage and PI-CS are not equivalent in general~\cite{G10a}, but for the queries we use here the annotations contain, more or less, the same information.

We use a family of queries similar to those for where-provenance.
Again we join $m=(16, 32, 64)$ tables, but this time we return only the first table's integer and English cardinal columns, and their lineage.
The number of joins is particularly interesting here because it increases the size of the provenance metadata without affecting the actual result size.

We run variants with lineage and PI-CS metadata, as well as just the plain queries.
Finally, we run the plain version of the \emph{Perm} query against the Postgres database used by \LLinks.
This time all variants, including \emph{Perm}, are filtered to 5\% of the result size, as seen below.
The \LLinks query and example output can be found in \ref{appendix:benchmark-code-perm}.





\begin{lstlisting}[language=SQL,escapechar={\#},morekeywords={PROVENANCE}]
SELECT PROVENANCE t_1.i, t_1.cardinal
FROM i_s_c_o_$n$_1 AS t_1, $\hdots$, i_s_c_o_$n$_$m$ AS t_$m$
WHERE t_1.i % 100 < 5 AND t_1.i = t_2.i AND $\hdots$ AND t_$(m-1)$.i = t_$m$.i
\end{lstlisting}

Instead of a list of annotations per result row, \emph{Perm} produces wider tables, adding columns to identify join partners.
Table rows are identified by the whole width, so for $m=64$ joined tables we have two columns for the actual result and $64 * 4$ columns of provenance metadata.
The example result below is transposed.

\vspace{0.6em}
\begin{tabular}{r|rrc}
  \toprule
  i&1&2&$\hdots$\\
  cardinal&one&two&$\hdots$\\
  \midrule
  prov\_public\_i\_s\_c\_o\_1000\_1\_i &1&2&$\hdots$\\
  prov\_public\_i\_s\_c\_o\_1000\_1\_s &1&2&$\hdots$\\
  prov\_public\_i\_s\_c\_o\_1000\_1\_cardinal &one&two&$\hdots$\\
  prov\_public\_i\_s\_c\_o\_1000\_1\_ordinal &first&second&$\hdots$\\
  $\vdots$\\
  \bottomrule
\end{tabular}\vspace{0.6em}

\begin{figure}
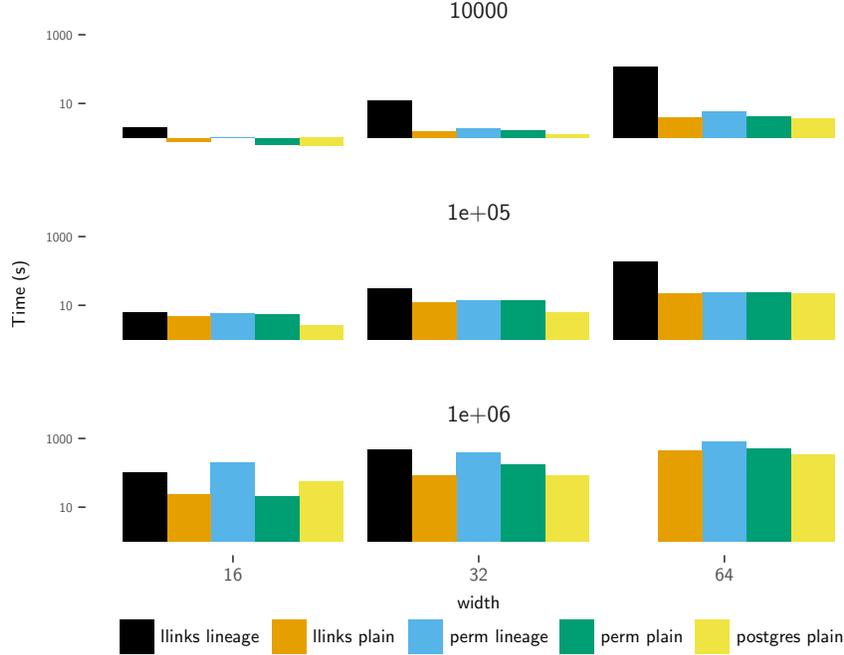

\centering{\sf \include{lineage_group_n} }
\caption{Lineage times grouped by relation size ($n$) and width ($m$). All queries are filtered to return only 5\% of results.}
\label{fig:lineage-group-n}
\end{figure}

We show query runtimes grouped by size of the tables ($n$) and number of tables joined ($m$) in Figure~\ref{fig:lineage-group-n}.
We omitted the largest \LLinks query (n=1000000, m=64); it ran for 33745 seconds, which would have distorted the graph too much.
This query just barely did not run out of memory, causing severe GC thrashing and leaving little memory for the database server and disk caches.

These timings are whole program execution and so include pre- and postprocessing steps.
\LLinks is translated to plain \Links, as described in Section~\ref{sec:translation-lineage}, which took less than 1 millisecond for all queries.
Query normalization for the lineage queries takes around 9 milliseconds for m=16, 41 milliseconds for m=32, and 194 milliseconds for m=64.
Postprocessing times (with data already in memory) range from almost 10 seconds for the lineage query at n=1000000, m=64 to 11 milliseconds for n=10000, m=16.

The queries executed by Postgres are on average a bit faster than the same queries executed by \emph{Perm}.
We did not investigate this further, a simple explanation would be that Postgres 9.6.3 is just a bit faster than Postgres 8.3 which is the version \emph{Perm} was forked from.

Below we show result size at $n=1000000$ for plain queries, and lineage queries at $m=16$ and $m=32$.
We measure the size simply as byte count of the printed result.
In some ways, the data is a worst case for \emph{Perm}, because the width of the result is so much smaller than the width of the annotations.
We can see this clearly in the result size table above.
Despite that, the query execution time overhead of lineage annotations is remarkably low in \emph{Perm}.

\vspace{0.6em}
\begin{tabular}{l|rrr}
  \toprule
  system & plain & lineage (m=16)  & lineage (m=32) \\
  \midrule
  \LLinks & 3.1\,MB & 38.5\,MB & 73.7\,MB \\
  \emph{Perm} & 2.7\,MB & 89.4\,MB & 176.2\,MB \\
  \bottomrule
\end{tabular}\vspace{0.6em}

\emph{Perm} considerably outperforms \LLinks when it comes to lineage computation.
Their performance on plain queries is similar, which comes at a bit of a surprise.
We expected \LLinks to be a worse database client than the native {\tt psql} client, even for flat queries.
This can partly be explained by experiment setup.
We had database clients and servers run on the same machine to avoid network issues.
However, this reduces the amount of memory available for caching, especially since \LLinks uses so much memory to nearly run out on some queries.
This means a lot of time is spent by the database system waiting for disk seeks and postprocessing time is low by comparison.
Except for the largest queries, postprocessing by \LLinks is typically well below 1 second.

We take away three things:
(1) A different experimental setup could alleviate memory pressure and cache behavior and bring out processing times.
(2) We could change \Links to emit queries that use \emph{Perm's} built-in provenance features when possible.
(3) Most interesting would be to look at different ways to rewrite \LLinks queries.
Currently, we use \Links's nested query capabilities which allow a fairly naive translation.
\emph{Perm} exploits the fact that lineage is bounded by the structure of the query, adding columns instead of nested data.
Perhaps we could do something similar in \LLinks.

\section{Related Work}
\label{sec:related}

\citet{ICDT2001BunemanKT} gave the first definition of
where-provenance in the context of a semistructured data model.  The
DBNotes system of \citet{VLDB2005BhagwatCTV} supported
where-provenance via SQL query extensions.  DBNotes provides several
kinds of where-provenance in conjunctive SQL queries, implemented by
translating SQL queries to one or more provenance-propagating queries.
\citet{TODS2008BunemanCV} proposed a where-provenance model for nested
relational calculus queries and updates, and proved expressiveness
results.  They observed that where-provenance could be implemented by
translating and normalizing queries but did not implement this idea;
our approach to where-provenance in \WLinks is directly inspired by
that idea and is (to the best of our knowledge) the first
implementation of it.  One important difference is that we
\emph{explicitly} manage where-provenance via the $\Prov$ type, and
allow the programmer to decide whether to track provenance for some,
all or no fields.  Our approach also allows inspecting and comparing
the provenance annotations, which \citet{TODS2008BunemanCV} did not
allow; nevertheless, our type system prevents the programmer from
forging or unintentionally discarding provenance.  On the other hand,
our approach requires manual \lstinline!data! and \lstinline!prov!
annotations because it distinguishes between raw data and
provenance-annotated data.

\LLinks is inspired by prior work on lineage~\cite{TODS2000CuiWW} and
why-provenance~\cite{ICDT2001BunemanKT}.  There have been several
implementations of lineage and why-provenance.  Cui and Widom
implemented lineage in a prototype data warehousing system called
WHIPS.  The Trio system of \citet{VLDBJ2008BenjellounSHTW} also
supported lineage and used it for evaluating probabilistic queries;
lineage was implemented by defining customized versions of database
operations via user-defined functions, which are difficult for
database systems to optimize. \citet{ICDE2009GlavicA} introduced the
Perm system, which translated ordinary queries to queries that compute
their own lineage; they handled a larger sublanguage of SQL than
previous systems such as Trio, and subsequently
\citet{EDBT2009GlavicA} extended this approach to handle queries with
nested subqueries (e.g.\ SQL's EXISTS, ALL or ANY operations).  They
implemented these rewriting algorithms inside the database system and
showed performance improvements of up to 30 times relative to Trio.
In another line of work, Corcoran et al.~\cite{SIGMOD2009CorcoranSH}
and Swamy et al.~\cite{ICFP2009SwamyHB} developed SELinks, a variant of
Links with sophisticated support for security policies, including a
form of provenance tracking implemented using database extensions and
type-based coercions. Our approach instead shows that it is feasible to
perform this rewriting outside the database system and leverage the
standard SQL interface and underlying query optimization of relational
databases.

Both \WLinks and \LLinks rely on the conservativity and query
normalization results that underlie \Links's implementation of
language-integrated query, particularly Cooper's
work~(\citeyear{DBPL2009Cooper}) extending conservativity to queries
involving higher-order functions, and previous work by
\citet{SIGMOD2014CheneyLW} on ``query shredding'', that is, evaluating
queries with nested results efficiently by translation to equivalent
flat queries.  There are alternative solutions to this problem that
support larger subsets of SQL, such as grouping and aggregation, which
are not currently supported by Links.  There are other approaches to
nested data or grouping and aggregation, such as Grust et al.'s
\emph{loop-lifting} (\cite{PVLDB2010GrustRS}) and more recent work on
\emph{query flattening}~\cite{SIGMOD2015UlrichG} in the Database
Supported Haskell (DSH) library, or Suzuki et al.'s
Que$\Lambda$~\cite{suzuki16pepm}, and it would be interesting to
evaluate the performance of these techniques on provenance queries, or
to extend Links's query support to grouping and aggregation.

Other authors, starting with \citet{green07pods}, have proposed
provenance models based on annotations drawn from algebraic structures
such as semirings.  While initially restricted to conjunctive queries,
the semiring provenance model has subsequently been extended to handle
negation and aggregation operations~\cite{PODS2011AmsterdamerDT}.
\citet{SIGMOD2010KarvounarakisIT} developed ProQL, an implementation
of the semiring model in a relational database via SQL query
extensions.
\citet{Festschrift2013GlavicMG} present further details of the Perm
approach described above, show that semiring provenance can be
extracted from Perm's provenance model, and also describe a row-level
form of where-provenance.
  It is not yet clear how to support other instances of the semiring
  model via query rewriting in \Links.

\WLinks and \LLinks are currently separate extensions, and cannot be
used simultaneously, so another natural area for investigation is
supporting multiple provenance models at the same time.  We are
currently investigating this; one possible difficulty may be the need
to combine multiple type translations.  We intend to explore this
further (and consider alternative models).  \citet{cheney14ppdp} presented
a general form of provenance for nested relational calculus based on
execution traces, and showed how such traces can be used to provide
``slices'' that explain specific results.  While this model appears to
generalize all of the aforementioned approaches, it appears nontrivial
to implement by translation to relational queries, because it is not
obvious how to represent the traces in this approach in a relational
data model. (\citet{DDFP2013GiorgidzeGUW} show how to support
\emph{nonrecursive} algebraic data types in queries, but the trace
datatype is recursive.)  This would be a challenging area for future
work.

Our translation for lineage is similar in some respects to the
doubling translation used in \citet{PEPM2014CheneyLRW} to compile a
simplified form of \Links to a F\#-like core language.  Both
translations introduce space overhead and overhead for normal function
calls due to pair projections.  Developing a more efficient
alternative translation (perhaps in combination with a more efficient
and more complete compilation strategy) is an interesting topic for
future work.

As in most work on provenance, we have focused on explaining
questionable results in terms of the source data, and we assume that
the query itself is correct and not the source of the problem.  It
would also be interesting to consider a different problem where the
query (or other parts of the program) might have errors, and the
question is to identify which parts of the query or program could have
contributed to erroneous data.  This would require a combination of
program slicing~\cite{perera12icfp} and query
slicing~\cite{cheney14ppdp} techniques.

\section{Conclusions}

This article makes several contributions regarding integrating
provenance management with programming languages.  First, we present
language extensions to the \Links web programming language that
accommodate where-provenance (\WLinks) and lineage (\LLinks), give their
semantics, and establish basic provenance correctness properties.
Second, we show how to implement both extensions by translation back
to plain \Links, relying on \Links's existing sophisticated support
for language-integrated query, normalization and nested queries.  

Our approach shows that it is feasible to implement provenance by
rewriting queries \emph{outside} the database system, so that a
standard database management system can be used.  By building on the
well-developed theory of query normalization that underlies \Links's
approach to language-integrated query, our translations remain
relatively simple, while still being translated to SQL queries that
are executed efficiently on the database.  To the best of our
knowledge, our approach is the first efficient implementation of
provenance for \emph{nested} query results or for queries that can
employ \emph{first-class functions}; at any rate, SQL does not provide
either feature.  Our results show that provenance for database queries can be
implemented efficiently and safely at the language-level. This is a
promising first step towards systematic programming language support
for provenance.

\Links is a research prototype language, but the underlying ideas of
our approach could be applied to other systems that support
comprehension-based language-integrated query, such as F\# and
Database Supported Haskell.  There are a number of possible next
steps, including extending \Links's language-integrated query
capabilities to support richer queries and more forms of provenance.
Another area for future work is establishing the correctness of the
provenance translations.  We believe it would be better to develop a
general translation that abstracts the two given in this article, and
prove its correctness once and for all.  Finally, we have placed some
restrictions on the correctness properties for \WLinks and \LLinks:
specifically, we have  not considered the impact of updates on
provenance correctness, and we have restricted attention to monotonic
queries for \LLinks.  Lifting restrictions in a satisfying way is also
an intriguing direction for future work.

\paragraph{Acknowledgments}
We would like to James McKinna and several anonymous reviewers for
comments and helpful suggestions on this work.  This work was
supported by EU FP7 project DIACHRON (grant no. 601043) and by a
Google Research Award.

\bibliographystyle{abbrvnat}
\bibliography{bib}

\newpage
\appendix

  \section{Notation}
  \label{app:notation}
  \noindent\begin{tabular}{|c|c|p{7cm}|}
             \hline
             Notation & Sec. & Meaning\\
             \hline
             $\Sigma,M ⟶ \Sigma',M'$ & 3 & Database state $\Sigma$ and expression
                                           $M$ evaluate in one step to $\Sigma'$ and
                                           $M'$\\
             $\Sigma,M ⟶^* \Sigma',M'$ & 3 & Reflexive, transitive closure of $⟶$\\
             $A :: \QueryType$ & 3 & Type $A$ is allowed as a query result type\\
             $R :: \BaseRow$ & 3 & Row $R$ contains only fields of base types\\
             $\Gamma \vdash M : A$ & 3 & In type context $\Gamma$, expression $M$
                                         has type $A$\\
             $\Gamma \vdash S : \ProvSpec(R)$ & 4.1 & In type context $\Gamma$,
                                                      specification $S$ is a valid
                                                      provenance specification matching $R$\\
             $\erase{A}$ & 4.1 & Erasure of  $A$, replacing occurrences of
                                 $\Prov(O)$ with $O$\\
             $R \triangleright S$ &  4.1 & Augment row $R$ with provenance
                                           specification $S$\\
             $\mathit{cso}_\Sigma(M)$ & 4.1 & Set of colored subobjects of expression $M$,
                                              with respect to database state
                                              $\Sigma$\\
             $𝔏⟦A⟧$ & 4.2 & Lineage type translation of type $A$\\
             $\hat{\Sigma},M \linarr \hat{\Sigma}',M'$ & 4.2 & Lineage-enabled
                                                               evaluation\\
             $\|M\|$ & 4.2 & Collection of all lineage annotations from $M$\\
             $M|_b$ & 4.2 & Restriction of $M$  to collection elements
                            whose lineage is contained in $b$\\
             $V\sqsubseteq V'$ & 4.2 & $V$ is obtainable from $V'$ by deleting some
                                       list elements\\
             $𝔚⟦A⟧$ & 5.1 & Where-provenance type translation\\
             $𝔚⟦M ⟧$ & 5.1 & Where-provenance expression translation\\
             $R \triangleright^n_x S$ & 5.1 & A row expression constructing initial
                                              provenance from a row of type $R$ with
                                              table name $n$ and variable $x$ according
                                              to provenance specification $S$\\
             $𝔇⟦A⟧$ & 5.2 & Doubling translation of type $A$\\
             $𝔇⟦M⟧$ & 5.2 & Doubling translation of expression $M$\\
             $𝔏⟦M⟧$ & 5.2 & Lineage translation of query expression $M$\\
             $𝔏^*⟦M⟧$ & 5.2 & Closing lineage translation of $M$\\
             $d2l⟦A⟧(M)$& 5.2 & Mapping from doubling translation to lineage translation\\
             \hline
           \end{tabular}

\section{Proofs}

\subsection{Proof of Theorem~\ref{thm:where-correctness}}\label{sec:where-correctness-proof}
The statement of the theorem was:
  \[
    \Sigma, M \longrightarrow \Sigma, N \Rightarrow \cso(N) \subseteq \cso(M) 
  \]
  where $M$ and $N$ are \WLinks terms, and $\Sigma$ is a context that provides annotated table rows.
  \begin{proof}
    The proof is by induction on $⟶$.
    \begin{itemize}
  
  \item Case $(\kw{fun}\,f(x_i)\,M)(V_i) ⟶ M[f \coloneqq \kw{fun}\,f(x_i)\,M, x_i \coloneqq V_i]$:
  \begin{align*}
    \cso(M[f \coloneqq \kw{fun}\,f(x_i)\,M, x_i \coloneqq V_i])
    & \subseteq \cso(M) \cup \cso(\kw{fun}\,f(x_i)\,M) \cup \bigcup_{i=0}^n \cso(V_i) \\
    & = \cso(\kw{fun}\,f(x_i)\,M) \cup \bigcup_{i=0}^n \cso(V_i) \\
    & = \cso\left((\kw{fun}\,f(x_i)\,M)(V_i)\right)
  \end{align*}

  \item Case $\kw{var}\, x = V; M ⟶ M[x \coloneqq V]$:
  \begin{align*}
    \cso(M[x \coloneqq V])
    & \subseteq \cso(M) \cup \cso(V)  = \cso(\kw{var}\, x = V; M) 
  \end{align*}

  \item Case $(l_i = V_i)_{i=1}^n.l_k ⟶ V_k$ where $1 \leq k \leq n$:
  \begin{align*}
    \cso(V_k)
    & \subseteq \bigcup_{i=1}^n\cso(V_i) \\
    & = \cso((l_i = V_i)_{i=1}^n) \\
    & = \cso((l_i = V_i)_{i=1}^n.l_k)
  \end{align*}

  \item Case $\kw{if}\,(\kw{true})\,M\,\kw{else}\,N ⟶ M$:
  \begin{align*}
    \cso(M)
    & \subseteq \cso(M) \cup \cso(N) \\
    & = \cso(\kw{if}\,(\kw{true})\,M\,\kw{else}\,N) \\
  \end{align*}

  \item Case $\kw{if}\,(\kw{false})\,M\,\kw{else}\,N ⟶ N$:
  \begin{align*}
    \cso(N)
    & \subseteq \cso(M) \cup \cso(N) \\
    & = \cso(\kw{if}\,(\kw{false})\,M\,\kw{else}\,N) \\
  \end{align*}

  \item Case $\kw{query}\, M ⟶ M$: $\cso(M) = \cso(\kw{query}\, M)$.

  \item Case $\kw{table}\, n ⟶ \Sigma(n)$: $\cso(\Sigma(n) = \cso(\kw{table}\, n)$.

  \item Case $\kw{empty}(\texttt{[]}) ⟶ \kw{true}$:
  \begin{align*}
    \cso(\kw{true})
    & = \emptyset  = \cso(\kw{empty}(\texttt{[]}))
  \end{align*}

  \item Case $\kw{empty}(V) ⟶ \kw{false}$, where $V \neq \texttt{[]}$:
  \begin{align*}
    \cso(\kw{false})
    & = \emptyset  \subseteq \cso(V)  = \cso(\kw{empty}(V))
  \end{align*}

  \item Case $\kw{for}\,(x\,\texttt{<-}\,\texttt{[]})\,M ⟶ \texttt{[]}$:
  \begin{align*}
    \cso(\texttt{[]})
    & = \emptyset  \subseteq \cso(\kw{for}\,(x\,\texttt{<-}\,\texttt{[]})\,M)
  \end{align*}

  \item Case $\kw{for}\,(x\,\texttt{<-}\,\texttt{[}V\texttt{]})\,M ⟶ M[x\coloneqq V]$:
  \begin{align*}
    \cso(M[x\coloneqq V])
    & \subseteq \cso(M) \cup \cso(V) \\
    & = \cso(\kw{for}\,(x\,\texttt{<-}\,\texttt{[}V\texttt{]})\,M)
  \end{align*}

  \item Case $\kw{for}\,(x\,\texttt{<-}\,V \concat W)\,M ⟶ (\kw{for}\,(x\,\texttt{<-}\, V)\,M) \concat (\kw{for}\,(x\,\texttt{<-}\, W)\,M)$:
  \begin{align*}
    \cso(\kw{for}\,(x\,\texttt{<-}\,V \concat W)\,M)
    & = \cso(V \concat W) \cup \cso(M) \\
    & = \cso(V) \cup \cso(W) \cup \cso(M) \\
    & = \cso((\kw{for}\,(x\,\texttt{<-}\, V)\,M) \concat (\kw{for}\,(x\,\texttt{<-}\, W)\,M))
  \end{align*}

  \item Case $\kw{for}\,(x\,\texttt{<--}\, V)\, M ⟶ \kw{for}\,(x\,\texttt{<-}\, V)\, M$:
  \begin{align*}
    \cso(\kw{for}\,(x\,\texttt{<-}\, V)\, M)
    & = \cso(V) \cup \cso(M) \\
    & = \cso(\kw{for}\,(x\,\texttt{<--}\, V)\, M)
  \end{align*}

  \item Case $M ⟶ M' \Rightarrow \mathcal{E}[M] ⟶ \mathcal{E}[M']$ (evaluation step inside a context):
  \begin{align*}
    \cso(\mathcal{E}[M'])
    & = \cso(\mathcal{E}) \cup \cso(M') & \text{Lemma~\ref{lemma:cso-evaluation-context}} \\
    & \subseteq \cso(\mathcal{E}) \cup \cso(M) & \text{IH} \\
    & = \cso(\mathcal{E}[M]) & \text{Lemma~\ref{lemma:cso-evaluation-context}}
  \end{align*}
\end{itemize}
  \end{proof}

  \subsection{Full definitions of auxiliary functions for lineage annotation extraction and restriction}\label{def:lineage-aux-full}

  The interesting cases can be found in Figure~\ref{fig:llinks-supporting-definitions}.
  
  We extend $\| \cdot \|$, the lineage annotation collection function, by recursively collecting annotations.
  
  \begin{align*}
    \| \texttt{[} M \texttt{]}^a \| & = a \cup \| M \| \\
    \| \texttt{[]} \| & = \emptyset \\
    \| M \concat N \| & = \| M \| \cup \| N \| \\
    \| M^{\cup b}\| & = b \cup \| M \| \\
    \| \kw{table}\, t \| & = \| \hat{\Sigma}(t) \| \\
    \| \kw{var}\,x = M; N \| & = \| M \| \cup \| N \| \\
    \| c \| & = c \\
    \| (l_i = M_i)_{i=1}^n \| &= \bigcup_{i=1}^n \| M_i \|\\
    \| M.l \| &= \|M\| \\
    \| \kw{fun}\,f(x_i| _{i = 1}^n)\,M \| &= \kw{fun}\,f(x_i| _{i = 1}^n)\,\|M \| \\
    \| \kw{if}\,(L)\,M\,\kw{else}\,N \| &= \|L\| \cup \|M\| \cup \|N\|\\
    \|\kw{query}\,M\| &= \|M\|\\
    \|\kw{for}\,(x\,\texttt{<-}\,M)\,N \| &= \|M\| \cup \|N\|\\
    \|\kw{for}\,(x\,\texttt{<--}\,M)\, N \| &= \|M\| \cup \|N\| \\
%
%
  \end{align*}

  We extend $\cdot |_b$, the erasure function, by recursively erasing.
  
  \begin{align*}
    \texttt{[} M \texttt{]}^a|_b & = \begin{cases}
      \texttt{[} M|_b \texttt{]}^a & \text{if } a \subseteq b \\
      \texttt{[]} & \text{otherwise}
    \end{cases} \\
    \texttt{[]}|_b & = \texttt{[]} \\
    (M \concat N)|_b & = M|_b \concat N|_b \\
    M^{\cup a}|_b & = \begin{cases}
      (M|_b)^{\cup a} & \text{if } a \subseteq b\\
      \texttt{[]} & \text{otherwise}
    \end{cases}\\
    \kw{table}\, t|_b & = \kw{table}\, t \\
    (\kw{var}\,x = M; N)|_b & = \kw{var}\,x = M|_b; N|_b \\
    c |_b & = c \\
    (l_i = M_i)_{i=1}^n |_b &= (l_i = M_i|_b)_{i=1}^n \\
    M.l |_b &= (M|_b).l \\
    (\kw{fun}\,f(x_i| _{i = 1}^n)\,M) |_b &= \kw{fun}\,f(x_i| _{i = 1}^n)\, (M |_b) \\
    (\kw{if}\,(L)\,M\,\kw{else}\,N) |_b &= \kw{if}\,(L|_b)\,M|_b\,\kw{else}\,N|_b \\
    (\kw{query}\,M)|_b &= \kw{query}\,(M|_b)\\
    (\kw{for}\,(x\,\texttt{<-}\,M)\,N) |_b &= \kw{for}\,(x\,\texttt{<-}\,M|_b)\,N |_b \\
    (\kw{for}\,(x\,\texttt{<--}\,M)\, N) |_b &= \kw{for}\,(x\,\texttt{<--}\,M|_b)\, N |_b \\
  \end{align*}

  \subsection{Proof of Theorem~\ref{thm:where-type-preservation}}\label{thm:where-type-preservation-proof}
Recall the statement of the theorem:
  \begin{enumerate}
  \item For every \WLinks context $\Gamma$, term $M$, and type $A$, if $
    \Gamma \vdash_\WLinks M : A$ then $ 𝔚⟦\Gamma⟧ \vdash_\Links
    𝔚⟦M⟧ : 𝔚⟦A⟧$.
\item For every \WLinks context $\Gamma$, provenance specification
  $S$, row $R$ and subrow $R'$ such that $R'\triangleright^n_x S$ is defined, if $Γ ⊢ S : \ProvSpec(R)$ then 
$𝔚⟦\Gamma⟧,x{:}\rcd{R}   ⊢ \rcd{R' \triangleright^n_x  S} :  𝔚⟦\rcd{R' \triangleright S} ⟧ $.
  \end{enumerate}
\begin{proof}
  Proof is by induction on the structure of \WLinks derivations.  Most
  cases for the first part are immediate; we show some representative examples.
  \begin{itemize}
  \item If the derivation is of the form:
\[\infer [Data]
  { Γ ⊢ M : \Prov(A)}
  { Γ ⊢ \kw{data}~M : A }
\]
then by induction we have $𝔚⟦\Gamma⟧  ⊢ 𝔚⟦ M⟧ : 𝔚⟦ \Prov(A)⟧$, and can
conclude:
\[\infer
{ 𝔚⟦\Gamma⟧  ⊢ 𝔚⟦ M⟧ : (\mathsf{data}:  𝔚⟦A⟧,
                   \mathsf{prov}:(\StringTy,\StringTy,\IntTy) ) }
{ 𝔚⟦\Gamma⟧  ⊢ 𝔚⟦ M⟧.\mathsf{data} :  𝔚⟦A⟧}
\]
 \item If the derivation is of the form:
\[\infer [Data]
  { Γ ⊢ M : \Prov(A)}
  { Γ ⊢ \kw{prov}~M : \rcd{\StringTy, \StringTy, \IntTy} }
\]
then by induction we have $𝔚⟦\Gamma⟧  ⊢ 𝔚⟦ M⟧ : 𝔚⟦ \Prov(A)⟧$, and can
conclude:
\[\infer
{ 𝔚⟦\Gamma⟧  ⊢ 𝔚⟦ M⟧ : (\mathsf{data}:  𝔚⟦A⟧,
                   \mathsf{prov}:(\StringTy,\StringTy,\IntTy) ) }
{ 𝔚⟦\Gamma⟧  ⊢ 𝔚⟦ M⟧.\mathsf{prov} : \rcd{\StringTy, \StringTy, \IntTy} }
\]
\item If the derivation is of the form:
\[\infer[Table]
 { R :: \BaseRow\\
Γ ⊢ S : \ProvSpec(R)
}
 { Γ ⊢ \kw{table}\ n\ \kw{with}\ \rcd{ R }\ \kw{where}\ S : \kw{table}\rcd{R \triangleright S}}
\]
Then since $\|R \triangleright S\| = R$ (Lemma~\ref{lem:where-helper}) we can conclude:
\[\infer
 { 
}
 { 𝔚⟦\Gamma⟧   ⊢ \kw{table}~n~\kw{with} ~(R) :  \kw{table}(\|R \triangleright S\|)
}
\]
and by the second induction hypothesis,
\[\small\infer*
 { 
\infer*
{
\infer*
{R :: \BaseRow}
{𝔚⟦\Gamma⟧   ⊢ 
    \kw{table}~n~\kw{with} ~(R) :
    \kw{table}(R)}
~
\infer*
{𝔚⟦\Gamma⟧,x{:}\rcd{R}   ⊢ \rcd{R \triangleright^n_x
        S} :  
    𝔚⟦\rcd{R \triangleright S} ⟧}
{𝔚⟦\Gamma⟧,x{:}\rcd{R}   ⊢ \coll{\rcd{R \triangleright^n_x
        S}} :  
    \coll{𝔚⟦\rcd{R \triangleright S} ⟧}}
}
{
𝔚⟦\Gamma⟧   ⊢ \kw{for} (x \drarr
    \kw{table}~n~\kw{with} ~(R)) \coll{\rcd{R \triangleright^n_x
        S}} :  
    \coll{𝔚⟦\rcd{R \triangleright S} ⟧}
}
}
 { 𝔚⟦\Gamma⟧   ⊢ \kw{ fun} () \{ \kw{for} (x \drarr
    \kw{table}~n~\kw{with} ~(R)) \coll{\rcd{R \triangleright^n_x
        S}}\} :  () \arr
    \coll{𝔚⟦\rcd{R \triangleright S} ⟧}
}
\]
\item If the derivation is of the form
\[
\infer[For-Table]
  {Γ ⊢ L : \kw{table}(R)\\
   Γ, x : \rcd{R}  ⊢ M : \coll{B}}
  {\Gamma \vdash \kw{for}\ \rcd{ x\  \drarr\ L }\ M : \coll{B}}
\]
then by induction we have $𝔚⟦  Γ⟧ ⊢ 𝔚⟦  L⟧ : (\kw{table}(\|R\|),()
\arr \coll{𝔚⟦\rcd{R} ⟧} )$, so we can proceed as follows:
\[\small
\infer*
{
  {
    \infer*
    {𝔚⟦ \Gamma⟧ \vdash 𝔚⟦ L⟧.2 : () \arr \coll{𝔚⟦\rcd{R} ⟧}}
    {𝔚⟦ \Gamma⟧ \vdash 𝔚⟦ L⟧.2() :\coll{𝔚⟦\rcd{R} ⟧}}
  }
  \\
  𝔚⟦ \Gamma⟧,x:𝔚⟦\rcd{R}⟧ \vdash 𝔚⟦ M⟧ :\coll{𝔚⟦ B⟧}
}
{𝔚⟦ \Gamma⟧ \vdash \kw{for}\ \rcd{ x\  \rarr\ 𝔚⟦ L⟧.2() }\ 𝔚⟦ M⟧
  :\coll{𝔚⟦ B⟧}}
\]
\item If the derivation is of the form:
\[
\infer[Delete]
{Γ ⊢ L : \kw{table}(R)\\
Γ,x:\rcd{\|R\|} ⊢ M : \BoolTy}
{Γ ⊢ \kw{delete}~(x \drarr L)~\kw{where}~{M} : \rcd{}}
\]
then by induction we have $𝔚⟦Γ⟧ ⊢ 𝔚⟦ L⟧ : 𝔚⟦ \kw{table}(R)⟧$ and $ 𝔚⟦
Γ⟧,x: 𝔚⟦ \rcd{\|R\|}⟧ ⊢  𝔚⟦  M⟧ : \BoolTy$.
\[\small
\infer*
{
\infer*
{
𝔚⟦Γ⟧ ⊢ 𝔚⟦ L⟧ : (\kw{table}(\|R\|),\rcd{} \arr \coll{\rcd{R}})
}
{
𝔚⟦Γ⟧ ⊢ 𝔚⟦ L⟧.1 : \kw{table}(\|R\|)
}
\\
𝔚⟦Γ⟧,x:\rcd{\|R\|} ⊢ 𝔚⟦ M⟧ : \BoolTy
}
{
𝔚⟦Γ⟧ ⊢ \kw{delete}~(x \drarr  𝔚⟦ L⟧.1)~\kw{where}~{ 𝔚⟦ M⟧} : \rcd{}
}
\]
\end{itemize}

For the second part, we proceed by induction on the structure of the
derivation of $Γ ⊢ S : \ProvSpec(R)$.  We show one representative
case, for derivations of the form
\[\infer{
Γ ⊢ S : \ProvSpec(R)\\
Γ ⊢ M : \rcd{ R} ~\texttt{\small ->} ~\rcd{\StringTy, \StringTy ,\IntTy}
}{
Γ ⊢ S, l~\kw{prov}~{M} : \ProvSpec(R)
}
\]
In this case, by induction we have that $𝔚⟦\Gamma⟧,x{:}\rcd{R}   ⊢
\rcd{R' \triangleright^n_x  S} :  𝔚⟦\rcd{R' \triangleright S} ⟧ $
holds for any subrow $R'$ of $R$, and by the first induction
hypothesis we also know that $𝔚⟦Γ⟧ ⊢ 𝔚⟦M⟧ : 𝔚⟦\rcd{ R}⟧ ~\arr
~\rcd{\StringTy, \StringTy ,\IntTy}$.  

Suppose $R',l:O \triangleright ^n_x S,l~\kw{prov}~M$.   Then we can conclude that 
$𝔚⟦\Gamma⟧,x{:}\rcd{R}   ⊢
\rcd{R',l:\Prov(O) \triangleright^n_x  S,l~\kw{prov}~M}:
𝔚⟦\rcd{R',l:O \triangleright S,l~\kw{prov}~O} ⟧ $
because $\rcd{R',l:O \triangleright^n_x  S,l~\kw{prov}~M} =
\rcd{R'\triangleright^n_x  S},l=\rcd{\kw{data}=x.l,\kw{prov}=
  𝔚⟦M⟧(x)}$ and $R',l:O \triangleright S,l~\kw{prov}~O = (R'
\triangleright S), l:\Prov(O)$.
\end{proof}

\subsection{Proof of Lemma~\ref{lem:lin-helpers}}\label{lem:lin-helpers-proof}
Recall the statement of the lemma:
\begin{enumerate}
\item If $A :: \QueryType$ then  $𝔇⟦A⟧ = 𝔇⟦𝔏⟦A⟧⟧$.
\item If $\Gamma \vdash M : 𝔇⟦A⟧$ then $\Gamma\vdash d2l(M) : 𝔏⟦A⟧$.
\end{enumerate}
\begin{proof}
  For part 1, the proof is by induction on the derivation of $A ::
  \QueryType$, and is straightforward since both $𝔇$ and $𝔏$ are the
  identity on types formed only from base types, records or collection
  types.  

For the second part, the proof is by induction on the
  structure of $A$ but each case is straightforward.  We show the
  interesting cases for function types and collection types:
  \begin{itemize}
  \item If $A = B_1 \arr B_2$ then we proceed as follows:
\[
\infer*
{\Gamma \vdash M : (\DD{B_1} \arr \DD{B_2}, \LL{B_1} \arr \LL{B_2})}
{\Gamma \vdash M.2 : \LL{B_1} \arr \LL{B_2}}
\]
which suffices since $\LL{B_1 \arr B_2}  = \LL{B_1}\arr\LL{B_2}$.
\item If $A = \coll{B}$ then we proceed as follows:
\[\begin{array}{ll}
\Gamma \vdash M : \coll{\DD{B}}
 & \text{assumption}\\
\Gamma,x:\DD{B} \vdash x:\DD{B} & \text{by rule}\\
\Gamma,x:\DD{B} \vdash  d2l\sem{B}(x) : \LL{B} & 
\text{by IH}\\
\Gamma,x:\DD{B} \vdash  \coll{}:
\coll{(\StringTy,\IntTy)}
& \text{by rule}\\
\Gamma,x:\DD{B} \vdash  \rcd{\kw{data}=
      d2l\sem{B}(x),\kw{prov}=\coll{}}: \Lin(\LL{B})
&\text{by rule}
\\
\Gamma,x:\DD{B} \vdash  \coll{\rcd{\kw{data}=
      d2l\sem{B}(x),\kw{prov}=\coll{}}} :  \coll{\Lin(\LL{B})}
&\text{by rule}\\
\Gamma \vdash \kw{for}~(x \rarr M)~\coll{\rcd{\kw{data}=
      d2l\sem{B}(x),\kw{prov}=\coll{}}} :   \coll{ \Lin(\LL{B})}
&\text{by rule}
\end{array}
\]
  \end{itemize}
\end{proof}

\subsection{Proof of Theorem~\ref{thm:lin-type-preservation}}\label{thm:lin-type-preservation-proof}
Recall the statement of the theorem:
  \begin{enumerate}
 \item $𝔏⟦\Gamma⟧ \vdash_\Links 𝔏⟦M⟧ : 𝔏⟦A⟧ $
 \item $𝔇⟦\Gamma⟧ \vdash_\Links 𝔏^*⟦M⟧ : 𝔏⟦A⟧$
  \item $𝔇⟦\Gamma⟧ \vdash_\Links 𝔇⟦M⟧ : 𝔇⟦A⟧$
  \end{enumerate}
\begin{proof}
For the first part, we show the details of the cases for singleton
lists and list comprehensions.  Table comprehensions are similar.
\begin{itemize}
\item If the derivation is of the form:
\[
\infer[List]
{ Γ ⊢ M : A}
{Γ ⊢ \coll{ M} : \coll{ A }}
\]
then we proceed as
follows:
\[
\begin{array}{ll}
𝔏⟦\Gamma⟧ ⊢ 𝔏⟦M⟧ : 𝔏⟦A⟧ 
& \text{by IH}\\
\LL{\Gamma} \vdash \coll{}:  \coll{\rcd{\StringTy,\IntTy}}
& \text{by rule}\\
\LL{\Gamma} \vdash \rcd{\kw{data}=\LL{M},\kw{prov}=\coll{}} :  \Lin( \LL{A} )
& \text{by rule}\\
\LL{\Gamma} \vdash \coll{\rcd{\kw{data}=\LL{M},\kw{prov}=\coll{}}}:   \coll{\Lin( \LL{A} )}
& \text{by rule}
\end{array}
\]
which suffices since $\LL{\coll{A}} =  \coll{\Lin{ \LL{A} }} = \coll{\rcd{\kw{data} : \LL{A} , \kw{prov}:
  \coll{\rcd{\StringTy,\IntTy}}}}$.
\item If the derivation is of the form:
\[
\infer[For-List]
  {Γ ⊢ L : \coll{ A }\\
   Γ, x : A ⊢ M : \coll{B}}
  {\Gamma \vdash \kw{for}\ \rcd{ x\  \texttt{\small<-}\ L }\
    M : \coll{B}}
\]
then we proceed as follows:
\[\small\begin{array}[b]{ll}
\LL{Γ} ⊢ \LL{L} : \coll{\Lin(\LL{A})} & 
\text{by IH}\\
\LL{Γ}, x : \LL{A} ⊢ \LL{M} : \coll{\Lin(\LL{B})}&
\text{by IH}\\
\LL{\Gamma},y:\Lin(\LL{A}) \vdash y.\kw{data} : \LL{A} & \text{by rule}\\
\LL{\Gamma},y:\Lin(\LL{A}) \vdash 
𝔏⟦M⟧[x \mapsto y.\mathsf{data}] : \Lin(\LL{A})
& \text{by substitution}\\
\LL{\Gamma},y:\Lin(\LL{A}) , z:\Lin(\LL{B}) \vdash 
 z.\mathsf{data} : \LL{B} 
& \text{by rule}\\
\LL{\Gamma},y:\Lin(\LL{A}) , z:\Lin(\LL{B}) \vdash 
 y.\mathsf{prov} : \coll{\rcd{\StringTy,\IntTy}}
 & \text{by rule}\\
\LL{\Gamma},y:\Lin(\LL{A}) , z:\Lin(\LL{B}) \vdash 
 z.\mathsf{prov}  : \coll{\rcd{\StringTy,\IntTy}}
& \text{by rule}\\
\LL{\Gamma},y:\Lin(\LL{A}) ,z:\Lin(\LL{B}) \vdash y.\mathsf{prov} \plusplus z.\mathsf{prov} :\coll{\rcd{\StringTy,\IntTy}}
& \text{by rule}\\
\LL{\Gamma},y:\Lin(\LL{A}) ,z:\Lin(\LL{B}) \vdash \\
\qquad (\mathsf{data} = z.\mathsf{data}, \mathsf{prov} =
  y.\mathsf{prov} \plusplus z.\mathsf{prov}):\Lin(\LL{B})
& \text{by rule}\\
\LL{\Gamma},y:\Lin(\LL{A}) ,z:\Lin(\LL{B}) \vdash \\
\qquad [(\mathsf{data} = z.\mathsf{data}, \mathsf{prov} =
  y.\mathsf{prov} \plusplus z.\mathsf{prov})] :\coll{\Lin(\LL{B})}
& \text{by rule}\\
\LL{\Gamma},y:\Lin(\LL{A}) \vdash \\
\qquad \begin{array}[t]{l}
 \kw{for}~(z \rarr 𝔏⟦M⟧[x \mapsto y.\mathsf{data}])\\
\quad [(\mathsf{data} = z.\mathsf{data}, \mathsf{prov} =
  y.\mathsf{prov} \plusplus z.\mathsf{prov})] :\coll{\Lin(\LL{B})}
\end{array} & \text{by rule}\\
\LL{\Gamma} \vdash 
\begin{array}[t]{l}\kw{for}~(y \rarr 𝔏⟦L⟧)\\
\quad \kw{for}~(z \rarr 𝔏⟦M⟧[x \mapsto y.\mathsf{data}])\\
\qquad [(\mathsf{data} = z.\mathsf{data}, \mathsf{prov} =
  y.\mathsf{prov} \plusplus z.\mathsf{prov})] : \coll{\Lin(\LL{B})}
\end{array} & \text{by rule}
\end{array}
\]
\end{itemize}

Finally, for the third part, we show the interesting cases
for functions, function calls, and $\kw{lineage}$.
\begin{itemize}
\item If the derivation is of the form:
\[
\infer[Fun]
{Γ , x : A⊢ M : B}
{Γ ⊢ \kw{fun }\rcd{ x} \braces{M}: A \arr B}
\]
then by induction we have $\DD{Γ} ,x : \DD{A} ⊢ \DD{M}
: \DD{B}$ and by part 2 we know that $\DD{\Gamma} \vdash \LLstar{\kw{fun }\rcd{ x } \braces{M}}
:  \LL{\rcd{A } \arr B}$.  We can proceed as follows:

\[\begin{array}{ll}
\DD{Γ} , x:\DD{A} ⊢ \DD{M}: \DD{B} 
&\text{by IH}\\
\DD{Γ} ⊢ \kw{fun }\rcd{ x}
  \braces{\DD{M}}: \DD{ A} \arr \DD{B}
&\text{by rule}\\
 \DD{\Gamma} \vdash \LLstar{\kw{fun }\rcd{ {x} } \braces{M}}
  :  \LL{A} \arr \LL{B}
& \text{by part 2}\\
\DD{Γ} ⊢ (\kw{fun }\rcd{x }
  \braces{\DD{M}},\LLstar{\kw{fun }\rcd{ x }
    \braces{M}}): \DD{A \arr B}
&\text{by rule}
\end{array}
\]
where the final step relies on the fact that $\DD{A \arr B} = (\DD{A}
\arr \DD{B}, \LL{A} \arr \LL{B})$.
\item If the derivation is of the form:
\[
\infer[App]
{Γ ⊢ M : A \arr B\\
Γ ⊢ N : A }
{Γ ⊢ M(N) : B}
\]
then we proceed as follows:
\[
\begin{array}{ll}
\DD{Γ} ⊢ \DD{M} : (\DD{A} \arr \DD{B},\LL{A} \arr\LL{B})
&\text{by IH}\\
\DD{Γ} ⊢ \DD{M}.1 : \DD{A} \arr \DD{B}
&\text{by rule}\\
\DD{Γ} ⊢ \DD{N} : \DD{A} 
& \text{by IH}\\
\DD{Γ} ⊢ \DD{M}.1(\DD{N}) : \DD{B}
&\text{by rule}
\end{array}
\]
where in the first step we use the fact that
$\DD{A \arr B} = (\DD{A} \arr \DD{B},\LL{A} \arr\LL{B})$.
\item If the derivation is of the form:
\[
\infer[Lineage]
{\Gamma \vdash M : \coll{A} \\
A :: \QueryType}
{\Gamma \vdash \kw{lineage}~\{M\} :  𝔏⟦\coll{A}⟧}
\]
then by part (2) we know that $\DD{\Gamma}\vdash \LLstar{M} :
\LL{\coll{A}}$, so we proceed as follows:
\[
\infer*
{\DD{\Gamma}\vdash \LLstar{M} : \coll{\LL{A}}
\\
\LL{A} :: \QueryType}
{\DD{\Gamma} \vdash \kw{query}~\{\LLstar{M}\} : \coll{\LL{A}}}
\]
which suffices since $\DD{\LL{A}} = \DD{A}$ by Lemma~\ref{lem:lin-helpers}(1).
\end{itemize}
\end{proof}

\section{Benchmark code}\label{appendix:benchmark-code}

This appendix contains the full listings for the where-provenance and
lineage benchmarks. Figures~\ref{fig:table-declarations}
and~\ref{fig:table-declarations-where} show the plain table
declarations and declarations with where-provenance, respectively.
These tables also include \lstinline!readonly! and
\lstinline!tablekeys! annotations which were suppressed in the main
body of the article;
the former indicates that a field is read-only and the latter lists
the subsets of the fields that uniquely determine the others.

Figure~\ref{fig:noprov-helpers} shows the helper functions used by the
plain versions of the queries, and Figure~\ref{fig:prov-helpers} shows
the variants of these functions adapted to work with
where-provenance.  Some of the functions, such as \lstinline!any!,
need no changes at all because they are polymorphic.
Figure~\ref{fig:someprov-queries} shows the versions of the queries
with some provenance (the someprov benchmarks).

Figures~\ref{fig:nolineage-queries1} and~\ref{fig:nolineage-queries2}
show the plain queries without lineage annotations; these also employ
abbreviations from Figure~\ref{fig:noprov-helpers}.

\begin{figure}
\begin{lstlisting}[language=LLinks,morekeywords={database,tablekeys,readonly,from}]
var db = database "links";

var departments =
  table "departments"
  with (oid: Int, name: String)
  where oid readonly
  tablekeys [["name"],["oid"]]
  from db;

var employees =
  table "employees"
  with (oid: Int, dept: String, name: String, salary : Int)
  where oid readonly
  tablekeys [["name"],["oid"]]
  from db;

var tasks =
  table "tasks"
  with (oid: Int, employee: String, task: String)
  where oid readonly
  tablekeys [["oid"]]
  from db;

var contacts =
  table "contacts"
  with (oid: Int, dept: String, name: String, "client": Bool)
  where oid readonly
  tablekeys [["name"], ["oid"]]
  from db;
\end{lstlisting}
\caption{Table declarations for lineage, nolin, and noprov queries.}
\label{fig:table-declarations}
\end{figure}

\begin{figure}
\begin{lstlisting}[language=WLinks,morekeywords={database,tablekeys,readonly,from}]
var departments =
  table "departments"
  with (oid: Int, name: String)
  where oid readonly, name prov default
  tablekeys [["name"]]
  from db;

var employees =
  table "employees"
  with (oid: Int, dept: String, name: String, salary : Int)
  where oid readonly, dept prov default,
        name prov default, salary prov default
  tablekeys [["name"]]
  from db;

var tasks =
  table "tasks"
  with (oid: Int, employee: String, task: String)
  where oid readonly, employee prov default, task prov default
  tablekeys [["oid"]]
  from db;

var contacts =
  table "contacts"
  with (oid: Int, dept: String, name: String, "client": Bool)
  where oid readonly, dept prov default,
        name prov default, "client" prov default
  tablekeys [["name"]]
  from db;
\end{lstlisting}
\caption{Table declarations for where-provenance queries (except noprov).}
\label{fig:table-declarations-where}
 \end{figure}

 \begin{figure}
\begin{lstlisting}[language=Links]
sig tasksOfEmp: ((name:String|_)) -> [String]
fun tasksOfEmp(e) {
  for (t <-- tasks) where (t.employee == e.name) [t.task]
}

sig contactsOfDept: ((name:String|_)) -> [("client":Bool,name:String)]
fun contactsOfDept(d) {
  for (c <-- contacts)
  where (d.name == c.dept)
    [("client" = c."client", name = c.name)]
}

sig employeesByTask: ((employee:String|_)) -> [(name:String,salary:Int,tasks:[String])]
fun employeesByTask(t) {
  for (e <-- employees)
    for (d <-- departments)
    where (e.name == t.employee && e.dept == d.name)
      [(name = e.name, salary = e.salary, tasks = tasksOfEmp(e))]
}

sig employeesOfDept: ((name:String|_)) -> [(name:String,salary:Int,tasks:[String])]
fun employeesOfDept(d) {
  for (e <-- employees)
  where (d.name == e.dept)
    [(name = e.name, salary = e.salary, tasks = tasksOfEmp(e))]
}

sig any : ([a],(a) -a-> Bool) -a-> Bool
fun any(xs,p) { not(empty(for (x <- xs) where (p(x)) [()])) }

sig all : ([a],(a) -a-> Bool) -a-> Bool
fun all(xs, p) { not(any(xs, fun (x) { not(p(x)) })) }

sig contains: ([a], a) -> Bool
fun contains(xs, u) { any(xs, fun (x) { x == u }) }

fun isPoor(x) { x.salary < 1000 }
fun isRich(x) { x.salary > 1000000 }

sig get: ([(name:a::Any|b)], ((name:a::Any|b)) -c-> d::Any)
      -c-> [(name:a::Any,tasks:d::Any)]
fun get(xs, f) {
  for (x <- xs)
    [(name = x.name, tasks = f(x))]
}

sig outliers: ([(salary:Int|a)]) -> [(salary:Int|a)]
fun outliers(xs) { filter(fun (x) { isRich(x) || isPoor(x) }, xs) }

sig clients: ([("client":Bool|a)]) -> [("client":Bool|a)]
fun clients(xs) { filter(fun (x) { x."client" }, xs) }
\end{lstlisting}
\caption{Helper functions noprov.}
\label{fig:noprov-helpers}
\end{figure}

\begin{figure}[!h]
\begin{lstlisting}[language=WLinks]
# the original (allprov) Q1
fun q_org() {
  for (d <-- departments)
    [(contacts = contactsOfDept(d),
      employees = employeesOfDept(d),
      name = d.name)]
}

sig tasksOfEmp: ((name:Prov(String)|_)) -> [Prov(String)]
fun tasksOfEmp(e) {
  for (t <-- tasks)
  where ((data t.employee) == data e.name)
    [t.task]
}

sig contactsOfDept: ((name:Prov(String)|_)) -> [("client":Prov(Bool),name:Prov(String))]
fun contactsOfDept(d) {
  for (c <-- contacts)
  where ((data d.name) == data c.dept)
    [("client" = c."client", name = c.name)]
}

sig employeesByTask: ((employee:Prov(String)|_))
      -> [(name:Prov(String),salary:Prov(Int),tasks:[Prov(String)])]
fun employeesByTask(t) {
  for (e <-- employees)
    for (d <-- departments)
    where ((data e.name) == (data t.employee)
           && (data e.dept) == (data d.name))
      [(name = e.name, salary = e.salary, tasks = tasksOfEmp(e))]
}

sig employeesOfDept: ((name:Prov(String)|_))
     -> [(name:Prov(String),salary:Prov(Int),tasks:[Prov(String)])]
fun employeesOfDept(d) {
  for (e <-- employees)
  where ((data d.name) == data e.dept)
    [(name = e.name, salary = e.salary, tasks = tasksOfEmp(e))]
}

fun get(xs, f) {
  for (x <- xs) [(name = x.name, tasks = f(x))]
}

sig outliers: ([(salary:Prov(Int)|a)]) -> [(salary:Prov(Int)|a)]
fun outliers(xs) { filter(fun (x) { isRich(x) || isPoor(x) }, xs) }

sig clients: ([("client":Prov(Bool)|a)]) -> [("client":Prov(Bool)|a)]
fun clients(xs) { filter(fun (x) { data x."client" }, xs) }
\end{lstlisting}
\caption{Helper functions allprov, someprov (use \kw{data} in some places).}
\label{fig:prov-helpers}
\end{figure}

\begin{figure}
\begin{lstlisting}[language=WLinks]
# Q1
sig q1 : () -> [(contacts: [("client": Bool, name: String)],
                 employees: [(name: Prov(String), salary: Prov(Int),
                              tasks: [Prov(String)])],
                 name: Prov(String))]
fun q1() {  for (d <-- departments)
              [(contacts = for (c <- contactsOfDept(d))
                              [("client" = data c."client", name = data c.name)],
                employees = employeesOfDept(d),
                name = d.name)] }

# Q2
sig q2 : () -> [(d: String, p: (String, String, Int))]
fun q2() { for (d <- q_org())
              where (all(d.employees, fun (e) {
                     contains(map(fun (x) { data x }, e.tasks), "abstract")
                     }))
              [(d = data d.name, p = prov d.name)] }

# Q3: employees with lists of tasks
sig q3 : () -> [(b: [Prov(String)], e: Prov(String))]
fun q3() { for (e <-- employees) [(b = tasksOfEmp(e), e = (e.name))] }

# Q4: departments with lists of employees
sig q4 : () -> [(dpt:Prov(String), emps:[(String, String, Int)])]
fun q4() { for (d <-- departments)
             [(dpt = d.name, emps = for (e <-- employees)
                                          where ((data d.name) == (data e.dept))
                                          [prov e.name])] }

# Q5: Tasks with employees and departments
fun dropProv(l) { map(fun (x) { data x }, l) }

sig q5: () -> [(a: Prov(String), b: [(name: String, salary: Int, tasks: [String])])]
fun q5() { for (t <-- tasks)
              [(a = t.task, b = for (x <- employeesByTask(t))
                                      [(name = data x.name,
                                        salary = data x.salary,
                                        tasks = dropProv(x.tasks))])]
}

# Q6 Drop prov on department.
sig q6: () -> [(department: String, people: [(name: Prov(String), tasks: [String])])]
fun q6() {  for (x <- q_org())
            [(department = data x.name,
              people = get(outliers(x.employees),
                           fun (y) { map(fun (z) { data z }, y.tasks) }) ++
                       get(clients(x.contacts),
                           fun (y) { ["buy"] }))] }
\end{lstlisting}
\caption{Queries someprov.}
\label{fig:someprov-queries}
\end{figure}








\begin{figure}
\begin{lstlisting}[language=Links]
# AQ6 : [(department: String, outliers: [(name: String, ...
for (d <- for (d <-- departments)
            [(employees = for (e <-- employees)
                              where (d.name == e.dept)
                              [(name = e.name, salary = e.salary)],
              name = d.name)])
  [(department = d.name, outliers = for (o <- d.employees)
                                             where (o.salary > 1000000 || o.salary < 1000)
                                             [o])]

# Q3 : [(b: [String]), e: String)]
for (e <-- employees) [(b = tasksOfEmp(e), e = e.name)]

# Q4 : [(dpt: String, emps: [String]))]
for (d <-- departments)
  [(dpt = d.name, emps = for (e <-- employees)
                                 where (d.name == e.dept)
                                 [(e.name)])]

# Q5 : [(a: String, b: [(name: String, salary: Int, ...
for (t <-- tasks) [(a = t.task, b = employeesByTask(t))]

# Q6N : [(department: String, people:[(name: String, ...
for (x <-- departments)
  [(department = x.name,
    people = (for (y <-- employees)
                 where (x.name == y.dept && (y.salary < 1000 || y.salary > 1000000))
                 [(name = y.name,  tasks = for (z <-- tasks)
                                                   where (z.employee == y.name)
                                                   [z.task])]) ++
               (for (y <-- contacts)
                where (x.name == y.dept && y."client")
                [(name = y.dept, tasks = ["buy"])]))]
\end{lstlisting}
\caption{Nolineage queries, part 1}
\label{fig:nolineage-queries1}
\end{figure}

\begin{figure}
\begin{lstlisting}
# Q7 : [(department: String, employee: (name: String, ...
for (d <-- departments)
for (e <-- employees)
where (d.name == e.dept && e.salary > 1000000 || e.salary < 1000)
  [(employee = (name = e.name, salary = e.salary), department = d.name)]

# QC4 : [(a: String, b: String, c: [(doer: String, ...
for (x <-- employees)
for (y <-- employees)
where (x.dept == y.dept && x.name <> y.name)
  [(a = x.name, b = y.name,
    c = (for (t <-- tasks)
         where (x.name == t.employee)
           [(doer = "a", task = t.task)]) ++
          (for (t <-- tasks)
           where (y.name == t.employee)
             [(doer = "b", task = t.task)]))]

# QF3 : [(String, String)]
for (e1 <-- employees)
for (e2 <-- employees)
where (e1.dept == e2.dept && e1.salary == e2.salary && e1.name <> e2.name)
  [(e1.name, e2.name)]

# QF4 : [String]
(for (t <-- tasks) where (t.task == "abstract")[t.employee]) 
  ++
(for (e <-- employees) where (e.salary > 50000) [e.name])
\end{lstlisting}
\caption{Nolineage queries, part 2}
\label{fig:nolineage-queries2}
\end{figure}

\FloatBarrier 
\subsection{Perm comparison}\label{appendix:benchmark-code-perm}

Table declarations and where-provenance queries in \WLinks.

\begin{lstlisting}[language=WLinks,morekeywords={database,tablekeys}]
var db = database "links";
var i_s_c_o_$n$_1 =
  table "i_s_c_o_$n$_1" 
  with (oid: Int, i: Int, s: String, cardinal: String, ordinal: String)
  where cardinal prov default tablekeys [["oid"], ["i"]] from db;
$\hdots$
query { 
  for (t_1 <-- i_s_c_o_$n$_1) $\hdots$ for (t_$m$ <-- i_s_c_o_$n$_$m$)
  where (mod(t_1.i, 100) < 5 && t_1.i == t_2.i && $\hdots$ && t_1.i = t_$m$.i)
    [(c1 = t_1.cardinal, c2 = t_2.cardinal, $\hdots$, c$m$ = t_$m$.cardinal)]
}
\end{lstlisting}

The \WLinks results with \wpr enabled look something like this with pretty printing of provenance-annotated values disabled (we can see the type \lstinline!Prov(a)! really desugars to the tuple type \lstinline{(!data:a, !prov:(String, String, Int))}):

\begin{lstlisting}[language=Links]
[(c1=(!data="one",!prov=("i_s_c_o_10000_1", "cardinal", 715924950)),
  c2=(!data="one",!prov=("i_s_c_o_10000_2", "cardinal", 715925958)), $\hdots$), $\hdots$]
\end{lstlisting}

\emph{Perm} uses arrays to collect annotations of equal rows.
In our query, all rows are different, so these are all singleton arrays.

{\small
\begin{tabular}{lll}
  \toprule
  c1           &               annot\_c1 & \dots \\
  \midrule
  two hundred sixty-seven&\{public.i\_s\_c\_o\_10000\_1\#cardinal\#114040340\}& \dots \\
  three hundred seventeen&\{public.i\_s\_c\_o\_10000\_1\#cardinal\#114040390\}& \dots \\
  \dots\\
  \bottomrule
\end{tabular}
}

\LLinks lineage queries and part of an example result.

\begin{lstlisting}[language=LLinks]
lineage {
  query {
    for (t_1 <-- i_s_c_o_$n$_1) $\hdots$ for (t_$m$ <-- i_s_c_o_$n$_$m$)
    where (mod(t_1.i, 100) < 5 && t_1.i == t_2.i && $\hdots$ && t_$(m-1)$.i == t_$m$.i)
      [(i=t_1.i, c = t_1.cardinal)]
}}
\end{lstlisting}

\begin{lstlisting}[language=LLinks]
[(data=(c="one", i=1),
  prov=[(row=715924950, %table%="i_s_c_o_1000_1"),
          (row=715925958, %table%="i_s_c_o_1000_2"),
          $\hdots$]),
 (data=(c="two", i=2),
  prov=[(row=715924951, %table%="i_s_c_o_1000_1"), $\hdots$]),
 $\hdots$]
\end{lstlisting}

The template for ``equivalent'' \emph{Perm} queries is shown below.
We use the \lstinline[morekeywords={PROVENANCE}]!PROVENANCE! keyword which enables Perm influence contribution semantics.

\begin{lstlisting}[language=SQL,escapechar={\#},morekeywords={PROVENANCE}]
SELECT PROVENANCE t_1.i, t_1.cardinal
FROM i_s_c_o_$n$_1 AS t_1, $\hdots$, i_s_c_o_$n$_$m$ AS t_$m$
WHERE t_1.i % 100 < 5 AND t_1.i = t_2.i AND $\hdots$ AND t_$(m-1)$.i = t_$m$.i
\end{lstlisting}

\end{document}